\documentclass[conference]{IEEEtran}
\IEEEoverridecommandlockouts

\usepackage{times,amsmath}
\usepackage{textcomp}

\usepackage{graphicx}
\usepackage{balance}  % for  \balance command ON LAST PAGE  (only there!)
\usepackage{amsmath}

\usepackage{subfigure}
\usepackage{graphicx}
\usepackage{booktabs}
\usepackage{amsmath,epsfig,latexsym,graphics,subfigure,latexsym}
\usepackage[normalem]{ulem}
\usepackage{url}
\usepackage{epsfig}
\usepackage{amsthm}

\usepackage{weiwAlgorithm}

\usepackage{adjustbox}

\usepackage{pstricks}
\usepackage{xspace}
\usepackage{float}
\usepackage{latexsym}
\usepackage{color}
\usepackage{colortbl}
\usepackage{xcolor}
\usepackage{epstopdf}
\usepackage{caption}
\usepackage{multirow}
\usepackage{subfigure}
\usepackage{dsfont}
\usepackage{hhline}
\usepackage{balance}
\usepackage[euler]{textgreek}
\usepackage{caption}
\usepackage{mathrsfs}
\usepackage{cancel}
\usepackage{float}
\definecolor{mygray}{gray}{.9}
\definecolor{dkgreen}{rgb}{0,0.6,0}
\definecolor{gray}{rgb}{0.5,0.5,0.5}
\definecolor{mauve}{rgb}{0.58,0,0.82}
\newcommand{\rev}[1]{\textcolor[rgb]{0,0,0}{#1}}

\newcommand{\rrev}[1]{\textcolor[rgb]{0,0,0}{#1}}

\newcommand{\kpc}{hop-constrained cycle cover\xspace}
\newcommand{\kpcs}{hop-constrained cycle covers\xspace}
\newcommand{\cvc}{hop-constrained cycle cover\xspace}

\newcommand{\kc}{hop-constrained cycle\xspace}
\newcommand{\kcs}{hop-constrained cycles\xspace}

\newcommand{\SetVline}{\SetAlgoVlined}%
\newcommand{\bt}{block technique\xspace}

\newcommand{\bft}{BFS-filter technique\xspace}

\newcommand{\adp}{BUR\xspace}
\newcommand{\adpm}{BUR+\xspace}

\newcommand{\darc}{\text{DARC-DV}\xspace}
\newcommand{\np}{$NP$-$hard$\xspace}

\newcommand{\hcst}{hc-$s$-$t$\xspace}

\newcommand{\tpd}{Top-Down\xspace}
\newcommand{\tbs}{TDB\xspace}
\newcommand{\tbk}{TDB+\xspace}
\newcommand{\tdb}{TDB++\xspace}

\renewcommand{\baselinestretch}{0.97}
\newcommand{\eat}[1]{#1}
\newcommand{\reat}[1]{}

\newcommand{\etal}{\emph{et al.}}
\newtheorem{example}{\textbf{Example}}

\newtheorem{theorem}{\textbf{Theorem}}
\newtheorem{lemma}{\textbf{Lemma}}

\newtheorem{definition}{\textbf{Definition}}
\def\BibTeX{{\rm B\kern-.05em{\sc i\kern-.025em b}\kern-.08em
    T\kern-.1667em\lower.7ex\hbox{E}\kern-.125emX}}
\begin{document}
\title{TDB: Breaking All Hop-Constrained Cycles in Billion-Scale Directed Graphs% and Their Applications}
\thanks{Xuemin Lin is the corresponding author.}}

% Single author syntax
\author{{You Peng$^{\dagger}$, Xuemin Lin$^{\ddagger}$$^*$, Michael Yu$^{\dagger}$, Wenjie Zhang$^{\dagger}$, Lu Qin$^{\S}$
} %
\vspace{1.6mm}\\
\fontsize{10}{10}
\selectfont\itshape
$^\dagger$The University Of New South Wales, \\
$^{\ddagger}$Antai College of Economics \& Management, Shanghai Jiao Tong University, Shanghai, China,\\
$^\S$QCIS, University of Technology, Sydney\\
\fontsize{9}{9} \selectfont\ttfamily\upshape
\{unswpy, xuemin.lin\}@gmail.com, \{mryu,zhangw\}@cse.unsw.edu.au, lu.qin@uts.edu.au \\}

%\author{
%Anonymous Authors
%%First Author$^1$\footnote{Contact Author}\and
%%Second Author$^2$\and
%%Third Author$^{2,3}$\And
%%Fourth Author$^4$\\
%%\affiliations
%%$^1$First Affiliation\\
%%$^2$Second Affiliation\\
%%$^3$Third Affiliation\\
%%$^4$Fourth Affiliation\\
%%\emails
%%\{first, second\}@example.com,
%%third@other.example.com,
%%fourth@example.com
%}

%\author{
%You Peng$^1$\footnote{Contact Author}\and
%Ying Zhang$^2$\and
%Xuemin Lin$^{1}$\And
%Lu Qin$^2$\and
%Wenjie Zhang$^1$
%\affiliations
%$^1$The University of New South Wales\\
%$^2$The University of Technology Sydney
%\emails
%unswpy@gmail.com,
%\{ying.zhang,lu.qin\}@uts.edu.au,
%\{lxue,zhangw\}@cse.unsw.edu.au,
%
%}

% Multiple author syntax (remove the single-author syntax above and the \iffalse ... \fi here)
% Check the ijcai20-multiauthor.tex file for detailed instructions
%\iffalse
%\author{
%You Peng$^1$
%\and
%Ying Zhang$^2$\and
%Xuemin Lin$^{1}$\And
%Lu Qin$^2$\and
%Wenjie Zhang$^1$
%\affiliations
%$^1$The University of New South Wales\\
%$^2$The University of Technology Sydney
%\emails
%unswpy@gmail.com,
%ying.zhang@uts.edu.au,
%lxue@cse.unsw.edu.au,
%lu.qin@uts.edu.au,
%zhangw@cse.unsw.edu.au
%}
%\fi

%\begin{document}

\maketitle
\begin{abstract}
\rev{The Feedback vertex set with the minimum size is one of Karp's 21 NP-complete problems targeted at breaking all the cycles in a graph.} This problem is applicable to a broad variety of domains, including E-commerce networks, database systems, and program analysis. In reality, users are frequently most concerned with the hop-constrained cycles (i.e., cycles with a limited number of hops). For instance, in the E-commerce networks, the fraud detection team would discard cycles with a high number of hops since they are less relevant and grow exponentially in size. Thus, it is quite reasonable to investigate the feedback vertex set problem in the context of hop-constrained cycles, namely \textit{\kpc} problem. It is concerned with determining a set of vertices that covers all hop-constrained cycles in a given directed graph.
%Nevertheless, constraints for these cycles are naturally imposed in practice.
%Thus, breaking all the constrained cycles are also essential in many scenarios, e.g., E-commerce networks, database systems, and program analysis.
%Motivated by this, we investigate the classical feedback vertex set problem with constrained cycles named the \kpc. It aims to find a set of vertices 
%to cover all constrained cycles in a given directed graph. 
A common method to solve this is to use \rrev{a} \textit{bottom-up} algorithm, where \rev{it iteratively selects cover vertices into the result set}. \rev{Based on this paradigm, the existing works mainly focus on the vertices orders and several heuristic strategies. In this paper, a totally opposite cover process \textit{top-down} is proposed and bounds are presented on it}. Surprisingly, both theoretical time complexity and practical performance \rev{are improved}. On the theoretical side, \rev{this work is} the first to achieve $O(k \cdot n \cdot m)$ time complexity, whereas the state-of-the-art method achieves time complexity \rrev{of} $O(n^k)$.\footnote{$m$ and $n$ denote the number of edges and vertices, respectively. The $k$ denotes the hop constraint.}
On the practical level, \rev{the} proposed algorithm, namely \tdb, outperforms the state-of-the-art by $2$ to $3$ orders of magnitude on average while preserving the minimal property. 
As a result, \rev{the method in this paper} outperforms the state-of-the-art approaches in terms of both running time and theoretical time complexity. This is the first time, to our best knowledge, that the \kpc problem on billion-scale networks has been solved with a minimal\footnote{``Minimal" indicates local optimal in this paper.} cover set for $k > 3$.
%On the theoretical side, we prove it is UGC-hard(Unique Games Conjecture) to approximate the \kpc problem with length from $3$ to $k$ within $(k-1-\epsilon)$  for a directed graph G.
%On the practical side, two different algorithms are presented, which are \textbf{bottom-up} and \textbf{top-down}. The \underline{b}ottom-\underline{u}p app\underline{r}oach(\adpm) algorithm has the best cover size among all the algorithms. 
%The \underline{T}op-\underline{D}own \underline{B}locks(\tbs) algorithm has the fastest runtime.
%The block and BFS-filter techniques are also proposed to speed up \tbs. With all the optimizations, \tdb could reduce the worst time complexity from $O(n \cdot m^k)$\footnote{$m$ and $n$ denote the number of edges and vertices, respectively. The $k$ denotes the hop constraint.} to  $O(k \cdot n \cdot m)$.
%Our comprehensive experiments demonstrate the effectiveness and efficiency of our proposed methods in both cover size and runtime over 10+ datasets comparing with the state-of-the-art $k$-cycle transversal algorithm \darc. The \tdb outperforms it by $2$ to $3$ orders of magnitude on average while their cover sizes are similar. 
\end{abstract}

%\maketitle
\section{Introduction}
\label{sec:introduction}
\rev{The Feedback vertex set with the minimum size was one of the 21 NP-complete problems Karp~\cite{karp1972reducibility} developed in an effort to break every cycle in a graph. It was created way back in the early 1960s (see the survey of Festa \etal~\cite{festa1999feedback}). A wide variety of applications, such as operating systems~\cite{galvin2003operating}, database systems~\cite{gardarin1976integrity}, and circuit testing~\cite{leiserson1991retiming}, make use of it. Numerous studies have been conducted over many years, including those on approximation algorithms~\cite{bafna19992,bar1998approximation,even2000approximating,kleinberg2001wavelength}, linear programming~\cite{chudak1998primal}, parameterized complexity~\cite{dehne20072, downey2012parameterized, guo2005parameterized}, etc.}
% \rev{Karp~\cite{karp1972reducibility} introduced the Feedback vertex set with the minimum size as one of the 21 NP-complete problems aimed at breaking all the cycles in a graph. Its origins date all the way back to the early 1960s (see the survey of Festa \etal~\cite{festa1999feedback}). It is utilized in a wide variety of applications, including operating systems~\cite{galvin2003operating}, database systems~\cite{gardarin1976integrity}, and circuit testing~\cite{leiserson1991retiming}.
% Thus, tons of works have been investigated over many decades, including those on approximation algorithms~\cite{bafna19992,bar1998approximation,even2000approximating,kleinberg2001wavelength}, linear programming~\cite{chudak1998primal}, parameterized complexity~\cite{dehne20072, downey2012parameterized, guo2005parameterized}, etc.}

Nevertheless, \rev{it is observed} that individuals are primarily concerned about hop-constrained cycles in a large number of real-world applications. For instance, a team from Alibaba group recently \rrev{investigated} the \textit{hop-constrained cycles}~\cite{qiu2018real} in the context of financial fraud detection on e-commerce networks. As mentioned in~\cite{qiu2018real}, the fraud detection team disregards the cycles with a high number of hops due to their lack of relevance and exponential growth in search space. This leads us to investigate the feedback vertex set problem in terms of hop-constrained cycles, namely \textit{\kpc}. \rev{Specifically, given a directed graph $G$, we need a set of vertices containing all hop-constrained cycles whose length is constrained by a parameter on the graph; for each constrained cycle in the graph, at least one of its vertices must be in the feedback vertex.}
%\rev{To be more precise, given a directed graph $G$, we want a set of vertices that contains all the hop-constrained cycles on the graph; that is, for each constrained cycle in the graph, at least one of its vertices must be in the feedback vertex set.}

\eat{
The hop-constrained cycles are more practical than cycles in real applications, since the real applications are inherently constrained.
%Motivated by this, we consider the classical feedback vertex set with constraints, namely \textbf{\cvc}. Specifically, given a graph, we aim to find a set of vertices covering all constrained cycles on directed graphs.
Users are interested in cycles in a variety of graph analytic tasks. The cycle is a vital graph pattern that is usually associated with certain behaviors in many real-life applications, such as financial fraud detection, program analysis, and compiler optimization. Significant research efforts have been devoted to the cycle-related studies such as cycle enumeration (e.g.,~\cite{birmele2013optimal,ALLPath78,ALLPath80,ALLPath86}) and real-time cycle detection (e.g.,~\cite{qiu2018real}).

The problem of \textit{\cvc} can be regarded as an important extension of the well-known feedback vertex set problem~\cite{bafna19992} with specific concerns \rrev{about} constrained simple cycles (i.e., cycles with only the first and last vertices repeated).
%Specifically, given a graph, we aim to find a minimal set of vertices covering all constrained cycles on directed graphs.
}

\vspace{1mm}
\noindent\textbf{Applications.}~The following are some compelling applications of the \kpc problem.

%A new data structure called k-all-path cover (k-APC) has recently been developed for spatial networks~\cite{funke2014k}. A k-APC is a subset of vertices that literally covers all paths of size k. For a given graph, the problem of finding the minimum k-path vertex cover, an equivalent notion to k-APC, has been proposed and studied by Bresar et al. as a generalization of the vertex cover problem~\cite{brevsar2011minimum}.
%In this paper, we propose a similar problem named \kpc, which aims to find the $optimal$ k-hop bounded cycle vertex cover. The \kpc problem is attracting exceptional research interest because, it could be used in many applications in E-commerce Networks, transaction networks, and compiler system.

\vspace{0.5mm}
\noindent \textit{(1) \underline{Combinatorial Circuit Design}}. As shown in~\cite{bafna19992}, one typical use of the feedback vertex set is the combinatorial circuit design.
The circuits are depicted by graphs in which each cycle denotes a possible ``racing condition''. Certain circuit components may receive new inputs prior to \rrev{stabilization}. One method to avoid such a condition is to include a clocked register at each cycle in the circuit. Due to the fact that the ``racing condition'' is negligible for a long cycle~\cite{liu2018detailed,bafna19992}, the hop-constraint is imposed automatically in this application. As a consequence, this application enforces \rrev{hop-constraint} by default.

\vspace{0.5mm}
\noindent \textit{(2) \underline{E-commerce Networks}}.
Each node represents an account in an E-commerce Network, and each directed edge represents a money transfer between two accounts. Figure~\ref{fig:motivation} depicts an example of this kind. According to Alibaba Group specialists in~\cite{qiu2018real}, a \kc is a strong indicator of fraudulent activity or financial crime such as money laundering~\cite{yue2007review}. Figure~\ref{fig:motivation} depicts three simple cycles, i.e., potential money laundering behaviors. \rrev{By using a minimal \kc vertex cover,} we can identify a group of critical individuals who are more likely to participate in fraudulent activities. For instance, \kpc $\{a\}$ with the constraint $hop \leq 5$ is the most suspicious individual since it covers (i.e., becomes involved in) all three simple cycles with a length limitation of $5$.

%As experts in Alibaba Company~\cite{qiu2018real} suggested, finding an $optimal$ hop-constrained cycle vertex cover, which is the same as the \kpc problem, is critical in monitoring the E-commerce Networks.
%In Figure~\ref{fig:motivation}, there is an E-commerce network, where vertices represent accounts, and edges represent transactions. As shown in this figure, there are three simple cycles, which indicate potential financial crimes. If we find a minimal \kpc $\{a\}$, we only need to monitor the single account.

\begin{figure}[t]
	\centering

     \includegraphics[width=0.50\linewidth]{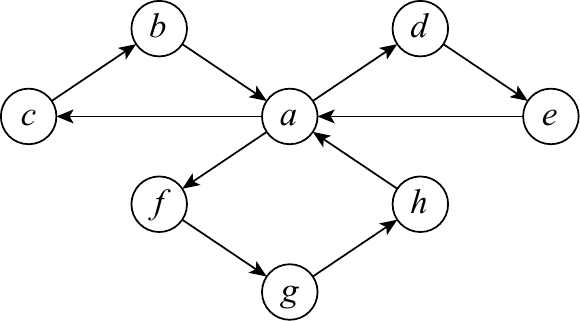}

%\vspace{-3mm}
\caption{\small An example of the e-commerce network, where vertices represent accounts, and edges represent transactions.}
%\vspace{-3mm}
\label{fig:motivation}
\end{figure}

%\vspace{0.5mm}
%\noindent \textit{(3) \underline{Deadlocks in Database Systems}}.
%The \cvc may assist in resolving deadlocks in database systems~\cite{da2011multicore}, where a vertex represents a transaction, and a directed edge reflects the wait-for relationship between two transactions
%owing to resource (e.g.,~\cite{agarwal2010detection,joshi2009randomized}) or communication (e.g.,~\cite{havelund2000using,joshi2010effective}) requirements.
%A cycle in the transaction network indicates the risk of a database system deadlock. We can resolve most of the deadlocks by identifying a minimal \cvc.
%Temporal constraints are automatically enforced in this application, since the edges of deadlocks are chronologically ordered. It is worth noting that after all the simple cycles are covered, all the cycles are covered as well. As a result, we only consider simple cycles in this paper.

\vspace{0.5mm}
\noindent \textit{(3) \underline{Program Analysis}}.
The \cvc could also be applied to identify and resolve deadlock potentials in program analysis, especially for concurrent applications~\cite{cai2014magiclock}. Deadlock is a frequent concurrency error occurs when a set of threads are blocked, and a constrained cycle in a lock graph signals the possibility of a deadlock~\cite{agarwal2010detection}. As a consequence, building a minimal \cvc is crucial in this field. %The term ``edge labels'' may refer to a range of various connections between programs. Thus, while analyzing certain types of edges, label constraints should be considered.

\noindent\textbf{Constraints.}~
%Real-world applications face a variety of constraints. As with~\cite{qiu2018real}, this paper considers two typical constraints: the simple-cycle constraint and the hop-constraint~\cite{liu2021hop}. The rationale for this is because these two constraints may assist avoid the formation of a large number of less significant cycles (e.g., a cycle with a high number of hops often implies a weak connection between the vertices).
\rev{ Various constraints might be placed on the cycle computation in real-world applications. We concentrate on two typical constraints in this paper, which follows the settings in~\cite{qiu2018real}: the simple and the hop constraint. Notably, these two constraints are beneficial in reality since they might result in considerably smaller results and fewer relevant cycles (e.g., a cycle with a high number of hops implies a weak connection between the vertices).}

\reat{
Additionally, in Section~\ref{subsec:moreCons}, we explore how to expand our methods to include other constraints such as label constraints and temporal constraints.}

% Our proposed techniques could be easily extended to address some other constraints such as label constraints or temporal constraints since we could use the constrained subgraph as input or impose the constraints during the search process.

Self-loops and bidirectional edges are not considered as cycles in this paper since they are uninteresting and substantially increase the result size.
%For instance, when bidirectional edges are considered as cycles in Citeseer, the cover size is about $6$ times larger than when they are not considered as cycles for $k=5$. 
Nota bene, the self-loop and bidirectional edge may be promptly verified if required.

\vspace{1mm}
\noindent\textbf{Contributions.}
\rrev{The} main contributions are listed as follows:
\begin{itemize}
\item \rev{\textit{Scalability}.~To the best of our knowledge, this is the first work to answer the problems of \kpc on billion-scale directed graphs, both of which have a wide range of real-life applications. To achieve this, we use an opposite cover process from existing works. In addition, delicate upper and lower bounds are proposed to achieve better theoretical time complexity and practical performance.}

\item \rev{\textit{Theoretical Analysis}.~For a given directed, unweighted graph $G$, we demonstrate that approximating the \kpc problem with lengths ranging from $3$ to $k$ within $(k-1-\epsilon)$ is UGC-hard (Unique Games Conjecture).}

% \item \textit{Efficient Algorithms}.~Two distinct algorithms are explored on a practical level: \textit{bottom-up} and \textit{top-down}. Among all methods, the \underline{b}ottom-\underline{u}p app\underline{r}oach algorithm (\adpm) has the best cover size. The \underline{T}op-\underline{D}own \underline{B}locks method (\tdb) achieves the fastest runtime and a comparable cover size by lowering the worst time complexity of the state-of-the-art from $O(n^k)$ to $O(k \cdot n \cdot m)$. Both methods are capable of ensuring the \textit{minimal} property.

\item \rev{\textit{Comprehensive Experiments.}~Compared to the baselines, our comprehensive experiments demonstrate the efficiency and effectiveness of our proposed method.}
%\item \textit{Efficiency and Effectiveness.}~Our comprehensive experiments on $10+$ real-life graphs demonstrate that our proposed approach \adpm outperforms all the baselines in terms of cover size. \tdb outperforms the state-of-the-art method \darc in terms of running time by an average of $2$-$3$ orders of magnitude for comparable cover sizes. To our best knowledge, \tdb is the only method capable of solving the \kpc problem in billion-scale graphs with a minimal result set.
\end{itemize}

\vspace{1mm}
\noindent\textbf{Organization of the paper.} The remainder of this paper is structured as follows. Section~\ref{sect:related} is devoted to related work.
In Section~\ref{sect:preliminary}, we introduce the preliminary. Section~\ref{sect:Theo} provides an examination of these problems from a theoretical standpoint. Bottom-up and \tpd algorithms are conducted in Sections~\ref{sect:rAda} and~\ref{sect:tpd}, respectively. Extensive experiments are conducted in Section~\ref{sect:experiment}. Finally, Section~\ref{sect:conclusion} concludes this paper.

%!TEX root = DamoGraph.tex
\section{Related Work}
\label{sect:related}
We review closely related works in this section.

\subsection{K-Cycle Traversal and K-Cycle-Free Subgraph}
\label{subsect:kCycleTransversal}
 \rev{The k-cycle transversal problem (K-cycle problem for short) is to find a minimum-size set of edges that intersects all simple cycles of length $k$ in a network. K-cycle problems are crucial in the disciplines of extremal combinatorics, combinatorial optimization, and approximation algorithms, as shown in~\cite{alon1996bipartite,alon2003maximum,erdos1996covering,krivelevich1995conjecture}.~\cite{krivelevich1995conjecture} investigates the $3$-Cycle Traversal. \cite{erdos1996covering} addresses $3$-Cycle Traversal, $3$-Cycle-Free Subgraphs, and their connections to related issues. In~\cite{pevzner2004novo}, the problem of locating a maximum subgraph devoid of cycles of length $\leq k$ is examined in the context of computational biology, and different heuristics are presented without analyzing their approximation ratio.}
% \rev{The k-cycle transversal problem (k-cycle problem for short), which is to find a minimum-size set of edges that intersects all simple cycles of length $k$ in a network.} As shown in~\cite{alon1996bipartite,alon2003maximum,erdos1996covering,krivelevich1995conjecture}, K-cycle problems are essential in the fields of extremal combinatorics, combinatorial optimization, and approximation algorithms.~\cite{krivelevich1995conjecture} investigates the $3$-Cycle Traversal. \cite{erdos1996covering} addresses $3$-Cycle Traversal, $3$-Cycle-Free Subgraphs, and their connections to related issues. The problem of finding a maximum subgraph without cycles of length $\leq k$ is studied in~\cite{pevzner2004novo} in the context of computational biology, and some heuristics are proposed for the problem without analyzing their approximation ratio. 

\reat{However, the majority of related studies examined k-Cycle-Free Subgraph in the context of extremal graph theory, with a particular emphasis on the maximum number of edges in a graph without $k$-cycles (or without cycles of length $\leq k$). This is essentially the k-Cycle-Free Subgraph problem on complete graphs.}

The most closely related problem is~\cite{kortsarz2008approximating}, which tackles both the problem of discovering a minimal edge subset of $E$ that intersects every \kc, and the problem of discovering a maximum edge subset of $E$ that does not intersect any \kcs. They provide a $(k-1)$-approximation for this problem, when $k$ is odd.~\cite{xia2012kernelization} investigates the kernelization for the cycle traversal problems. The most efficient method is described in~\cite{kuhnle2019scalable}. We use their method as our baseline.

%Nevertheless, the developed algorithms focus on undirected graphs and could not be adapted to our problem.

\subsection{Feedback Vertex Set}
\label{subsect:feedbackVertexSet}
Another related work is the \rev{feedback vertex set (FVS)} problem, which seeks to intersect a minimum-size of vertices with all cycles of any length in the network.~\cite{bafna19992} proposes a 2-approximation method for FVS in undirected graphs. The analogous edge version issue may be reduced to the minimum spanning tree problem~\cite{vazirani2013approximation}. As a result, it can be solved in polynomial time.% However, the equivalent problem is not much more complex in the vertex variant of directed graphs. 

\subsection{Cycle and Path Enumeration}
\label{subsect:pathCycleEnumeration}
With the rapid development of information technology, a growing number of applications represent data as graphs~\cite{peng2018efficient, jin2021fast, peng2021answering,yuan2022efficient,peng2022finding,chen2022answering,yang2021huge,qin2020software,lai2021pefp}. Numerous \rrev{research has} been conducted on the subject of enumerating $s$-$t$ simple paths (e.g.,~\cite{bohmova2018computing,knuth2011art,nishino2017compiling, yasuda2017fast,peng2021dlq,peng2020answering,peng2021efficient,qing2022towards}), with the \textit{simpath} algorithm proposed in~\cite{knuth2011art} being one among them. Due to the huge amount of results, \cite{qiu2018real} applies the hop-constrained path enumeration problem on the dynamic graphs and analyzes \rrev{them}.
An indexing technique named HP-index is proposed to continuously maintain the pairwise \hcst paths among a set of hot points (i.e., vertices with \rrev{a} high degree).%~\cite{peng2019towards} presents a \hcst method not only fulfills the state-of-the-art polynomial delay (O($km$) time per output), but also significantly outperforms the state-of-the-art practical methods. A join-oriented algorithm is developed to further improve performance.

There is a long history of study on enumerating all simple paths or cycles on a graph~\cite{ALLPath78,ALLPath80,ALLPath86,chen2022answering}.
Another area of study (e.g.,~\cite{DBLP:journals/jgaa/RobertsK07}) is counting or estimating the number of paths connecting two given vertices, which is a well-known $\#P$ hard problem.

\noindent\textbf{Hop-constrained Path Enumeration~\cite{Hao21,Hao22,peng2019towards}}.~\rev{Peng et al.~\cite{peng2019towards,peng2021efficient} designed a barrier-based method, which dynamically maintains the distance from each vertex to $t$. Initially, they set the barrier for each $v \in V(G)$ as $S(v,t|G)$. During the enumeration, if they find that a sub-tree rooted at a node in the search tree contains no result, then they will increase the barrier to avoid falling into the same sub-tree again. T-DFS~\cite{rizzi2014efficiently} and T-DFS2~\cite{grossi2018efficient} are two theoretical works. They achieve polynomial delay by ensuring that each search branch in the search tree leads to a result. For example, before extending $M$ by adding $v'$' in Algorithm 1, T-DFS checks whether there is a shortest path from $v$' to $t$ without vertices in $M$ whose length is bounded by $k-L(M)-1$. Although all three algorithms achieve $O(k \times |E(G)|)$ polynomial delay. Peng et al. showed that their method runs much faster than T-DFS and T-DFS2 in practice because their pruning strategy incurs lower overhead~\cite{peng2019towards}.}

%But these cannot be trivially extended to large scale dynamic graphs to speed up our problem in this paper due to the matrix operations involved or the sampling techniques for approximate solutions.

 %The classical problem of listing all the cycles of a graph has been extensively studied for its many applications in several fields, ranging from the mechanical analysis of chemical structures~\cite{sussenguth1965graph} to the design and analysis of reliable communication networks, and the graph isomorphism problem~\cite{welch1966mechanical}. There is a vast body of work. (see ~\cite{bezem1987enumeration,mateti1976,birmele2013optimal} for excellent surveys).

%!TEX root = DamoGraph.tex
\section{Preliminaries}
\label{sect:preliminary}
This section introduces the \cvc problem explicitly. Then, the state-of-the-art methods are presented.
% proposed in~\cite{hassan2016graph}, namely \textit{Edge-Disjoint Partitioning}(\textbf{EDP for short}).
\rev{Table~\ref{tb:notations} provides} the mathematical notations that are most frequently used throughout this paper.
%two baseline algorithms for LCR query: the naive $2$-hop based index, namely \textit{Naive Pruned 2-Hop}(\textbf{NP2H}) and the-state-of-the-art indexing technique proposed in~\cite{valstar2017landmark}, namely \textit{Landmark Indexing}(\textbf{LI+}).

%
\begin{table}[tb]
\small
  \centering
  \vspace{-1mm}
  \caption{The summary of notations.}
\label{tb:notations}
    \begin{tabular}{|c|l|}
      \hline
      \cellcolor{gray!25}\textbf{Notation} & \cellcolor{gray!25}\textbf{Definition}        \\ \hline
      $G$ &  a given directed unweighted graph \\ \hline
      $m$, $n$ & the number of edges(vertices) for $G$ \\ \hline
      $C$   & a \kpc for given graph G  \\ \hline
      $c$   & a simple \kc , where $ 3 \leq |c| \leq k$  \\ \hline
      $| C |$,$| c |$   & the size(length) for a vertice set $C$($c$)  \\ \hline
      $\mathcal{H}[v]$   & the hit times for vertex $v$  \\ \hline
      $CN$   & the cover node  \\ \hline
      $opt(G,k)$   & the optimal \kpc in $G$  \\ \hline
       $len(p)$ & length (i.e., number of hops) of path $p$ \\ \hline
      $sd(u,v)$ & shortest path distance from $u$ to $v$ \\
                & i.e., the minimal number of hops from $u$ to $v$ \\ \hline
      $sd(u,v|T)$ & shortest path distance from $u$ to $v$ \\
      & not containing any vertex in $T$ \\ \hline
         %$P_{k}$ & \hcst paths returned   \\ \hline
      $p[x]$ & number of hops the vertex $x$ can reach \\
             & the end vertex of $p$ along the path $p$ \\ \hline

      $\mathcal{S}$, $|\mathcal{S}|$ & the stack in DFS and its size \\ \hline
      $p(\mathcal{S})$ & the path associated with stack $\mathcal{S}$ \\ \hline
      $adj_{in}[v]$ & the in-neighbors of vertex $v$ \\ \hline
      $adj_{out}[v]$ & the out-neighbors of vertex $v$ \\ \hline
      $len(\mathcal{S})$ & the length of the path associated with $\mathcal{S}$ \\
                         & where $len(\mathcal{S}) = len(p(\mathcal{S}))= |\mathcal{S}|-1$  \\ \hline
   %   $\zeta_0 \cup \zeta_1, \zeta_0 \setminus \zeta_1$ & the union(difference) of two label sets\\ \hline
%      $p$, $p(u,v)$, $p (u \leadsto v)$ & a path in $G$, a path from $u$ to $v$ \\ \hline
%      $s$, $t$ & source and target vertices \\ \hline
%      hcst path & hop-constrained $s$-$t$ simple path \\ \hline
%      $len(p)$ & length (i.e., number of hops) of path $p$ \\ \hline
%      $sd(u,v)$ & shortest path distance from $u$ to $v$ \\
%                & i.e., minimal number of hops from $u$ to $v$ \\ \hline
%      $sd(u,v|T)$ & shortest path distance from $u$ to $v$ \\
%      & not containing any vertex in $T$ \\ \hline
%      $k$ & hop constraint \\ \hline
%      $P_{k}$ & hcst paths returned   \\ \hline
%      $p[x]$ & number of hops the vertex $x$ can reach \\
%             & the end vertex of $p$ along the path $p$ \\ \hline
%
%      $\mathcal{S}$, $|\mathcal{S}|$ & the stack in DFS and its size \\ \hline
%      $p(\mathcal{S})$ & the path associated with stack $\mathcal{S}$ \\ \hline
%      $len(\mathcal{S})$ & the length of the path associated with $\mathcal{S}$ \\
%                         & where $len(\mathcal{S}) = len(p(\mathcal{S}))= |\mathcal{S}|-1$  \\ \hline
%     % $C_{s,t,k}$, $C_{k}$ & a s-t k-path cut \\ \hline
%      %$m(p)$ & middle point vertex of the path $p$  \\ \hline
%      $P_m$ & the middle vertex cut \\ \hline
     \end{tabular}
%\vspace{-2mm}

%\vspace{-4mm}
\end{table}

\subsection{Problem Definition}
\label{sect:definition}
\rev{This subsection begins by} formally introducing the \kpc problem.
%In this subsection, we first formally introduce the problem of the \kpc.

$G = (V, E)$ is a directed graph containing the vertex set $V$ and the edge set $E$.
The symbol $e(u,v) \in E$ denotes a directed edge connecting the vertex $u$ to the vertex $v$.
When the context is clear, ``neighbor'' \rev{refers} to the ``out-going neighbor''. $adj_{in}[v]$ ($adj_{out}[v]$) \rev{designate} the in-neighbors (out-neighbors) of vertex $v$, respectively. A \textit{path} $p$ from the vertex $v$ to the vertex $v'$ is a sequence of vertices $v = v_0$, $v_1$, $\ldots$, $v_h = v'$, where $e(v_{i-1},v_i)$ $\in E$ for all $i \in [1,h]$.
A $circuit$ is a non-empty path with repeated start and end vertices, e.g., $(v_1,v_2,...,v_u,v_1)$.
The length of a $circuit$ is equal to its number of vertices.
A $cycle$ or $simple$ $circuit$ is a circuit where only the initial and final vertices are repeated.

The \kpc problem is defined as follows:
\begin{definition}[A Constrained Cycle]
\rev{Given a hop constraint $k$, \textit{A Constrained Cycle} $C$ is a cycle with $ 3 \leq |C| \leq k$, and there are no repeat vertices except the starting and ending vertices.}
\end{definition}

\begin{definition}[Hop-Constrained Cycle Cover]~\rev{Assume that $k$ is a positive integer. A \kpc of graph $G(V,E)$ is a subset of vertices $C \subset V$ such that, for any constrained cycle $C$, $C \cap \pi \neq \emptyset$.} %Nota bene, self-loops, and bidirectional edges are not considered as cycles, and $k \geq 3$.
\end{definition}

%However, there is also a similar definition but a different problem. We called it the \cvc.

%\begin{definition}
%\textbf{\cvc }~An \cvc of graph $G(V,E)$ is a subset of vertices $C \subset V$ such that, for every simple cycle $\pi = (v_1, v_2,...,v_i, v_1) $, $C \cap \pi \neq \emptyset$.
%\end{definition}

The following definitions apply to the optimal and minimal \kpc.
\begin{definition}[Optimal Hop-Constrained Cycle Cover]
The term ``an optimal \kpc'' refers to a collection of vertices $C_0$ that has the smallest size of all the \kpcs given a directed graph $G$.
%Given a directed graph $G$, a set of vertices $C_0$ is called an optimal \kpc if it has the smallest size among all the \kpcs. %, is an \kpc. ($2 \leq k \leq n$) It indicates that for every simple cycle $c=(v_1,v_2,..,v_u,v_1)$ and $|c| \leq k$, we have $C_0 \cap (v_1,v_2,..,v_u,v_1) \neq \emptyset$. This \kc is called an optimal \kpc if it has the smallest size among all the \kpcs.
\end{definition}

\begin{definition}[Minimal Hop-Constrained Cycle Cover]
Given a directed graph $G$, a collection of vertices $C_0$ is said to be a minimal \kpc if no vertex in the cover could not be deleted. %, is an \kpc. ($2 \leq k \leq n$) It indicates that for every simple cycle $c=(v_1,v_2,..,v_u,v_1)$ and $|c| \leq k$, we have $C_0 \cap (v_1,v_2,..,v_u,v_1) \neq \emptyset$. This \kc is called an optimal \kpc if it has the smallest size among all the \kpcs.
\end{definition}
%Note that the optimal \cvc is similar.

\eat{
\rev{This paper covers} all the simple cycles (\kcs), which is described as Theorem~\ref{the:simpleCycle}.
\begin{theorem}
\label{the:simpleCycle}
If $C = \{ \ c \  | \ c$ is a cycle $\in G \}$, $C_s = \{ \ c \ | \ c$ is a simple cycle $\in G \}$, and a vertex set $V_0$ is cover set for $C_s$ s.t. $\forall c \in C_s$, $\exists v \in V_0$, and $ v \in c$, then $\forall c \in C$, $\exists v \in V_0$, and $ v \in c$.
%Once all the simple cycles are covered, all the cycles are covered.
\end{theorem}
\begin{proof}
\rev{It is proven} by contradiction. Assume there is a cycle $c_0 \in C$ and there does not $\exists v \in V_0$, s.t. $v \in c_0$. Notably, any non-simple cycles could be decomposed into several simple cycles. Thus, there $\exists c_1 \subset c_0$ where $c_1$ is a simple cycle since $V_0$ is a cover set for the simple cycles. Then, there must $\exists v_2 \in V_0$ and $v_2 \in c_1$. If \rev{it chooses} $v_2$ the vertex to cover $c_0$, there is a contradiction.
\end{proof}
\noindent \textbf{Remark.}~This is also $true$ for hop-constrained cycles. The proof is similar. The main idea is that a $k$-hop constrained cycle could be divided into several small $k$-hop constrained simple cycles. In fact, Theorem~\ref{the:simpleCycle} is $true$ for any \rrev{constraint} if a cycle of such constraints could be decomposed into several small simple cycles that still satisfy these constraints. 

Theorem~\ref{the:simpleCycle} indicates that once all the simple circuits are covered, all the cycles are covered. Hence, \rev{only simple cycles are considered} in this paper.
}

\subsection{The State-of-the-art }
\label{subsec:baseline}
%to do here
\eat{
\noindent\textit{Theoretical Perspective.}~The most related theoretical works are the path covers. There are two types of them, i.e., \underline{k}-hop \underline{A}ll \underline{P}ath \underline{C}over (k-APC) and \underline{k}-hop \underline{S}hortest \underline{P}ath \underline{C}over (k-SPC). The properties of k-APC were theoretically analyzed in~\cite{funke2014k}. Unfortunately, by reduction from the vertex cover problem, the problem of minimizing the size of k-APC has proven to be NP-hard for $k$ $\ge$ $2$. Furthermore, the k-APC problem cannot be approximated for $k$ $\ge$ $2$ in polynomial time within a factor of $1.3606$ (unless $P$ $=$ $NP$), which is also inherited from the vertex cover problem~\cite{brevsar2011minimum}. The NP-hardness and inapproximability also apply to k-SPC, as pointed out in~\cite{funke2014k}. The upper bound on the size of a k-SPC has been discussed using the theory of VC dimensions. Please refer to ~\cite{funke2014k} for details.
}
%In this paper, we prove that \kpc problem cannot be approximated for $k \ge 2$ in polynomial time within any constant factor (unless P= NP). The proof will be discussed in the following section.

\noindent\textit{Practical Perspective.}~Following that, the state-of-the-art $k$-cycle algorithm DARC from~\cite{kuhnle2019scalable} \rev{is introduced} as the baseline. \rev{Firstly, the definition of $k$-cycle problem is given:}

\begin{definition}[K-cycle problem]
\rev{Given a graph $G = (V, E)$, determine minimum-size set $S \subseteq E$ such that for each constrained cycle $C \in C_k$, $C \cap S \neq \emptyset$, or, equivalently, $C_k(G \backslash S) = \emptyset$.}
\end{definition}

%The \darc consists of two phases.

% \textit{AUGMENT}. In this stage, we augment uncovered vertices and discover uncovered hop-constrained cycles. The final cover is constructed entirely from its vertices.% are used to create the final cover.%All the vertices on it are chosen for the resulting cover. 

% \textit{PRUNE}. At this step, we remove any extraneous vertices to ensure that the result set remains feasible.

Algorithm~\ref{alg:darc} illustrates the details of this algorithm. The process is to iteratively go through all edges in $E$ in Line~\ref{darc:iter}. If the current edge is not included in the result set $S$, then $AUGMENT(e)$ is called. After all edges are evaluated, \rev{it is called} $PRUNED$ function.

\begin{algorithm}[htb]
\SetVline
\SetFuncSty{textsf}
\SetArgSty{textsf}
\small
\caption{ \rev{\textbf{DARC}}}
\label{alg:darc}
\Input
{   \rev{Graph $G = (V, E)$, sets $S, W, P \subseteq E$, $\mathcal{U} \subseteq C_k$ \\
        function $h$ : $(S \cup W ) \rightarrow \mathcal{U}$ \\
	$k$: the hop constraint \\}
}   
% \Output
% {
% $\mathcal{R}$: the cover set
% }
\rev{\State{$S \leftarrow \emptyset$, $h \leftarrow \emptyset$, $P \leftarrow \emptyset$, $W \leftarrow \emptyset$, $\mathcal{U} \leftarrow \emptyset$}
\For{$e \in E$}
{
\label{darc:iter}
    \If{$e \notin S$}
    {
        \State{AUGMENT($e$)}
    }
}
\State{PRUNE()}}
\end{algorithm}

Thus, the DARC consists of two phases.

\textit{AUGMENT}. In this stage, \rev{it augments} uncovered vertices and \rev{discovers uncovered hop-constrained cycles.} The final cover is constructed entirely from its vertices. \rev{The details are shown in Algorithm~\ref{alg:augment}. If $e \in S$, then this terminates in Line~\ref{aug:skip}. It is noted that $W = \emptyset$ in the initial stage, thus \rev{it} can ignore it. Then, in Line~\ref{aug:check} it checks every constrained cycle going through $e$ and \rev{adds} all edges to the result set $R$ and $P$. In addition, the relationship between constrained cycles and edges \rev{is} recorded in Line~\ref{aug:record}.}

\begin{algorithm}[htb]
\SetVline
\SetFuncSty{textsf}
\SetArgSty{textsf}
\small
\caption{ \rev{\textbf{AUGMENT($e$)}}}
\label{alg:augment}
\Input
{   \rev{Graph $G = (V, E)$, sets $S, W, P \subseteq E$, $\mathcal{U} \subseteq C_k$ \\
        function $h$ : $(S \cup W ) \rightarrow \mathcal{U}$ \\
	$k$: the hop constraint \\}
}   
% \Output
% {
% $\mathcal{R}$: the cover set
% }
\rev{\If{$e \in S$}
{
\label{aug:skip}
    \State{Return}
}
\ElseIf{$e \in W$}
{
    \State{Move $e$ from $W$ to $S$}
    \State{Add $e$ to $P$}
    \State{Return}
}
\For{$C$ $\in$ $\Delta_k(e)$}
{
\label{aug:check}
    \If{$C \cap S = \emptyset$}
    {
        \If{$C \cap W = \emptyset$}
        {
            \State{Add all edges of $C$ to $S$ and $P$}
            \label{aug:add}
            \State{Add $C$ to $\mathcal{U}$ and set $h(e) = C$ for all $e \in C$}
            \label{aug:record}
        }
        \Else{
            \State{Move any edge in $C \cap W$ to $S$ and $P$}
        }
    }
}}
\end{algorithm}

\textit{PRUNE}. At this step, \rev{it removes} any extraneous vertices to ensure that the result set remains feasible. \rev{The details are in Algorithm~\ref{alg:pruned}. For every candidate edge $e \in P$, it tries to remove it from the result set $S$ in Line~\ref{prune:try}. If yes, it will delete $e$ from $S$ and add it to $W$.}

\begin{algorithm}[htb]
\SetVline
\SetFuncSty{textsf}
\SetArgSty{textsf}
\small
\caption{ \rev{\textbf{PRUNE()}}}
\label{alg:pruned}
\Input
{   \rev{Graph $G = (V, E)$, sets $S, W, P \subseteq E$, $\mathcal{U} \subseteq C_k$ \\
        function $h$ : $(S \cup W ) \rightarrow \mathcal{U}$ \\
	$k$: the hop constraint \\}
}   
% \Output
% {
% $\mathcal{R}$: the cover set
% }
\For{\rev{$e \in P$}}
{
    \State{\rev{$P \leftarrow P \backslash \{ e\}$}}
    \If{\rev{$e \in S$}}
    {
        \If{\rev{$S \backslash e$ remains feasible to a cover}}
        {
        \label{prune:try}
            \rev{\State{$S \leftarrow S \backslash \{e \}$}
            \label{pruned:yes1}
            \State{$W \leftarrow W \cup \{ e \}$}}
            
        }
    }
}
\end{algorithm}

The specifics of these two phases may be seen in~\cite{kuhnle2019scalable}, the source code for which is publicly accessible.

\rev{\noindent\textbf{Modification to the vertex version}.~Since DARC is a method developed for discovering a minimum edge subset that does not intersect with any hop-constrained cycles, \rev{it is modified} as \rev{the} baseline. The modification is as follows: for the original graph $G(V, E)$, it \rev{is converted} to a new graph $G'(V', E')$. For every $e_{u,v} \in E$, \rev{it generates} a corresponding $v_{u,v} \in V'$, an edge $e'$ is added from vertices $v_{u,v}$ to $v_{v,w}$ since there is a common $v$ between them. Then, the edge set could be converted to the vertex result set. The modified algorithm is called $DARC$-$DV$.}
%Notably, the key idea of \darc is similar to that of our \textit{Bottom-Up} algorithm.
%Since the authors indicate that the worse case time complexity of \darc is $O(n^k)$, there is a lot of room for improvement.
As demonstrated in~\cite{kuhnle2019scalable}, the worse case time complexity of \darc is $O(n^k)$.
\rev{This paper intends} to enhance the findings in terms of cover size and efficiency.

\section{Theoretical Analysis}
\label{sect:Theo}

\rev{This section begins} with a theoretical study of the \kpc issue.

%When the cycles of length $2$ are included, it is trivial to prove the NP-hardness for the decision problem. When $k = 2$ for an undirected graph, we have a typical vertex cover issue since each edge is a cycle. Additionally, such a problem is UGC-hard (Unique Games Conjecture) to approximate within  $(k-1-\epsilon)$~\cite{guruswami2014}. Nonetheless, it is non-trivial to decide if it is the NP-hard and UGC-hard when the bidirectional edges are excluded as cycles.

\rev{Theorem~\ref{the:np_opt} proves that it is NP-hard to decide whether there is a \kpc for a given directed, unweighted graph $G$ with size $s$.}
\begin{theorem}
\label{the:np_opt}
\rev{Deciding whether there is a \kpc for a given directed, unweighted graph $G$ with size $s$ is \np for all constrained cycles.}
%It is \np to find an optimal \kpc for a given directed, unweighted graph G for all the cycles length from $3$ to $k$.
\end{theorem}
\begin{proof}
\rev{It reduces} an NP-complete problem, the decision version of vertex cover problem to our \kpc problem. Take note that the \kpc issue does not include self-loops or bidirectional edges. 
As a result, \rev{it sets} $k = 3$ for a \kpc issue on undirected graphs. Then, \rev{it} can reduce the traditional vertex cover issue to ours.

\rev{It adds} a virtual vertex $u'$ and two bidirectional edges $(u, u')$ and $(v, u')$ to each bidirectional edge $u, v$. Figure~\ref{fig:np_proof} illustrates the proof. The original graph is shown in Figure~\ref{fig:np_proof1}, whereas our constructed graph is shown in Figure~\ref{fig:np_proof2}. Take note that the virtual vertex $b'$ is dominated by the matching edge $(b,c)$, since $b'$ only participates in only single cycle. On the built graph, the \kpc problem is equivalent to the classical vertex cover on the original graph.
%Thus, we prove that finding the optimal \kpc is an \np problem. 
\end{proof}

\begin{figure}[htb]
	\vspace{-2mm}
	\newskip\subfigtoppskip \subfigtopskip = -0.1cm
	\newskip\subfigcapskip \subfigcapskip = -0.1cm
%     	\begin{minipage}[b]{\linewidth}
%		\centering
%		\includegraphics[width=0.5\linewidth]{main_graph_keys.eps}%
%	\end{minipage}	
	\centering
     \subfigure[\small{Original graph}]{
     \includegraphics[width=0.45\linewidth]{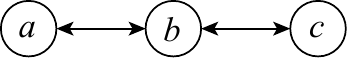}
	 \label{fig:np_proof1}
     }
     \subfigure[\small{Constructed graph}]{
     \includegraphics[width=0.45\linewidth]{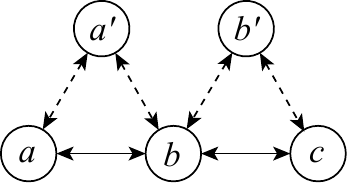}
	 \label{fig:np_proof2}
     }
%	\vspace{-3mm}
\caption{The proof of NP-hardness.}
%	\vspace{-3mm}
%\vspace{-0.3cm}
\label{fig:np_proof}
\end{figure}

After that, \rev{this work demonstrates} the inapproximation.

\begin{theorem}
For a given directed, unweighted graph $G$, approximating the \kpc problem (length from $3$ to $k$) within $(k-1-\epsilon)$ is UGC-hard (Unique Games Conjecture).
%It is UGC-hard(Unique Games Conjecture) to approximate the \kpc problem(length from $3$ to $k$) within  $(k-1-\epsilon)$  for a given directed, unweighted graph G.
\end{theorem}
\begin{proof}
Guruswami and Lee~\cite{guruswami2014} established the UGC-hard for the inapproximability of the feedback set for bounded cycles. Our issue is with the version that excludes self-loops and cycles with a length of $2$. Due to the fact that self-loops have a single vertex, they could be included in the result set. \rev{As a consequence, it places a premium on excluding \rrev{the} $2$-cycle.} \rev{The formal proof of inapproximability when self-loops are included is similar to the proof of $2$-cycle exclusion}. It is proved by contradiction. Assume that a $(k-1-\epsilon)$ approximation method exists for covering all \rrev{cycle lengths} between $3$ and $k$. 

$S(G,3,k)$ denotes the approximation method, with $G$ denoting the graph, and $3$ and $k$ denoting the length is between $3$ and $k$. $| S(G,3,k) | \leq (k-1-\epsilon) |Opt(G,3,k)|$, where $S$ is the approximation algorithm and $Opt$ is the optimal algorithm for cycles in graph $G$ lengths ranging from $3$ to $k$. Furthermore, $|Opt(G,2,k)| \geq |Opt(G,2,2) + Opt(G-(Opt(G,2,2)),3,k)|$. There is a simple $2$-approximation method for $2$-cycle by selecting all the vertices on them, i.e., $S(G,2,2) \leq 2|Opt(G,2,2)|$.
\reat{
Then, \rev{it} could combine $S(G,2,2)$ and $S(G-S(G,2,2),3,k)$.
\begin{align}
\notag
S &= |S(G,2,2)+ S(G-S(G,2,2),3,k)| \\
\notag
	&\leq 2 |Opt(G,2,2)| + (k-1-\epsilon)|Opt(G-S(G,2,2),3,k)|\\ 
%\notag
%	&= 2 |Opt(G,2,2)| + (k-1-\epsilon)|Opt(G-S(G,2,2),3,k)| \\
\notag
	&\leq (k-1-\epsilon)|Opt(G,2,k)|  (k \geq 3)\\
\end{align}
Note that \rev{it chooses} $S(G, 2, 2)$ as the set of all vertices in 2-cycles.}

All the cycles could be divided into three types.
\begin{itemize}
    \item $C_2$ : $2$-cycle (length equals to $2$) and no intersection with $k$-cycles ($k \geq 3$).
    \item $C_3$ : $k$-cycles and no intersection with $2$-cycle.
    \item $C_{23}$ : $2$-cycles intersect with $k$-cycles ($k \geq 3$) and $k$-cycles intersect with $2$-cycles.
\end{itemize}
\rev{It divides} all the vertices on cycles into three categories.
\begin{itemize}
    \item $V_2$ : only appears on the $C_2$.
    \item $V_3$ : only appears on the $C_3$. Note that \rev{it only includes} the $k$-cycles, where all the vertices on it are in $V_3$. If a $k$-cycle intersection with a $2$-cycle ($C_{23}$), \rev{it only selects} the intersection part into $V_{23}$. The remaining part would be covered because our approximation algorithm would select all the intersection vertices.
    \item $V_{23}$ : appears on the intersection part of $2$-cycles and k-cycles ($k \geq 3$). For example, if vertex $v$ is in both a $2$-cycle and a $k$-cycle, then it is in the $V_{23}$.
\end{itemize}
Similarly, $Opt(G,2,k)$ could be divided into $V^{Opt}_2$, $V^{Opt}_3$, and $V^{Opt}_{23}$. $V^{Opt}_2$ is the intersection of $Opt(G,2,k)$ and $C_2$.
For $V^{Opt}_3$,
\begin{align}
\notag
S(V_3, 3, k) &\leq (k- \epsilon) Opt(V_3, 3, k)\\
\notag
                  &\leq (k- \epsilon)V^{Opt}_3 \\
\end{align}
For $V^{Opt}_2$, \rev{there is} a trivial $2$-approximation algorithm. Thus, 
\begin{align}
%\notag
S(V_2, 2, 2) \leq 2 Opt(V_2,2,2) \leq V^{Opt}_2
\end{align}
For $V^{Opt}_{23}$, \rev{it selects all of} them, and it is not hard to prove that $S(V^{Opt}_{23}, 2, k) = V_{23} \leq 2 Opt(V^{Opt}_{23}, 2, k)$. The proof is similar to the approximation of $V_2$.
\rev{It} could also divide these cycles into $2$-cycle and $k$-cycle, while $V_{23}$ is the $2$-approximation solution of $2$-cycle part. The Optimal solution of all cycles is no less than the $2$-cycle part.

Since the limit of $\epsilon$ is close to $0$, 
\begin{align}
\notag
S&(V_2, 2, 2) + S(V_3, 3, k) + S(V_{23}, 2 , k) \\
\notag
  &\leq 2 Opt(V_2,2,2) + (k - \epsilon) Opt(V_3, 3, k) + V_{23} \\
\notag
  &\leq 2 V^{Opt}_2 + (k- \epsilon) V^{Opt}_3 + 2 V^{Opt}_{23} \\
\notag
  &\leq (k - \epsilon) Opt(G,2,k) \\
\end{align}
This violates the inapproximability for $Opt(G,2,k)$ within $(k-1-\epsilon)$.
\end{proof}

\section{Bottom-Up Approach}
\label{sect:rAda}
This section introduces our bottom-up \kpc algorithm, which is designed to minimize the size of the cover.

%In this section, we introduce our bottom-up \kpc algorithm, which aims to reduce the cover size.
\subsection{Motivation}
\label{subsec:AdaMotivation}
This subsection discusses the rationale for the \textit{bottom-up} \kpc algorithm.
%The rationale for the bottom-up \kpc algorithm is discussed in this part.
%In this subsection, the motivation is provided for the bottom-up \kpc algorithm. 
%A straightforward solution of the \kpc problem is to enumerate all possible hop-constrained cycles exhaustively and then select the cover vertex using the greedy strategy.

Given the \textit{NP-hardness} \rev{of this problem}, it is typical to employ the greedy heuristic. That is to determine the best cover vertex iteratively, i.e., the vertex that covers the largest number of uncovered cycles in the current iteration. Nonetheless, determining the optimal cover vertex requires enumerating all the \kcs, which are prohibitively difficult in terms of time and space complexities. The enumerating time complexity $O(2^n \times cost_{c} )$ is prohibitively expensive, while $cost_c$ is the cost for the check. 

The more \kcs a vertex covered in the previous iterations, the more probable it will cover additional cycles in the remaining graph. \rev{Based on this motivation, a heuristic greedy algorithm is proposed based on the number of cycles the vertices covered.}

\begin{example}
A motivational scenario is depicted in Figure~\ref{fig:adp_motivation}. $C$ is the graph's center vertex. Assume that in the first iteration, \rev{a cycle $A \rightarrow B \rightarrow C$ is found and it adds $A$ to the cover set by random.} $C$ \rev{will} be chosen in the second iteration since it occurred in the preceding cycle and is therefore more likely to cover more cycles.
\end{example}
\rev{Based on this example, a bottom-up \kpc method is proposed.}
The key \rev{idea} is that when discovering a cycle during the search, \rev{it records} hit-times ($\mathcal{H}$) of all the vertices on it and \rev{chooses} the one with the highest hit-times ($\mathcal{H}$).

\begin{figure}[htb]
	\centering

     \includegraphics[width=0.4\linewidth]{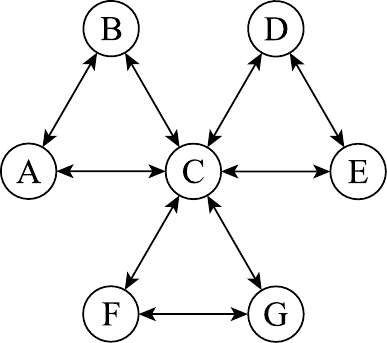}

%\vspace{-3mm}
\caption{\small A motivation example of the bottom-up approach.}
%\vspace{-3mm}
\label{fig:adp_motivation}
\end{figure}

\subsection{The Bottom-Up Approach}
\label{subsec:Adaptive}
The greedy method is an efficient solution for solving the \kpc problem. The key principle is that for each iteration, we select the vertex that covers the most \kcs.
Due to the difficulty of enumerating all the \kcs, a heuristic algorithm is proposed to solve it.

\begin{algorithm}[htb]
\SetVline
\SetFuncSty{textsf}
\SetArgSty{textsf}
\small
\caption{ \textsc{Bottom-Up}($G,k$)}
\label{alg:adp}
%\Input
%{
%	$G$: the input graph \\
%	$k$: the hop constraint \\
%}
%\Output
%{
%$\mathcal{R}$: the cover set
%}
\State{$\mathcal{R}$ $\leftarrow$ $\emptyset$}
\State{$\mathcal{H}[v] \leftarrow 0$, for {\bf each} $v$ $\in$ $G$}
\label{adp:init}
\For{{\bf each} $v_i$  $\in$ $V$}
{
\label{adp:for}
	\State{$c$ $\leftarrow$ \textsc{FindCycle}($G,k,\emptyset,v_i$)}
	\label{adp:findC}
	\While{$c \neq \emptyset$}
	{
	\label{line:while_nempty}
	%\State{updateH($\mathcal{H},c$)}
	\For{\bf{each} $v$ $\in$ $c$}
	{
		\label{line:updateHT}
		\State{$\mathcal{H}$[c] $\leftarrow$ $\mathcal{H}$[c]+1}
	}
	%\If{$c \neq \emptyset$}
	%{
		\State{$u$ $\leftarrow$ \textsc{FindCoverNode}($v_i,\mathcal{H}, c$)}
		\label{adp:findNode}
		\State{Insert $u$ into $R$}%, $i = i-1$}
		\label{adp:iMinus}
		\State{Remove in-edges and out-edges of $u$ $\in$ $G$}
		\label{line:removeEdges}
	%}
	\State{$c$ $\leftarrow$ \textsc{FindCycle}($G,k, \emptyset, v_i$)}
	\label{adp:findC2}
	}
}
\State{Return $\mathcal{R}$}
\end{algorithm}
Algorithm~\ref{alg:adp}'s \rev{main idea} is to \rev{find} a cover vertex \rev{with} the $\mathcal{H}$ array. \rev{ Line~\ref{adp:init} initializes} $\mathcal{H}[v]$ to $0$ for each vertex $v$ on graph $G$.
Then, Line~\ref{adp:for} is a for-loop that iterates over all vertices in $G$. \rev{Lines~\ref{adp:findC} and~\ref{adp:findC2} attempt to identify} a \kc $c$ beginning at $v$.
The $\mathcal{H}$ array is updated for each vertex on $c$. Line~\ref{adp:findNode} specifies that \rev{it selects one vertex $u$ from $c$}, and \rev{eliminates} all of $u$'s associated edges in Line~\ref{line:removeEdges}.
Whenever $c \neq \emptyset$ in Line~\ref{line:while_nempty}, the algorithm continues the procedure.

Algorithm~\ref{alg:findCycle} employs a recursive approach to \rev{find} a \kc \rev{starting from the vertex} $v$. The graph $G$ is a reduced graph, since it has no vertex in the \rev{current} result set $R$. The key \rev{point} is to identify a \kc using a DFS method. 
%The $curPath$ stores the current path. 
Line~\ref{findC:true} demonstrates that this method identifies a valid \kc, with the condition $\mathcal{CD} > 0$ indicating that self-loops should be avoided. Line~\ref{findC:false} illustrates the situation \rev{where} the algorithm fails to locate a valid \kc. Line~\ref{findC:out} investigates all of $v$'s out-neighbors. Then, Lines~\ref{findC:push} and~\ref{findC:pop} are to recursively push and pop all $v$'s out-neighbors.

%In Algorithm~\ref{alg:findCycle}~Line~\ref{line:updateHT}, it updates the array $\mathcal{H}$ by using currently found \kc $c$. The update rule is straightforward by plus one for each vertex on the \kc. 
%We could update it in other ways, e.g., plus the amount of degree of $v$.

% If several vertices have the same $\mathcal{H}$, we choose the vertex, which is farthest from $v$. 
Algorithm~\ref{alg:findCoverNode} employs $\mathcal{H}$ to determine the cover vertex. Line~\ref{findCN:nodeOrder} attempts to locate the vertex with the maximum $\mathcal{H}$.
%When multiple vertices have the same maximum $\mathcal{H}$, we choose the vertex that is the furthest from our beginning vertex $v$.

%\eat{We choose the vertex in the following order $\{v_{\frac{l}{2}},v_{\frac{l}{2} +1},v_{\frac{l}{2}-1},..,v_{\frac{l}{2}+i},v_{\frac{l}{2}-i},..,v_0\}$, where $l$ is the number of vertices on $c$.} 
%The most tricky part is the vertex order of the for-loop in Line~\ref{findCN:nodeOrder}.

%Here is an example to illustrate how this algorithm works.

\begin{algorithm}[htb]
\SetVline
\SetFuncSty{textsf}
\SetArgSty{textsf}
\small
\caption{ \textsc{FindCycle}($G,k,\mathcal{CP},v$)}
\label{alg:findCycle}
%\Input
%{
%	$G$: the input graph \\
%	$k$: the hop constraint \\
%	$\mathcal{CP}$: the current path set \\
%	$v$: the starting vertex\\
%}
%\Output
%{
%$\mathcal{C}$: the result cycle
%}
\State{$\mathcal{C}$ $\leftarrow$ $\emptyset$}
\StateCmt{$\mathcal{CD}$ $\leftarrow$ len($\mathcal{CP}$)}{the current distance}
\If{$\mathcal{CD}$ $>$ 0 $\land$ $v = \mathcal{CP}[0]$}
{
\label{findC:true}
	\State{Return $\mathcal{CP}$}
}
\If{$\mathcal{CD}$ $>$ k}
{
\label{findC:false}
	\State{Return $\emptyset$}
}
\State{$\mathcal{CP}$.pushBack($v$)}
\label{findC:push}
\For{{\bf each} vertex $u$ of $adj_{out}[v]$ on $G$}
{
\label{findC:out}
	\State{$\mathcal{C}$ $\leftarrow$ \textsc{FindCycle}(G,k,$\mathcal{CP}$,$u$)}
	\If{$C \neq \emptyset$}
	{
		\State{Return $\mathcal{C}$}
	}
}
\State{$\mathcal{CP}$.pop()}
\label{findC:pop}
\State{Return $\mathcal{C}$}
\end{algorithm}

%\begin{algorithm}[htb]
%\SetVline
%\SetFuncSty{textsf}
%\SetArgSty{textsf}
%\small
%\caption{ \textbf{UpdateHitTimes($\mathcal{H},c$)}}
%\label{alg:updateHT}
%\For{every vertex $v$ on $c$}
%{
%	\State{\mathcal{H}[c] $=$ \mathcal{H}[c]+1}
%}
%\State{Return $\mathcal{H}$}
%\end{algorithm}

\begin{algorithm}[htb]
\SetVline
\SetFuncSty{textsf}
\SetArgSty{textsf}
\small
\caption{ \textsc{FindCoverNode}($v,\mathcal{H},c$)}
\label{alg:findCoverNode}
%\Input
%{
%	$v$: the source vertex \\
%	$\mathcal{H}$: the array to store the hit information \\
%	$c$: the cycle\\
%}
%\Output
%{
%$\mathcal{CN}$: the result cover node
%}
\State{$\mathcal{H}_{max} = \mathcal{H}[v_0]$, $\mathcal{CN} = v_0$ }
\For{\bf{each} $v$ $\in$ $c$}
{
\label{findCN:nodeOrder}
	\If{$\mathcal{H}[v]> \mathcal{H}_{max}$}
	{
		\State{$\mathcal{H}_{max} = \mathcal{H}[v]$}
		\State{$\mathcal{CN}=v$}
	}
}
\State{Return $\mathcal{CN}$}
\end{algorithm}

\vspace{1mm}
\noindent\textbf{Correctness.}~\rev{Since Algorithm~\ref{alg:adp} traverses all the vertices in $V$ and increments the result set by one vertex until no new constrained cycle is detected, it is self-evident that the induced graph (which is constructed by eliminating all the vertices from the cover set) contains no hop-constrained cycles.}
\eat{
Assume that there is a \kc $c_0$, which is initially explored by Algorithm~\ref{alg:adp} with vertex $v_c$.
When we explore vertex $v_c$ for the first time, the $FindCycle$ produces a \kc $c_i$. The $c_i$ could be $c_0$ or not. There are two cases.
\begin{itemize}
\item If $c_i = c_0$, our algorithm covers $c_0$ in the subsequent steps.
\item If $c_i \neq c_0$, we choose a cover vertex $CN$. 
\end{itemize}
Additionally, there are two scenarios depending on whether $CN=v_c$ or not if $c_i \neq c_0$.
\begin{itemize}
\item If $CN = v_c$, the $c_0$ is covered by vertex $v_c$.
\item If $CN \neq v_c$, we will continue to check vertex $v_c$, in accordance with Line~\ref{adp:iMinus} Algorithm~\ref{alg:adp}.
\end{itemize}
Thus, the $c_0$ will be covered by the result set of Algorithm~\ref{alg:adp}.
}

\vspace{1mm}
\noindent\textbf{Time and Space Complexities.}~Algorithm~\ref{alg:adp} employs an array of $\mathcal{H}$ and a $k$-step DFS procedure. As a result, the space complexity is $O(m)$. \rev{Algorithm~\ref{alg:adp} Line~\ref{adp:for} needs $n$ iterations.} Each iteration contains three \rev{steps}: $FindCycle$, $UpdateH$, and $FindCoverNode$. Both $UpdateH$ and $FindCoverNode$ \rev{takes} a time complexity of $O(k)$ time complexity, \rev{due to the fact that} they include a for-loop on \kc $c$. $FindCycle$ takes $O(n^k)$, \rev{since} it is a DFS algorithm that determines whether there is a \kc within $k$ steps. \rev{Thus}, the overall time complexity is $O(n^{k+1})$. 

Nonetheless, the practical performance is acceptable, due to the following reasons.
To begin, $FindCycle$ could early terminate when it finds a \kc. It is not necessary to locate all the \kcs.
Second, each time we choose a cover vertex, the in-edges and out-edges of it would be eliminated from the graph $G$. Thus, the graph is becoming smaller and smaller.

\subsection{The Minimal Pruning Algorithm}
\label{subsec:minimal}
\rev{This subsection proposes a minimal pruning method s.t. it can further reduce the \kpc to a \textit{minimal} result set.} \eat{The main idea is to remove every \rev{possible} vertex that can be eliminated until the minimal one.}

\begin{algorithm}[htb]
\SetVline
\SetFuncSty{textsf}
\SetArgSty{textsf}
\small
\caption{ \textsc{FindMinimalCover}($G,k,R$)}
\label{alg:minimal}
%\Input
%{
%	$G$: the input graph \\
%	$k$: the hop constraint \\
%}
%\Output
%{
%$\mathcal{R}$: the cover set
%}
\For{\bf{each} $v$ $\in$ $\mathcal{R}$}
{
\label{lmin:for}
	\State{$c$ $\leftarrow$ \textsc{FindCycle}($G-R+(v),k, \emptyset,v$)}
	\label{lmin:findC}
	\If{$c = \emptyset$ }
	{
	\label{lmin:update}
		\State{Remove $v$ from $\mathcal{R}$}
		\label{lmin:remove}
		%\State{Add all in-edges and out-edges of $v$ into $G$}
		%\label{lmin:add}
		%\State{$i = i -1$}
	}
}
\State{Return $\mathcal{R}$}
\end{algorithm}

Algorithm~\ref{alg:minimal} Line~\ref{lmin:for} verifies each vertex in the $R$, which is the \kpc generated by Algorithm~\ref{alg:adp}. \rev{Line~\ref{lmin:findC} attempts to find a \kc inside $G-R+(v)$.} Nota bene, \rev{in Algorithm~\ref{alg:minimal}}, the input graph $G-R+(v)$ is the reduced graph, which has no vertex in $R$ except for the vertex $v$. Line~\ref{lmin:update} determines if there exists a \kc for $v$. Otherwise, $v$ will be deleted from the $R$.%, and its edges will be inserted into graph $G$ in Lines~\ref{lmin:remove} and~\ref{lmin:add}, respectively.

Theorem~\ref{theorem:min} establishes that Algorithm~\ref{alg:minimal} provides a minimal solution of \kpc.

\begin{theorem}
\label{theorem:min}
Algorithm~\ref{alg:minimal} \rev{returns} a feasible and minimal \kpc $R$.
\end{theorem}
\begin{proof}
Take note that Algorithm~\ref{alg:adp} generates the input vertex set $R$. $R$ is a viable \kpc, as \rev{its correctness has been proven}. Algorithm~\ref{alg:minimal} \rev{removes a vertex $v$ from $R$} only if the reduced $G$ contains no \kc. As a consequence, the \rev{result} set $R$ \rev{includes} all the \kcs after the termination. 

\rev{As for the minimality property}, if Algorithm~\ref{alg:minimal} does not prune each vertex $v$, there will exist a witness \kc $c_w$, where $c_w \cap (R - \{v\}) = \emptyset$. Given that the result set is $R_f \subset R$, then  $c_w \cap (R_f - \{v\}) = \emptyset$. Thus, if any vertex $v$ in the final result set $R_f$ is deleted, no vertex in $R_f - \{v\}$ will cover the witness \kc $c_w$. Thus, $R_f$ is a \kpc of $G$ that is both feasible and minimal.

\end{proof}

\vspace{1mm}
\noindent\textbf{Time and Space Complexities.}~\rev{Since there is no index and $FindCycle$ is a DFS algorithm, its space complexity is $O(m)$.} Following that, \rev{the time complexity is investigated}. Line~\ref{lmin:for} of Algorithm~\ref{alg:minimal} contains only one for-loop. Given that the size of $R$ is no larger than $n$, Line~\ref{min:for} requires at most $n$ iterations. In the worst-case scenario, the procedure $FindCycle$ requires $O(n^k)$. As a result, the overall time complexity is $O(n^{k+1})$.

%!TEX root = DamoGraph.tex
\section{Top-Down Algorithm}
\label{sect:tpd}
\rev{This section presents the Top-Down method in this section, which aims to improve the efficiency.}

%In this section, we introduce the Top-Down algorithm, which aims to improve the efficiency. Note that its cover size is also minimal.
%\subsection{Motivation}
%\label{subsec:SpdMotivation}
%In the previous section, although we could get the minimal \kpc , the algorithm runs slowly in our initial experiments.
%In the algorithm, the most time-consuming part is the cycle enumeration part. Hence, we mainly focus on speeding up this part. 

\subsection{Motivation}
\label{subsect:tpd_motivation}
\textit{Why Costly?}~The most expensive aspect of the \kpc problem is the repeated usage of the constrained cycle search.
This paper accelerates it from two aspects.
 \begin{itemize}
 \setlength{\itemindent}{0em}
\item \textit{Decrease the Search Space.}~The purpose of researching top-down algorithms is to reduce search space. We must verify all of the vertices. The search spaces in the bottom-up method range from the whole graph $G$ to the graph $G-R$, where $R$ is the cover set. In the $top-down$ algorithm, the search areas would range from $\emptyset$ to $G-R$. As a result, the top-down approach has the potential to substantially shrink the search space.
\item \textit{Increase the Speed of the Cycle Search Function.}~\rev{Since the cycle search function is frequently utilized in the whole process, it is one of the bottlenecks. Thus, this work proposes a delicate block-based and BFS-filter-based method to accelerate this process.}
%Is there a more efficient way to verify the node necessary for the same input graph? We offer block-based and BFS-filter methods to further accelerate the top-down algorithm in response to this problem.
 \end{itemize}

\subsection{\tpd Algorithm Description}
\label{subsect:al_des}
%The \rev{key idea} of \tpd algorithm is \rev{different} from that of the bottom-up algorithm. The method begins with an empty graph $G_0$ and a full cover set with all vertices in it. Then, \rev{it evaluates} each vertex $v$ in it. It determines whether or not to remove $v$ from the result set. \rev{If $true$, all in-edges and out-edges of $v$ are inserted into $G_0$}. \rrev{The main trick of the cycle validation part is adapted from~\cite{peng2019towards}. The general idea is when conducting a DFS search to validate whether there is a cycle including the query vertex. If a vertex $u$ is visited in the previous iteration and no cycle is found, then we could avoid visiting it until a cycle is found. As for the k-hop constraint cycle, it could not simply be avoided but set a block value. Every time $u$ is visited, its block value would increase by one. Thus, we visited it at most $k$ times, and every time it explored at most $m$ edges. Thus, the \tpd runs in $O(kmn)$, where $n$ is the number of vertices required for the validation. The correctness of this technique is not straightforward and we would prove it with detailed proof.}

The \rev{key idea} of \tpd algorithm is \rev{different} from that of the bottom-up algorithm. The method begins with an empty graph $G_0$ and a full cover set with all vertices in it. Then, \rev{it evaluates} each vertex $v$ in it. It determines whether or not to remove $v$ from the result set. \rev{If $true$, all in-edges and out-edges of $v$ are inserted into $G_0$}.
\rrev{The cycle validation algorithm is adapted from~\cite{peng2019towards}. The general idea is as follows: It conducts a DFS search to validate whether there is a cycle including the query vertex. In the DFS search, for each vertex $u$, if it has been searched before and we can guarantee that it is not included in a cycle with a certain length threshold, we can avoid searching $u$. The length threshold is at most $t$, and it is maintained during the search. In this way, the threshold value is updated at most $k$ times and every time it explores at most $m$ edges. Therefore, the \tpd algorithm runs in $O(kmn)$ time, where $n$ is the number of vertices to be validated.}

%The $top-down$ is a standard framework, but lower bounds ($block$) are carefully maintained in this section.

%\rrev{\textbf{Summary of \tpd}.~}
% \rrev{\textbf{Summary of \tpd}.~The main technique of cycle validation part is adapted from~\cite{peng2019towards}. To adapt it to the constrained cycle validation, we split the query vertex $a$ into $a_{in}$ and $a_{out}$. $a_{in}$ ($a_{out}$) is related with in-edges (out-edges) for $a$, respectively. Since the hop constrained path enumeration algorithm runs in O$(km\delta)$, if we detect one constrained cycle here, the cycle validation part would run in O$(km)$. Since there are $n$ vertices, the final algorithm runs in O(km). This is a general idea our \tpd, but the devil is in the detail.}

\rev{Algorithm~\ref{alg:top_down} Line~\ref{line:td_init} initializes the graph for verifying the node necessary.}
Line~\ref{min:for} is a for-loop that verifies all of the cover set's vertices.
Lines~\ref{line:tp_insert} and~\ref{min:findC} attempt to insert all of the vertex $v$'s edges and determine whether there \rev{exists} a constrained cycle.
If not, the vertex $v$ \rev{is deleted} from the cover set in Line~\ref{min:remove}.
Otherwise, vertex $v$ is \rev{maintained}, but all of its edges are removed in Line~\ref{min:delete}.
\begin{algorithm}[htb]
\SetVline
\SetFuncSty{textsf}
\SetArgSty{textsf}
\small
\caption{ \textsc{Top-Down}($G,k,R$)}
\label{alg:top_down}
%\Input
%{
%	$G$: the input graph \\
%	$k$: the hop constraint \\
%}
%\Output
%{
%$\mathcal{R}$: the cover set
%}
\State{$G_0 \leftarrow \emptyset$, $\mathcal{R} \leftarrow G$ }
\label{line:td_init}
\For{\bf{each} $v$ $\in$ $\mathcal{R}$}
{
\label{min:for}
	\State{Insert all in-edges and out-edges of $v$ into $G_0$}
	\label{line:tp_insert}
	\State{$c$ $\leftarrow$ \textsc{FindCycle}($G_0,k, v$)}
	\label{min:findC}
	\If{$c = \emptyset$ }
	{
	\label{min:update}
		\State{Remove $v$ from $\mathcal{R}$}
		\label{min:remove}
		%\State{$i = i -1$}
	}
	\Else
	{
		\State{Delete all in-edges and out-edges of $v$}
		\label{min:delete}
	}
}
\State{Return $\mathcal{R}$}
\end{algorithm}

\subsection{\rev{Node Necessary Validation}}
\label{subsec:nodeNec}
A frequent operation is to \rev{validate whether a vertex $v$ is in a constrained cycle in the current graph $G_0$.}
%determine whether or not a vertex $v$ should be deleted from the result set.
%Then, given a vertex $v$ and a directed graph $G_0$, the fundamental \rev{operation} is whether there exists \kc that contains vertex $v$.
A straightforward method is modified DFS \rev{BFS}. \rev{As for modified BFS}, Figure~\ref{fig:BFScts} \rev{demonstrates} counter-examples.
\begin{example}
\label{exam:mBFS}
%The following is a straightforward modified BFS:
Specifically, by executing BFS from vertex $v$ and assigning a new color to each neighbor. When \rev{it explores} a vertex further and \rev{locates} a vertex with a neighbor of a different color for the first time, \rev{it} discovered a shortest cycle through $v$.
Nevertheless, when beginning from vertex $a$, the modified BFS was unable to distinguish between Figure~\ref{fig:bfs_ct1} and~\ref{fig:bfs_ct2}. When the modified BFS is employed, $b$ and $c$ are marked as visited during the first iteration. Then, in the third iteration, \rev{it} returned to the edge $(d, c)$ of $c$. In such an instance, it is unable to determine if a constrained cycle begins at $a$ and produces the cycle $a \rightarrow c$ in Figure~\ref{fig:bfs_ct1}.
\end{example}

%Without blocking all the visited vertices and without modifying BFS, the worst-case time complexity is the same as that of the modified DFS.\footnote{If we regard bidirectional edges as cycles, then BFS could be used to test the node necessary.}

\begin{figure}[htb]
%	\vspace{-2mm}
	\newskip\subfigtoppskip \subfigtopskip = -0.1cm
	\newskip\subfigcapskip \subfigcapskip = -0.1cm
%     	\begin{minipage}[b]{\linewidth}
%		\centering
%		\includegraphics[width=0.5\linewidth]{main_graph_keys.eps}%
%	\end{minipage}	
	\centering
     \subfigure[\small{Counter-example 1.}]{
     \includegraphics[width=0.45\linewidth]{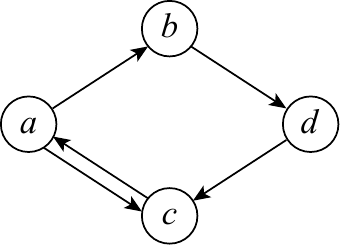}
	 \label{fig:bfs_ct1}
     }
     \subfigure[\small{Counter-example 2.}]{
     \includegraphics[width=0.45\linewidth]{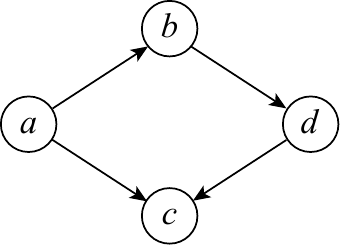}
	 \label{fig:bfs_ct2}
     }
%	\vspace{-3mm}
\caption{Counter-examples for the modified BFS.}
%	\vspace{-3mm}
%\vspace{-0.3cm}
\label{fig:BFScts}
\end{figure}

To accelerate \rev{\textit{Node Necessary Validation}}, an $O(km)$ time complexity method is proposed, \rev{which is inspired by the barrier technique in~\cite{peng2019towards}}. Firstly, the $block$ for a given vertex $u$ during the search procedure \rev{is formally defined}. \rev{It could be regarded as} the lower bound of $dis(u, s)$ in the current stack $S$. $s$ is the starting vertex. \rev{The $block$ is utilized to prune unnecessary candidates.}
%the condition is introduced, which shows $u.block$ is correct for a given vertex $u$.
\begin{definition}($u.block$)
\label{def:cor_blocks}
For a given vertex $u$, $u.block$ is \textit{correct} if and only if
given the current stack $\mathcal{S}$, if there is a path $p(u \rightarrow s)$, not containing any vertex in $\mathcal{S}$, we have $len(p) \geq u.block$, i.e., $sd(u,s| \mathcal{S}) \geq u.block$.
\end{definition}

% Definition~\ref{def:cor_blocks} implies that  $u.block$ is \textit{correct} in two cases:
% (1) $u \in \mathcal{S}$;
% \textit{or} (2) given the current stack $\mathcal{S}$, if there is a path $p(u \rightarrow s)$, not containing any vertex in $\mathcal{S}$, we have $len(p) \geq u.block$, i.e., $sd(u,s| \mathcal{S}) \geq u.block$.

\begin{figure}[htb]
	\centering

     \includegraphics[width=0.6\linewidth]{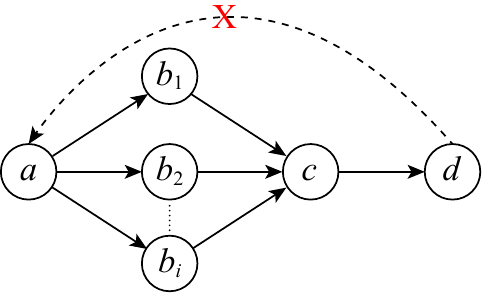}

%\vspace{-3mm}
\caption{\small An example of the block idea.}
%\vspace{-3mm}
\label{fig:ex_block}
\end{figure}

\rev{An example of the \textit{block} is illustrated in Figure~\ref{fig:ex_block}}.
%illustrates the \textit{block} method in \kpc.
\begin{example}
Assume \rev{vertex $a$ is the starting vertex with $k = 5$}. \rev{With a DFS search, it validates} whether there \rev{exists} a constrained cycle \rev{containing} vertex $a$.
\rev{In the first iteration, path} $a \rightarrow b_1 \rightarrow c \rightarrow d$ \rev{is checked}. \rev{In this path,} vertex $c$ cannot reach vertex $a$ in $5-2=3$ hops.
%A supposition occurs:
\rev{Then, $c.block$ is set to $3+1=4$.} This block information may be utilized in the subsequent DFS exploration, e.g., $a \rightarrow b_2 \rightarrow c \rightarrow d$. It may end prematurely when exploring $a \rightarrow b_i \rightarrow c$. \rev{In this path, there are $5-2=3$ remaining \rev{depths} for DFS, which is smaller than $c.block$.}
\end{example}

\rev{The idea of $block$ is to utilize the failure information to prune invalid search space. The algorithm is formally described as follows:} \rev{$S$ is used to denote the stack of the currently explored path}. $S_{old}$ indicates the past exploration path, whereas $S_{new}$ denotes the current exploration path.
% When $S_{old} \subset S_{new}$, the hypothesis that the block values are correct is self-evidently $true$. Nevertheless, when $S_{old}$ is not a subset of $S_{new}$, the correctness of block value is non-trivial. Assume they are correct, and that the method is as described in Algorithm~\ref{alg:node_blocks}. 
\rev{Algorithm~\ref{alg:node_blocks} illustrates the whole process.}
%\textbf{Step 1}. We use a DFS-oriented cycle-finding algorithm, we could block vertex $v$ with 

%\scalebox{0.95}{
%\begin{minipage}{1.0\linewidth}
	\begin{algorithm}[htb]%[H]%[htb]%{\textwidth}
\SetVline
\SetFuncSty{textsf}
\SetArgSty{textsf}
\small
\caption{ \textsc{NodeNecessary}($s$, $u$, $\mathcal{S}, R,G'$)}
\label{alg:node_blocks}

%\begin{adjustbox}{width=0.8\columnwidth,center}{

%\Input{
%$s$: the start vertex;\\
%$u$: the vertex to be pushed; \\
%$\mathcal{S}$: the stack used for DFS search;\\
%$\mathcal{R}$: the result cycle set;\\
%$G'$: the graph $G'$.
%}
%\Output{
%$\mathcal{R}$ : the result cycle set. Not empty indicates a true-query, otherwise a false-query. 
%}
\If{$\mathcal{R} \neq \emptyset$}
{
	\label{line:block_trueR}
	\StateCmt{Return R}{Vertex $s$ is necessary}
}

\State{$u.block \leftarrow  k - len(\mathcal{S})+1$ }%{Set the block level.}
\label{line:blockLevel}
\State{ $\mathcal{S}$.push($u$)}
\If {$u=s$ $\land$ $len(\mathcal{S}) \geq 2$}
{
  \label{line:find_true}
  \State{ $u$ is unstacked from $\mathcal{S}$}
  \State{\textsc{Unblock}($\mathcal{S}.top()$, $\mathcal{S}$, $1$)}
  \If{$len(\mathcal{S}) > 2$}
  {
  	\label{line:largerTwo}
  \State{insert $p(\mathcal{S})$ into $\mathcal{R}$}%where $P_k := P_k \cup p(\mathcal{S})$}
  \label{alg:node_blocks_output}
  \Return{ $\mathcal{R}$}
  }
}
\ElseIf{ $len(\mathcal{S}) < k$}
{

  \For{ vertex $v$ of $adj_{out}[u]$ where $v \not \in (\mathcal{S}-\{ s \})$ }
  {
      \If{ $len(\mathcal{S})+1+v.block \leq k$ }
      {
      \label{line:blockRule}
          \State{ $\mathcal{R}$ $\leftarrow$ \textsc{NodeNecessary}($s$,$v$, $\mathcal{S}, \mathcal{R},G'$)}
          \If{$\mathcal{R} \neq \emptyset$}
	{
	\label{line:block_trueR2}
		\State{Return $\mathcal{R}$}%{Vertex $s$ is necessary}
	}
      }
  }
}
\State{ $u$ is unstacked from $\mathcal{S}$ }
\label{alg:node_blocks_unstack}
\Return {$\mathcal{R}$}

%}
\end{algorithm}
%    \end{minipage}%
%}

\begin{algorithm}[htb]
\SetVline
\SetFuncSty{textsf}
\SetArgSty{textsf}
\small
\caption{ \textsc{Unblock}($u$,$\mathcal{S}$, $l$)}
\label{alg:unblock}
%\Input{
%$u$: the vertex to be updated; \\
%$\mathcal{S}$: the stack for the DFS search\\
%$l$: the number of hops from $u$ to the target vertex $t$,
%     not touching any vertex in $\mathcal{S}$
%}
\State{ $u.block = l$ }
    \For{ {\bf each} vertex $v$ of $adj_{in}[u]$ with $v \not \in \mathcal{S}$}
    {
    \If { $v.block > l+1$ }
	{
        	\State{ \textsc{Unblock}($v$, $\mathcal{S}$, $l + 1$)}
           }
    }
\end{algorithm}

%to do here

\noindent\textbf{\rev{Details}.}~\rev{Algorithm~\ref{alg:node_blocks} Line~\ref{line:block_trueR} terminates} the recursive algorithm \rev{if} a valid constrained cycle is \rev{found}.
Line~\ref{line:blockLevel} initializes the block \rev{value} associated with the current vertex $u$ \rev{to $k - len(\mathcal{S} + 1)$}.
Line~\ref{line:find_true} is the condition \rev{for} a valid \kc. 
%Without the requirement $len(\mathcal{S}) \geq 2$, the algorithm terminates at the source vertex $s$.
Line~\ref{line:largerTwo} determines whether it is a \rev{valid} constrained cycle. 

If $true$, the cycle \rev{will} be \rev{inserted} into the result set $\mathcal{R}$.
%\rev{Thus, it terminates the process.}
\rev{Line~\ref{line:blockRule} is the condition for blocking.} When $len(\mathcal{S})+1+v.block > k$, vertex $v$ is \rev{blocked}. \rev{The worst-case time complexity is $O(km)$ due to the block level of each vertex ranging from $0$ to $k$.}
\rev{Algorithm~\ref{alg:unblock} iteratively updates the block values for the vertices whose block is larger than $l$.}

\noindent\textbf{Modification to Cycle Cover without Constraints}.~\rev{To cope with the variant without hop constraint, we only need to modify the node necessary function. The modification can be summarized as the following steps}:
    \begin{itemize}
        \item \rev{i) Replace $u.block \leftarrow  k - len(\mathcal{S})+1$ to $u.block \leftarrow \infty$ in Line 3 Algorithm 9.}
        \item \rev{ii) Remove condition $len(\mathcal{S}) < k$ in Line 11 Algorithm 9.}
        \item \rev{iii) Replace $len(\mathcal{S})+1+v.block \leq k$ with $v.block \neq \infty$ in Line 13 Algorithm 9.}
    \end{itemize}

%Here is an example how this algorithm runs and then
%\rev{Since the correctness of the block is not straightforward, the proof is given}
%The following part contains the proof for the correctness of this method.

\subsection{Analysis}
\label{subsec:ana_blocks}
This part \rev{proves} the correctness of Algorithm~\ref{alg:node_blocks}. 
To begin, \rev{the} condition \rrev{of} \rev{the correct} $u.block$ \rev{value} \rev{is given}.
\begin{lemma}
\label{lemma:cor_blocks}
\rev{$u.block$ is \textit{correct} iff given the stack $\mathcal{S}$, there is a path $p(u \rightarrow s)$ without any vertex in $\mathcal{S}$, $len(p) \geq u.block$, i.e., $sd(u, s| \mathcal{S}) \geq u.block$.}
\end{lemma} 

% Lemma~\ref{lemma:cor_blocks} indicates that  $u.block$ is \textit{correct} \rev{for two scenarios}:
% \begin{itemize}
%     \item (1) $u \in \mathcal{S}$; 
%     \item (2) if there is a path $p(u \rightarrow s)$ \rev{without and} vertex in the current stack $\mathcal{S}$, \rev{s.t.} $len(p) \geq u.block$, i.e., $sd(u, s| \mathcal{S}) \geq u.block$.
% \end{itemize}

The \textit{budget} of a vertex $u$ is defined as follows:
\begin{definition}[budget]
\label{def:budget}
\rev{budget[u]} is the number of hops \rev{remaining} for $u$ to continue the search.
\begin{equation}
    budget[u] = k - len(\mathcal{S}), when \ u = S[0]
\end{equation}
\end{definition}
  %Given a path $p(u \rightarrow s)$, $p[x]$ to denote the position of $x$ in the path $p$ with $p[u]=len(p)$ and $p[t]=0$, i.e., $p[x]$ is the number of hops $x$ can reach $t$ along the path $p$.
The following lemma \rev{gives} the condition under which a vertex $u$ may reach the target vertex $s$ in Algorithm~\ref{alg:node_blocks}.

\begin{lemma}
\label{lemm:c}
Assume that the top vertex of the current stack $\mathcal{S}$ is $u$. There is a path $p(u \rightarrow s)$. The vertex $u$ could reach the vertex $s$ in Algorithm~\ref{alg:node_blocks} only if $k -len(\mathcal{S}) \geq len(p)$ and every vertex in the path (except $u$) is not included by $\mathcal{S}$.
\end{lemma}
\begin{proof}
The inequation $k - len(\mathcal{S}) \geq len(p)$ indicates that the vertex $u$ has a sufficient \textit{budget} to reach $t$.
\rev{Since $\mathcal{S}$ does not include all vertices $\{x\}$ along the path, the search can only be early terminated due to their block values.}
Because $x$ is not in the $\mathcal{S}$ and $x.block$ is correct w.r.t $\mathcal{S}$, $x.block \leq p[x]$, as $x$ can reach $t$ within $p[x]$ hops. \rev{Thus,} $x$ cannot use the block \rev{value} to \rev{terminate} the search.
Thus, $u$ can reach $t$ \rev{in Algorithm~\ref{alg:node_blocks}}.
\end{proof}

\rev{The correctness of Algorithm~\ref{alg:node_blocks} is proved by demonstrating \rrev{that} block values are correct throughout it.}
%The proof outline is illustrated in Figure~\ref{fig:proofProcess}.

\begin{theorem}
\label{the:value}
All the block values are correctly computed, and they remain correct \rev{in} Algorithm~\ref{alg:node_blocks}.
%For any vertex $u$, $u.block$ is correctly maintained in Algorithm~\ref{alg:node_blocks}.
\end{theorem}
\begin{proof}
\rev{Firstly}, $u.block$ is correct \rev{if $u \in \mathcal{S}$}.
When $u.block$ \rev{is set} at Line~\ref{line:blockLevel}, the value is correct w.r.t $\mathcal{S}$.
%due to the thorough exploration of $\mathcal{S}$'s search space.

\rev{It is demonstrated that $u.block$ is properly specified in the following search.}
If a new vertex $v$ is pushed into $\mathcal{S}$, then $u.block$ is immediately \rev{correct} since $\mathcal{S} = \mathcal{S} \cup \{v\}$ \rev{leads to} a strictly smaller search space.

\rev{Consequently}, the sole remaining scenario is \rev{vertexs' unstacking}.
The vertex $v$ denotes the first vertex \rev{that leads} $u.block$ \rev{to be incorrect}.
If \rev{unstack} of $v$ \rev{does not affect} $u.block$, $u.block$ \rev{is still correct} for the new stack $\mathcal{S} \setminus \{v\}$.
Alternatively, $u.block$ may be updated in \rev{two cases}:

\begin{itemize}
\vspace{1mm}
\item \noindent (1) If $v=u$, \rev{$u.block$ is still $true$}. \rev{When a valid \kc containing $u$ exists}, the algorithm \rev{terminates}. \rev{Then}, $u$ \rev{will not be unstacked due to the early termination}.
Thus, the \rev{unstack} of vertex $u$ \rev{indicates} that the current $\mathcal{S}$ does not include \rev{any} valid \kc.
%In the case that there is a  
%Note that $u.block$ can only be increased under this case.

\vspace{1mm}
\item \noindent (2) Given $v \not = u$ and that the \rev{unstack} of $v$ impacts the $u.block$.
\rev{Thus, there exists a path $p(u \rightarrow v \rightarrow s)$, which does not include any vertex in $\mathcal{S}$, s.t. $len(p) < u.block$}. \rev{That is}, $u$ can reach $t$ with fewer hops \rev{due to} $v$'s \rev{unstack} from $\mathcal{S}$.
\rev{Assume} that vertices' block \rev{values} are properly \rev{maintained before} \rev{vertex $v$'s unstack}.
\end{itemize}

\rev{Assume} $v$ cannot reach $t$ in Algorithm~\ref{alg:node_blocks}, \rev{and} $u.block > len(p)=p(u)$ with respect to the current stack $\mathcal{S}$. \rev{It indicates that although} $v.block$ was correctly \rev{maintained} in the \rev{previous} step, it is \rev{incorrect} due to unstack of $v$ and $u.block > p(u)$.

%Here are three vital timestamps during the proof.
Three \rev{vital} timestamps occurred throughout the proof. \rev{They are} $T_{inU}$, $T_{outU}$, and $T_{outV}$. $T_{inU}$ indicates the time \rev{when} $u$ is added into the stack.
$T_{outU}$ and $T_{outV}$ indicate \rev{the time when} the first\footnote{A vertex may be pushed and unstack for many times.} unstack of $u$ and $v$, respectively, after $T_{inU}$.
$S_0(y)$ indicates the stack size when the vertex $y$ is pushed to the stack $S_0$.
\rev{Note that} $T_{inU} < T_{outV} < T_{outU}$, and $S_0(v) < S_0(u)$.

If $v$ is the sole vertex that \rev{blocks} $u \rightarrow s$ at $T_{inU}$, the algorithm will terminate before $T_{outV}$. \rev{Note} that $S_0(v) < S_0(u)$.
According to Lemma~\ref{lemm:c}, the vertex $v$ can reach the vertex $s$, and according to Algorithm~\ref{alg:node_blocks} Line~\ref{line:block_trueR} and~\ref{line:block_trueR2}, the unstack of $v$ is \rev{early terminated}.
It violates the \rev{assumption} that there \rev{exists} unstack of $v$. 
The proof is similar if there \rev{exist} more \rev{than one} vertices.%\rev{The proof} concentrate on the stack's lowest vertex, which blocks $u \rightarrow s$.
\end{proof}

\vspace{1mm}
\noindent \textbf{Correctness.}
According to Theorem~\ref{the:value}, the vertices' block \rev{values} are correct.
Algorithm~\ref{alg:node_blocks} is a hop-constrained DFS that utilizes a block-based \rev{technique}. If the hop-constrained DFS and the block-based method are valid, the algorithm is correct. Notably, Algorithm~\ref{alg:node_blocks} guarantees the simple cycle property by default.

\vspace{1mm}
\noindent \textbf{Time Complexity.}
\rev{The following theorem indicates that Algorithm~\ref{alg:node_blocks} is $O(km)$}. \rev{The time complexity of \tdb is $O(k \cdot m \cdot n)$}.

\begin{theorem}
Algorithm~\ref{alg:node_blocks} is capable of \rev{returning} a valid answer in $O(km)$.
%\bcdfs is a polynomial delay algorithm with $O(km)$ time per output. The time complexity of \bcdfs is $O(km\delta)$,
%where $\delta$ is the number of \hcst paths.
\end{theorem}
\begin{proof}
Assume that Algorithm~\ref{alg:node_blocks} unstacked a vertex $u$ twice.
This implies that none of these two unstacks is \rev{early} terminated.
Let $S_1$ and $S_2$ \rev{denote} the \rev{stacks} after the first and second times that $u$ is pushed into the stack, respectively.
After \rev{unstack} $u$ for the first time, $u.block = k - S_1 +1$.
%As there is no early termination, the propagation of block values will not be invoked.
When $u$ is pushed to the stack \rev{at the} second time, $u.block$ \rev{remains} unchanged.
As $u$ passes \rev{block conditions} in the second visit \rev{with} $S_2 + u.block \leq k$,
and \rev{thus} $S_2 < S_1$. \rev{Consequently}, $u.block$ will be increased by at least one \rev{every time} $u$ is unstacked.%, even if no new output is generated.

This \rev{indicates} that a vertex may be pushed to stack no more than $k$ times. When $u$ is \rev{added} into the stack, an edge $(u,v)$ is visited. \rev{An edge is visited at most $k+1$ times}. \rev{Hence, the time complexity of it is $O(km$}. It is worth noting that omitting the bidirectional edges as cycles has no effect on the time complexity. Assume the start vertex is $v$, and $u, w$ are its bidirectional out-neighbors. Assume that the unblock of $u$ and $w$ would have an effect on the vertex sets $A_u$ and $A_w$.
Then either $A_u \cap A_w = \emptyset$, or there exists a constrained cycle when $u$ and $w$ are explored. Both of them show that Algorithm~\ref{alg:node_blocks} has an $O(km)$ time complexity.
\end{proof}

\vspace{1mm}
\noindent \textbf{Space Complexity}.~\rev{Since the stack size is always bounded by $k$}, \rev{the} space complexity of Algorithm~\ref{alg:node_blocks} is $O(m+k)$.

\begin{theorem}
Algorithm~\ref{alg:top_down} \rev{produces} a \kpc $R$ that is both \rev{\textit{valid}} and \rev{\textit{minimal}}.
\end{theorem}
The proof is \rev{similar} to the \rev{proof} of Theorem~\ref{theorem:min}.

\noindent \textbf{Comparison with Barrier Technique}.~\rev{Firstly, the hop-constrained path enumeration problem focuses on how to efficiently enumerate all the paths, but in the constrained cycle cover problem, the point is how to efficiently detect the existence of any constrained cycle. There are many algorithms in the problem of hop-constrained path enumeration, IDX-DFS~\cite{sun2021pathenum}, IDX-JOIN~\cite{sun2021pathenum}, PathEnum~\cite{sun2021pathenum}, BC-DFS~\cite{peng2019towards}, T-DFS~\cite{rizzi2014efficiently}, T-DFS2~\cite{grossi2018efficient}, and JOIN~\cite{peng2019towards}. IDX-DFS, IDX-JOIN, PathEnum, and JOIN are more efficient than BC-DFS in terms of hop constrained path enumeration, but their technique is not suitable to adapt to the constrained cycle cover problem due to \rev{the} different focuses of these two different problems. They either need preprocessing costs to construct a light-weighted index~\cite{sun2021pathenum} or find the middle cut (JOIN~\cite{peng2019towards}) to accelerate the whole enumerate process. In the context of hop-constrained path enumeration, such cost is affordable since the bottleneck is the heavy enumeration stage due to a large number of results. Nevertheless, for the constrained cycle cover problem, only one cycle is needed, then the bottleneck is altered. }

\subsection{Upper Bounds Filtering}
\label{subsec:filter}
\rev{This subsection introduces an upper bound to filter some unnecessary vertices.}
%rewrite this part as the initial bounds for the block idea
%\subsection{\bft}
%\label{subsec:filter}

\vspace{1mm}
\noindent\textit{\bft.}~According to Example~\ref{exam:mBFS}, a modified BFS could not examine whether a vertex is necessary in the \textit{Top-Down} algorithm. Nevertheless, for a vertex $v$, if its distance to itself is larger than $k$ in the modified BFS, then it could be safely excluded in the current iteration. 
%The \textit{Distance Filtering} method is discussed in this subsection. 
%\begin{theorem}\textit{Distance Filtering}

%\end{theorem}
%We introduce the \bft in this subsection. In the previous subsection, we mention that the BFS could not be used due to a counter-example.
%Although BFS could not be used to solve this problem, we use it as a filter rule.
\rev{\noindent\textbf{Details}.~The details are presented as Algorithm~\ref{alg:bfs_filter}.}
\begin{algorithm}[htb]
\SetVline
\SetFuncSty{textsf}
\SetArgSty{textsf}
\small
\caption{ \textsc{BFS-filter}($G_0,k,v$)}
\label{alg:bfs_filter}
%\Input
%{
%	$G$: the input graph \\
%	$k$: the hop constraint \\
%	$v$: the tested vertex
%}
%\Output
%{
%$True$ or $False$: whether vertex $v$ is pruned in this proces.
%}
\State{$\mathcal{U}$ $\leftarrow$ the upper bound distance from $v$ to $v$ using the modified BFS}
\label{line:bfs}
\If{$\mathcal{U} > k$}
{
\label{line:prune}
	\State{Prune vertex $v$}
}
%\ElseIf{$\mathcal{U} = 2$}
%{
%\label{line:bfs_true}
%	\State{Vertex $v$ needs further verify}
%}
\Else
{
\label{line:bfs_false}
	\State{Vertex $v$ needs further verify.}
}
\end{algorithm}
Line~\ref{line:bfs} computes the upper bound of the distance of $v$ to itself using the modified BFS (see Example~\ref{exam:mBFS} for details) and represents it with $\mathcal{U}$. Two cases for the \bft method are shown in Lines~\ref{line:prune} to~\ref{line:bfs_false}. When $\mathcal{U} > k$, the vertex is pruned safely. Line~\ref{line:bfs_false} represents the situation where $\mathcal{U} \leq k$. Then, the vertex has to be verified further using Algorithm~\ref{alg:node_blocks}.

\section{Experimental Results}
\label{sect:experiment}
\rev{This section evaluates the effectiveness and efficiency of the proposed techniques on comprehensive experiments.}

\subsection{Experimental Settings}
\label{sect:settings}

\vspace{1mm}
\noindent\textbf{Compared Algorithms.}~\rev{The following baselines are compared in the experimental part.}
%We compare proposed algorithms with baseline solutions.

\begin{itemize}
\item \textit{\darc}. The state-of-the-art algorithm~\cite{kuhnle2019scalable} introduced in Section~\ref{subsec:baseline}.% We implement the vertex version of it in the directed graphs. The DV indicates the directed vertex version. 
%\item \textbf{\kapp}. The \underline{k-app}roximate algorithm introduced in Section~\ref{subsec:baseline}.%(or KAPP)
%\item \textbf{\rand}. Similar to \kapp, but only randomly choose a vertex instead of choosing all vertices.%(or RAND)
%\item \textbf{\degr}. Similar to \rand, but choose the vertex with the highest degree on the found cycle.
%\item \textbf{\adp}. The \underline{ad}a\underline{p}tive cycle cover algorithm introduced in Section~\ref{subsec:Adaptive}.
\item \textit{\adp}. The bottom-up approach introduced in Section~\ref{subsec:Adaptive}.
\item \textit{\adpm}. The bottom-up approach with the minimal technique introduced in Section~\ref{subsec:minimal}.
\item \textit{\tbs}. The \underline{T}op-\underline{D}own \underline{B}locks algorithm introduced in Section~\ref{sect:tpd}.
\item \textit{\tbk}. The Top-Down {B}locks algorithm with the \bt introduced in Section~\ref{sect:tpd}.
\item \textit{\tdb}. The Top-Down Blocks algorithm with the block and BFS-filter techniques introduced in Section~\ref{sect:tpd}.

%\item \textbf{EDP}.~Edge-Disjoint Partitioning(EDP)~\cite{hassan2016graph} is the state-of-the-art algorithm to solve the label constrained shortest path problem. In \cite{hassan2016graph}, they partition a graph by their labels and construct the index in response to the queries received.
%%We note that LCSP is more general than LCR, and constructing an index for such a problem is a more challenging and significantly different task. ~\cite{yuan2019constrained} proposed algorithms to solve the constrained shortest path queries in a time-dependent graph.
%\item \textbf{\oursns}.~Our upper and lower bounds based label constrained k-reach algorithm proposed in Section~\ref{subsect:bounds}.
%\item \textbf{\oursExactns}.~Our upper and lower bounds based exact label constrained shortest path distance algorithm proposed in Section~\ref{subsect:query}.
%\item \textbf{\thop}.Our proposed LC 2-hop index algorithm for LCR queries introduced in Section~\ref{subsect:Construction}.
%\item \textbf{\thopp}.Our proposed LC 2-hop index algorithm with advanced BFS Search Order for LCR queries introduced in Section~\ref{sect:exploreOrder} \footnote{\thopp uses degree order which is the same as LI+ for the sake of comparison fairness. The effect of node order will be compared in Table~\ref{tb:nodeOrder}.}.
\end{itemize}

\vspace{1mm}
\noindent\textbf{Datasets.}~Table~\ref{tb:datasets} summarizes the key statistics about the real graphs used in the experiments. Most of these graphs are from either SNAP \cite{leskovec2016snap} or KONECT \cite{kunegis2013konect}. \\

\vspace{1mm}
\noindent\textbf{Settings.}~All programs were implemented in standard C++11 and compiled using G++4.8.5.
\\
All experiments were performed on a machine with 36X Intel Xeon 2.3GHz and 385GB main memory running Linux (Red Hat Linux 7.3 64 bit).

\begin{table}
  \centering
      \caption{Statistics of datasets. K indicates $10^3$. M indicates $10^6$. B indicates $10^9$.}
%\vspace{-3mm}
\label{tb:datasets}
%\resizebox{0.48\textwidth}{!}{
%\resizebox{0.48\textwidth}{34mm}{
    \begin{tabular}{lrrrr}%{|c|c|ccc|}
      \hline
      % after \\: \hline or \cline{col1-col2} \cline{col3-col4} ...as-caida20071105
       \cellcolor{gray!25}\textbf{Name}	& \cellcolor{gray!25}\textbf{Dataset} 	&  \cellcolor{gray!25}\textbf{$\lvert V \rvert$}	 &  \cellcolor{gray!25}\textbf{$\lvert E \rvert$} 	&  \cellcolor{gray!25}\textbf{$d_{avg}$} \\ \hline %\hline %& \textbf{$d_{max}$}
	WKV		&  Wiki-Vote		&	7K	 & 	104K	& 	29.1			   \\ %\hline	
	ASC		&  as-caida&	26K	 & 	107K	& 	8.1			   \\ %\hline
	GNU		&  Gnutella31	&	63K	 & 	148K	& 	4.7			   \\ %\hline
	EU		& Email-Euall 		&	265K	 & 	420K	& 	3.2			   \\ %\hline
	%EPIN	& soc-Epinions1 		&	75K	 & 	508K	& 	3.2			   \\ %\hline
	%SEP		& soc-Epinions	&	76K	 & 	508K	& 	13.4			   \\ %\hline
	SAD		& Slashdot0902	&	82K	 & 	948K	& 	23.1			   \\ %\hline
	%DAU		& dblp-author 	&	317K	 & 	1.05M	& 	6.6			   \\ %\hline

	WND		& web-NotreDame 	&	325K	 & 	1.5M	& 	9.2			   \\ %\hline

	CT		& citeseer 		&	384K	 & 	1.7M	& 	9.1			   \\ %\hline
	WST		& webStanford	&	281K	 & 	2.3M	& 	16.4			   \\ %\hline
	LOAN		& prosper-loans	&	89K	 & 	3.4M	& 	76.1			   \\ %\hline

	WIT		& Wiki-Talk		&	2.4M	 & 	5.0M	& 	4.2			   \\ %\hline

	WGO		& webGoogle 	&	875K	 & 	5.1M	& 	11.7			   \\ 
		WBS		& webBerkStan 	&	685K	 & 	7.6M	& 	22.2			   \\ 
		FLK		& Flickr 	&	2.3M	 & 33.1M	& 			28.8	   \\ 
			LJ		& LiverJournal 	&	10.6M	 & 	112M	& 	21.0			   \\ 
		WKP		& Wikipedia 	&	18.2M	 & 	172M	& 	18.85			   \\ 
	TW		& Twitter(WWW) 	&	41.6M	 & 1.47B	& 	70.5			   \\ \hline
	
       \end{tabular}
%}
%\vspace{-2mm}

\end{table}

\vspace{1mm}
\begin{table}[htbp]
  \centering
      \caption{The cover size (the number of vertices) and runtime (seconds) for algorithms when $k = 5$.}%``K" indicates $10^3$.}
%\vspace{-1mm}
\label{tb:CoverTime}
 %\resizebox{\textwidth}{7mm}{
%\resizebox{\textwidth}{!}{
\scalebox{0.9}{
%\resizebox{1.0\textwidth}{38mm}{
    \begin{tabular}{l|rr|rr|rr}%{|c|cc|cc|cc|cc|}
      \hline
      	%\textbf{Name} 	& \multicolumn{2}{c|}{\adpm}	& \multicolumn{2}{c|}{\kapp }   & \multicolumn{2}{c|}{\degr}	& \multicolumn{2}{c|}{\rand }\\  % \textbf{$d_{max}$}
      	\cellcolor{gray!25}\textbf{Name} 	& \multicolumn{2}{r|}{\cellcolor{gray!25}\darc } & \multicolumn{2}{r|}{\cellcolor{gray!25}\adpm}& \multicolumn{2}{r}{\cellcolor{gray!25}\tdb }\\  % \textbf{$d_{max}$}
		\cellcolor{gray!25}	 	&	\cellcolor{gray!25}Size		& \cellcolor{gray!25}Time		&	\cellcolor{gray!25}Size		& \cellcolor{gray!25}Time	&	\cellcolor{gray!25}Size		& \cellcolor{gray!25}Time	    		   \\ \hline
	WKV	 	&490 			& 53.8			&	\textbf{469}	 	& 402.8	 &	 	 491		& 	\textbf{0.41}         		   \\ 
	ASC	 		&620 & 2.42		&	\textbf{607}	 & 44.01	 	&		 612		& 	 	\textbf{0.11}		   \\ 
	GNU	 	&184 	& 1.3		&	\textbf{180}	 	& 1.49	&	  193	 & \textbf{0.69} 			   \\ 	
	EU	 	&622 	& 114.7		&	\textbf{609}	 	& 702.1	&	 627	 & \textbf{1.25} 			   \\ 	
	
	%EPIN	 	&5,329 	& 28.3		&	\textbf{5,081}	 	& 1,127.6	&	 5,150	 & \textbf{1.04} 			   \\ 	
		
	%SEP	 	&\textbf{11,913}& 578		&	18,398 	& 		 &	12,431 	& 	240	      		   \\ 
	SAD	 		&6,377 & 440.1	&	\textbf{6,005}	& 4,717	&		6,380	& 	 	\textbf{3.13}     		   \\ 	
	%DAU	 	&\textbf{63}	& 60.5		&	270	 	& 76.4	 &	63		& 	\textbf{60.2}	&	64	 	& 	62.9        		& &   \\

	WND	 	&27,067& 29,916.8		&	\textbf{23,853}	 & 28,953.3	 &	 	24,290& \textbf{2.67}	 		   \\ 
 
	CT	 		&1,621	& 37.03		&	\textbf{1,610} 	& 43&	1,611 	& 	\textbf{16.2}  	   \\ 	
	WST	 		&31,253		& 140.7		&	\textbf{30,811} 	& 275.6	 	&	 	31,148	& 		\textbf{2.99}	    		   \\ 
	LOAN	 	&332	& 184.5		&	\textbf{320}	 	& 450.7&	 347 &			 \textbf{127.9}	      		   \\

	WIT	 		&7,040& 2,296.8	&	6,923	 & 4,708.3	 &		\textbf{6,894}	& 	\textbf{56.3}	      		   \\ 	
 	
	WGO	 	&130,382		 & 42.2		&	\textbf{129,009} 	& 110.8	 &	129,421& \textbf{5.99} 	   \\ 
		WBS	 	&98,570 & 3,571.4	&	\textbf{94,817} & 12,739	 &	 100,668		&\textbf{6.96} 		    		   \\ 
		FLK	 	&- 		& -				&	- 	& -	 &	\textbf{206,912} 	& 	\textbf{92.3}		   \\ 
	LJ	 		&- 		& -				&	- 	& -	 &	\textbf{39,183} 	& \textbf{20,466.8}		   \\ 
	WKP	 	&- 		& -				&	- 	& -	 &	\textbf{685,759} 	& 	\textbf{4,132}		   \\ 
		TW	 	&- 		& -				&	- 	& -	 &	\textbf{3,731,522} 	& 	\textbf{89,634}	   \\ \hline
       \end{tabular}

}
%}
	%\vspace{-1mm}

\end{table}

\begin{figure*}[htbp]
	\vspace{-2mm}
	\newskip\subfigtoppskip \subfigtopskip = -0.1cm
	\newskip\subfigcapskip \subfigcapskip = -0.1cm
%     	\begin{minipage}[b]{\linewidth}
%		\centering
%		\includegraphics[width=0.5\linewidth]{main_graph_keys.eps}%
%	\end{minipage}	
	\centering
     \subfigure[\small{WKV}]{
     \includegraphics[width=0.23\linewidth]{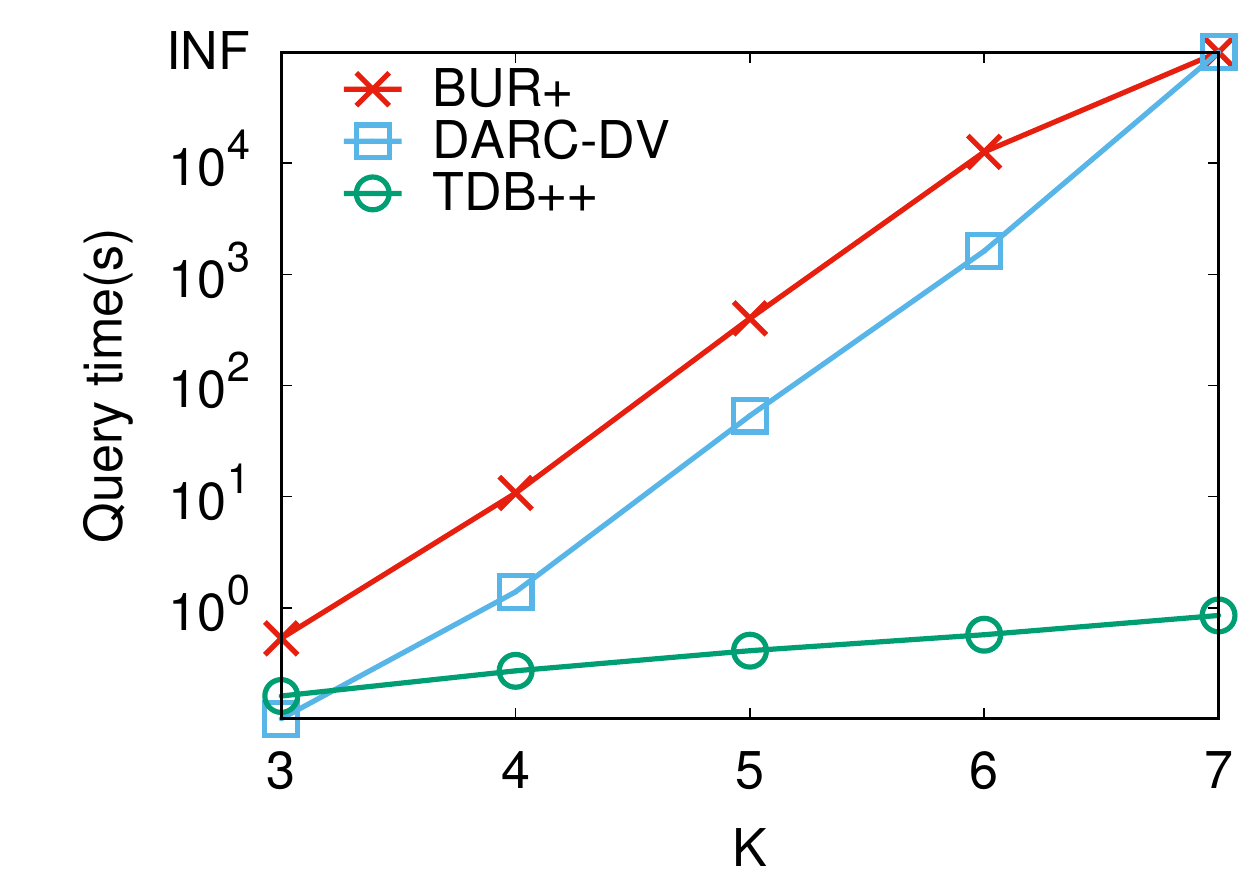}
	 \label{fig:ITL}
     }
     \subfigure[\small{ASC}]{
     \includegraphics[width=0.23\linewidth]{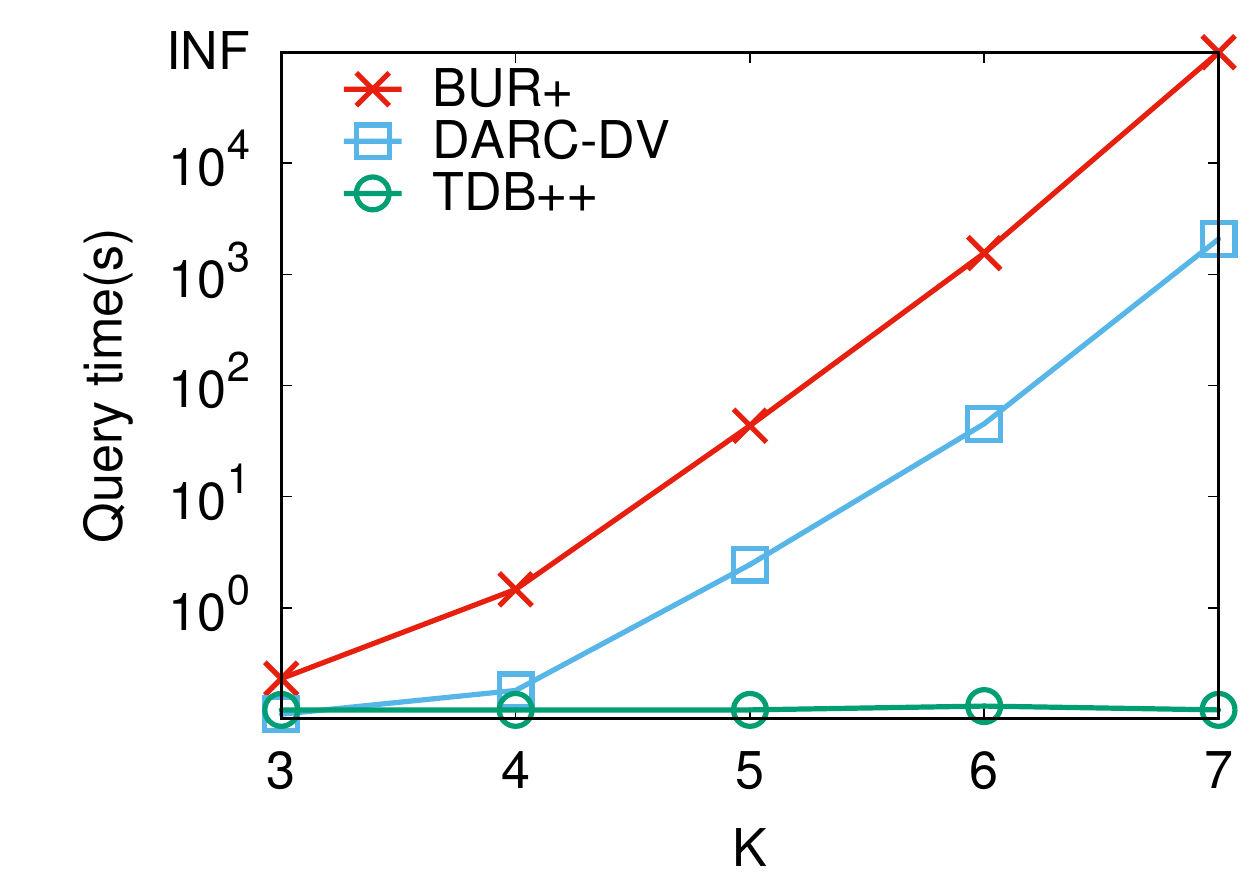}
	 \label{fig:ITW}
     }
          \subfigure[\small{GNU}]{
     \includegraphics[width=0.23\linewidth]{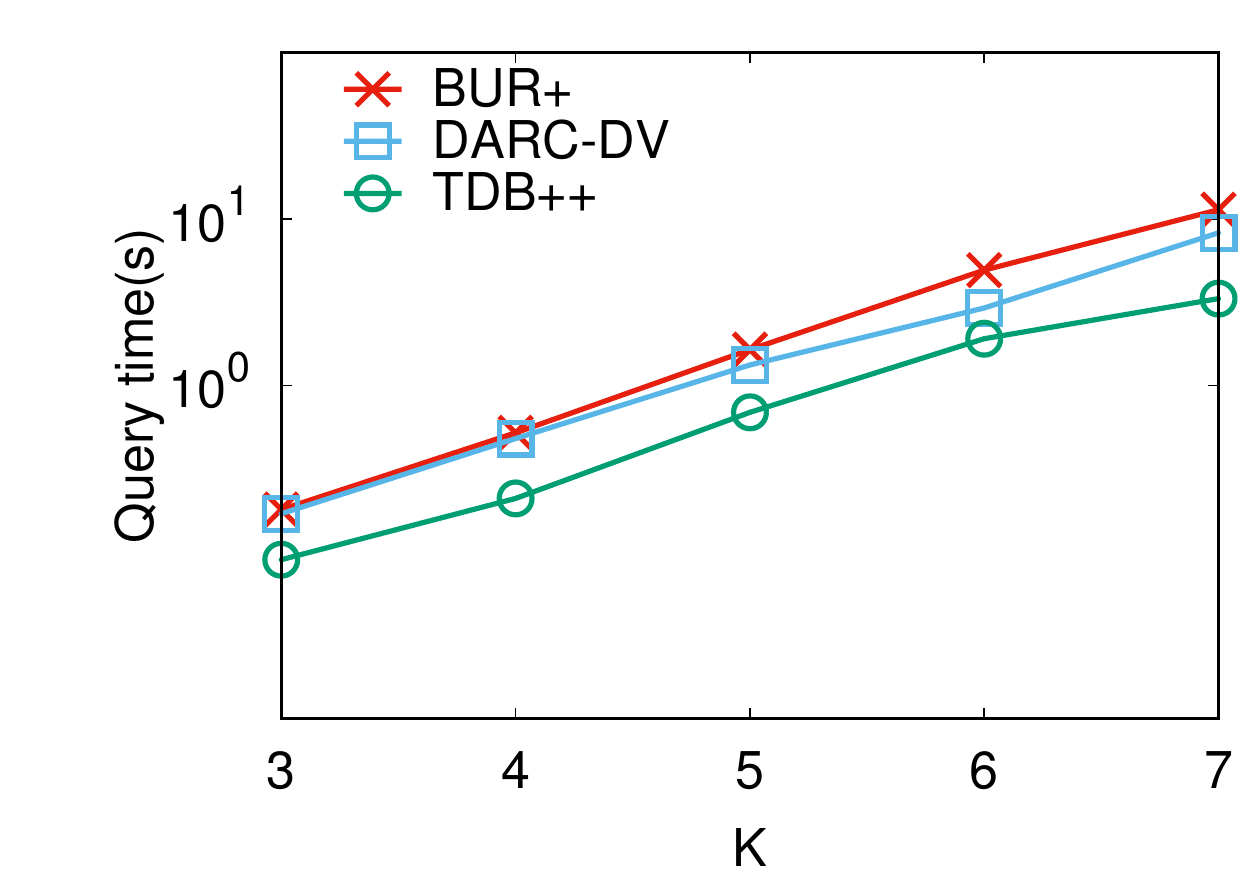}
	 \label{fig:ITL}
     }
     \subfigure[\small{EU}]{
     \includegraphics[width=0.23\linewidth]{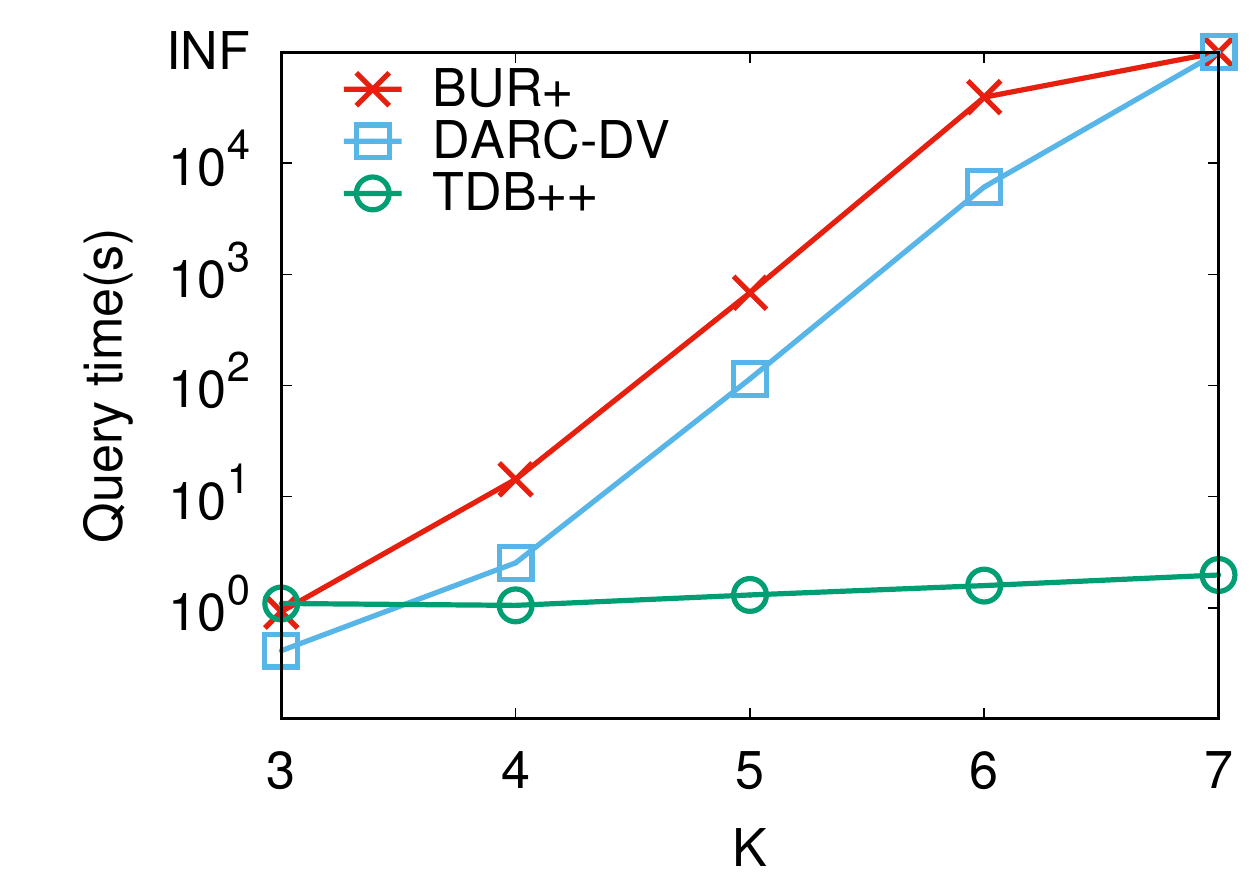}
	 \label{fig:ITW}
     }
          \subfigure[\small{SAD}]{
     \includegraphics[width=0.23\linewidth]{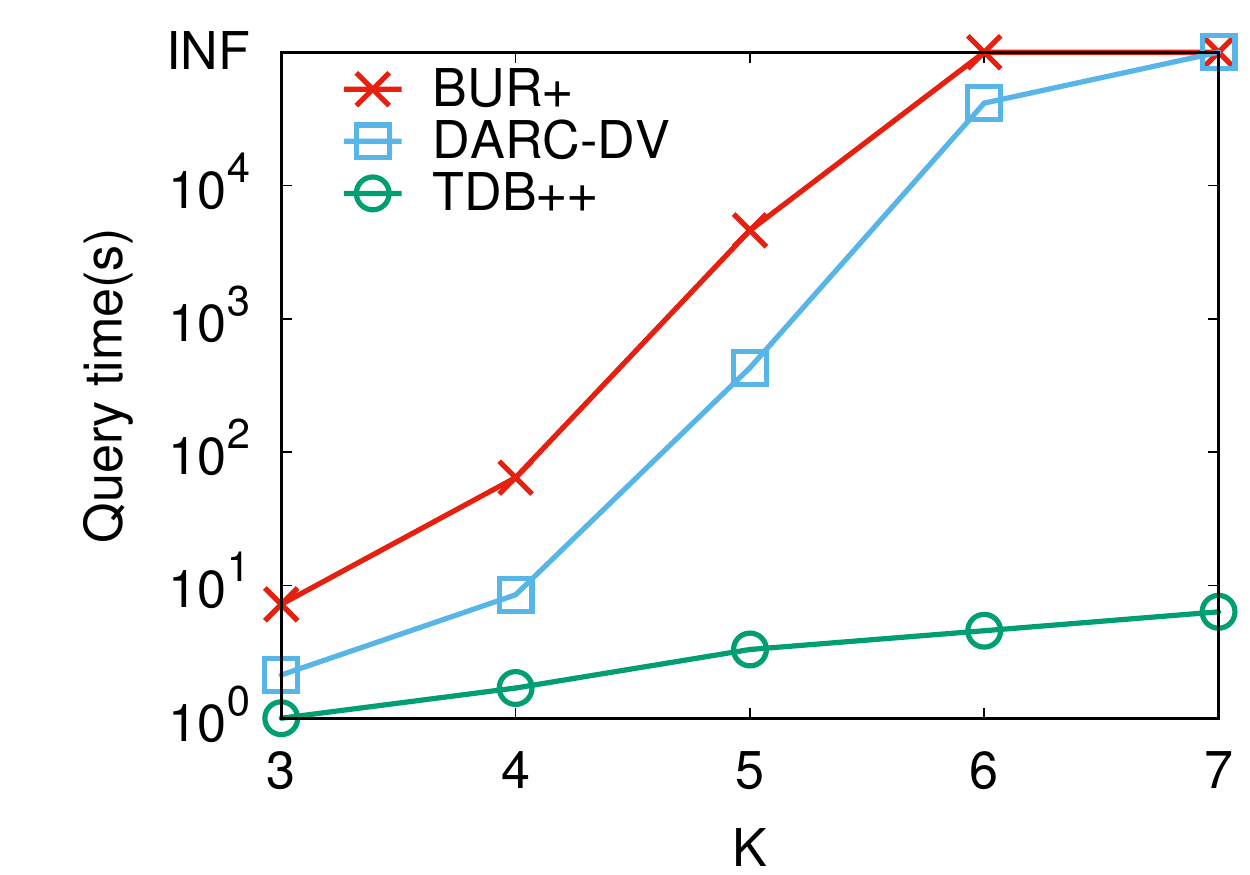}
	 \label{fig:ITL}
     }
     \subfigure[\small{WND}]{
     \includegraphics[width=0.23\linewidth]{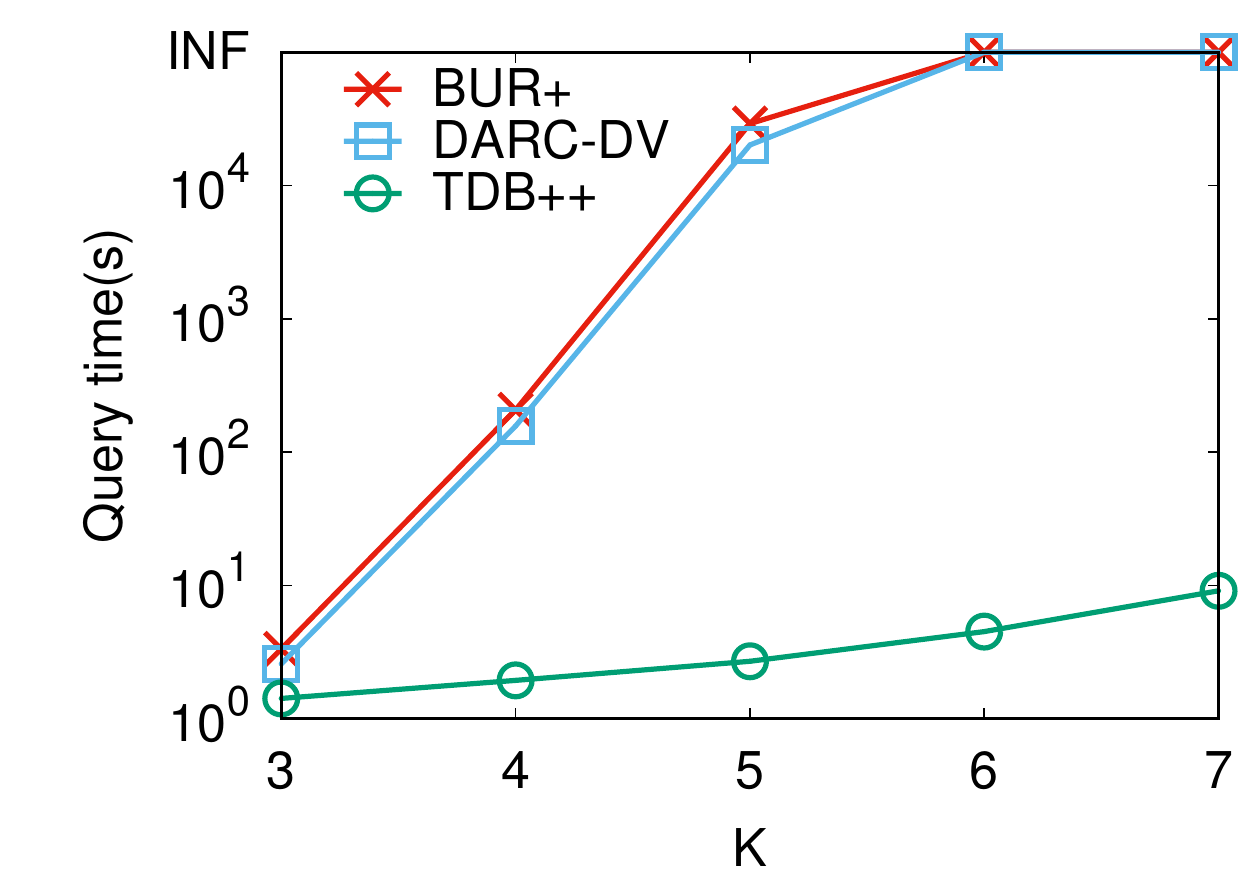}
	 \label{fig:ITW}
     }
          \subfigure[\small{CT}]{
     \includegraphics[width=0.23\linewidth]{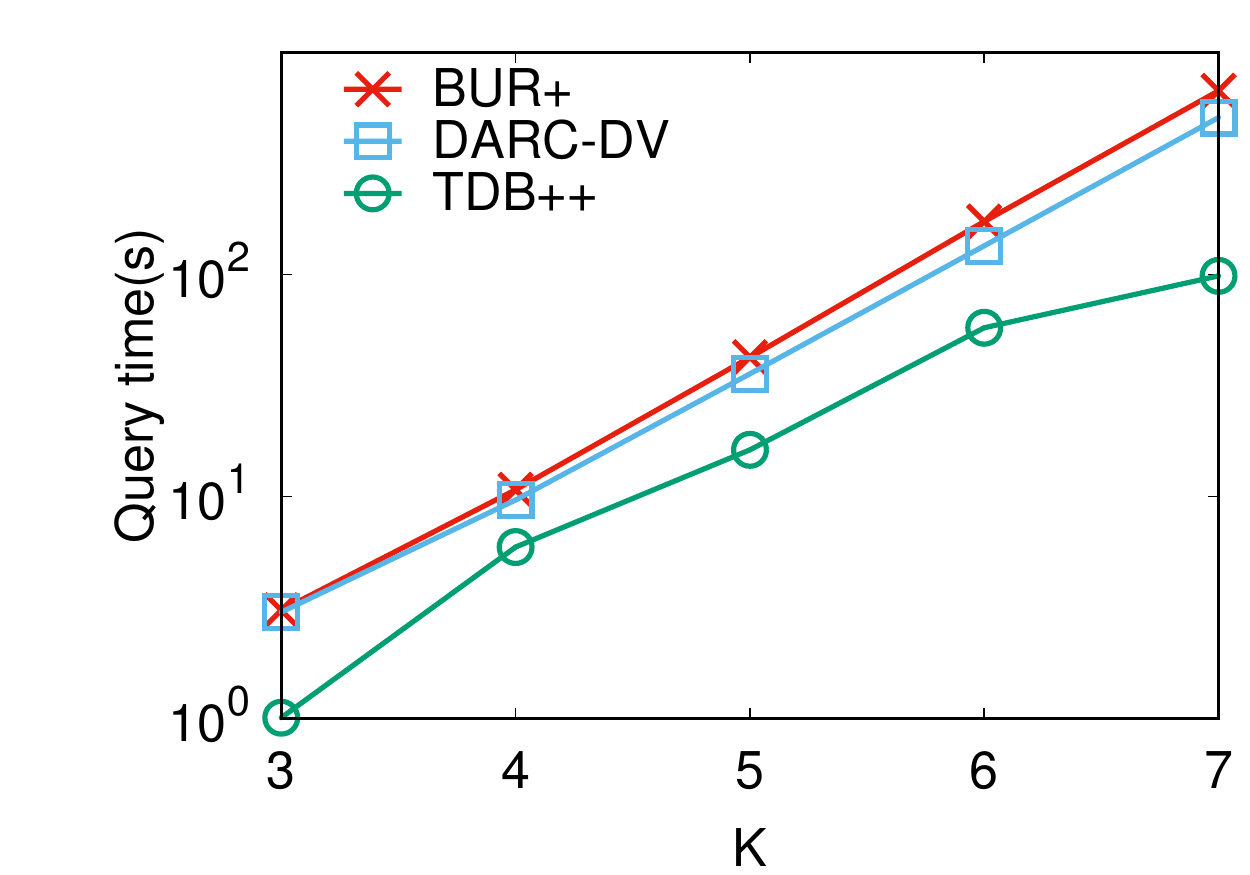}
	 \label{fig:ITL}
     }
     \subfigure[\small{WST}]{
     \includegraphics[width=0.23\linewidth]{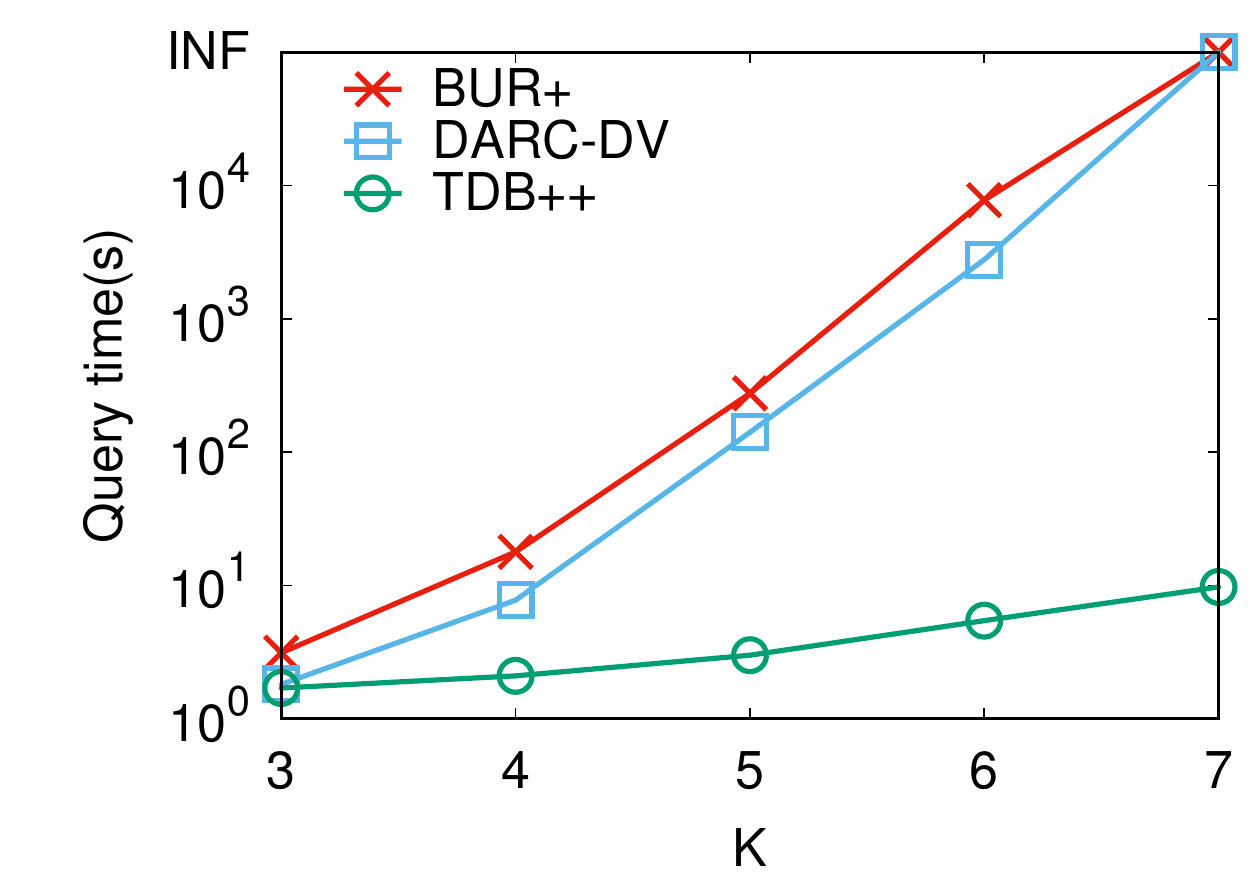}
	 \label{fig:ITW}
     }
               \subfigure[\small{LOAN}]{
     \includegraphics[width=0.23\linewidth]{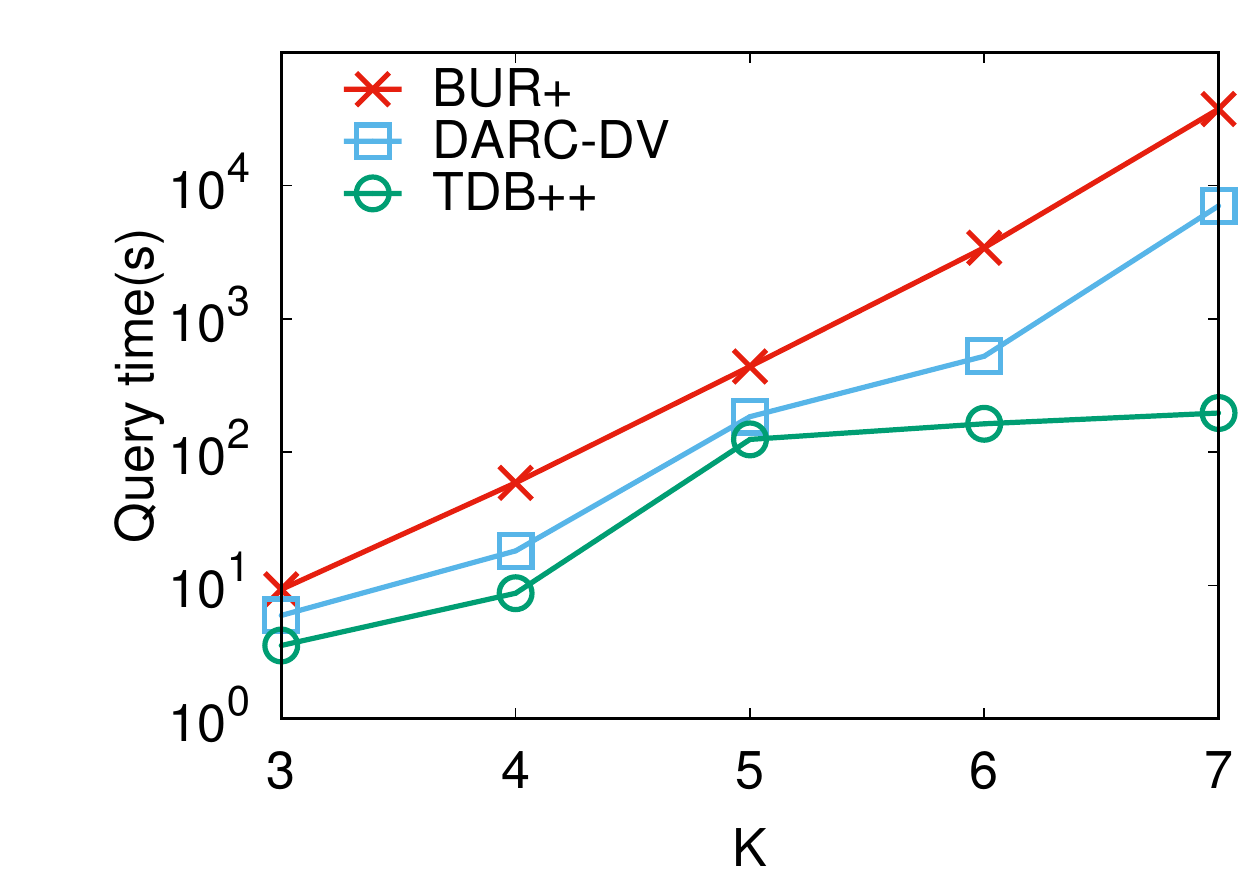}
	 \label{fig:ITL}
     }
     \subfigure[\small{WIT}]{
     \includegraphics[width=0.23\linewidth]{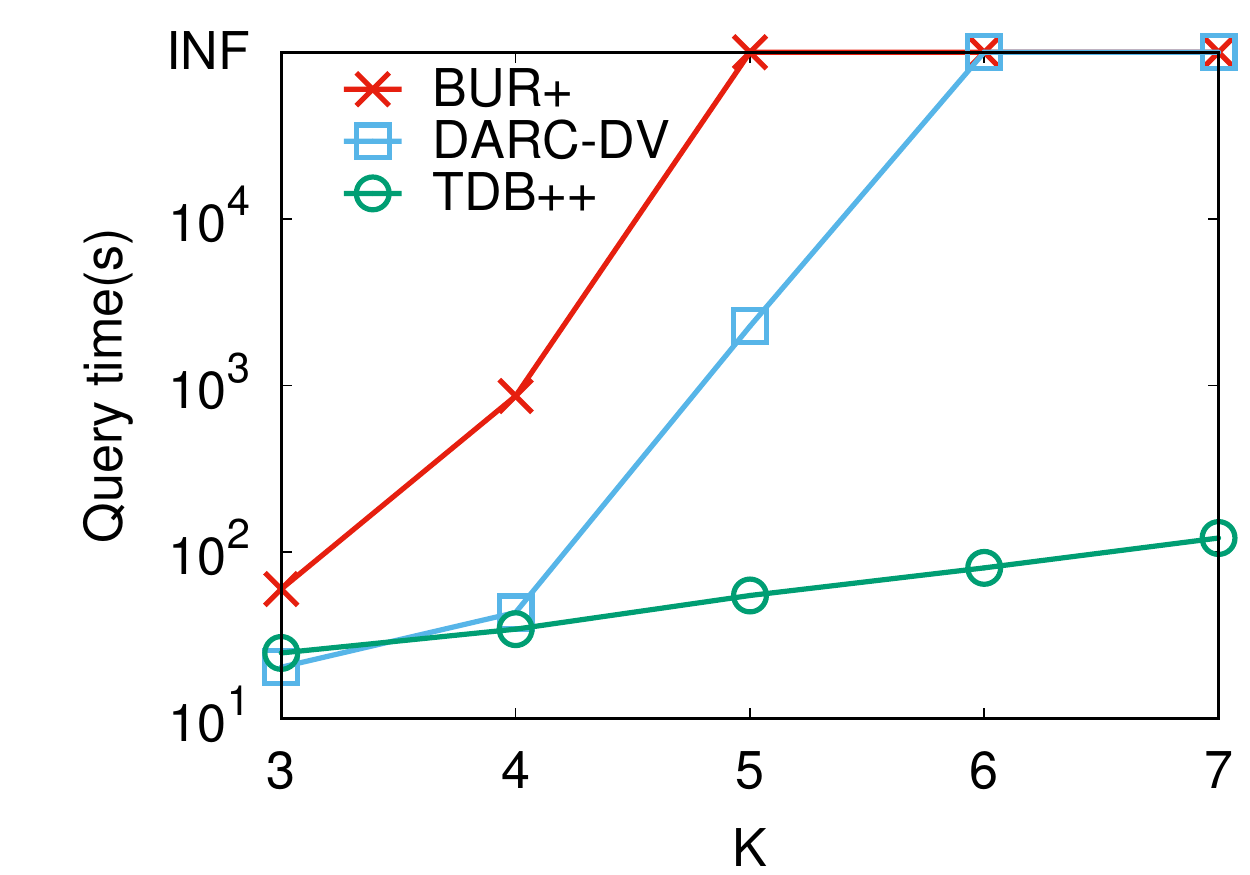}
	 \label{fig:ITW}
     }
               \subfigure[\small{WGO}]{
     \includegraphics[width=0.23\linewidth]{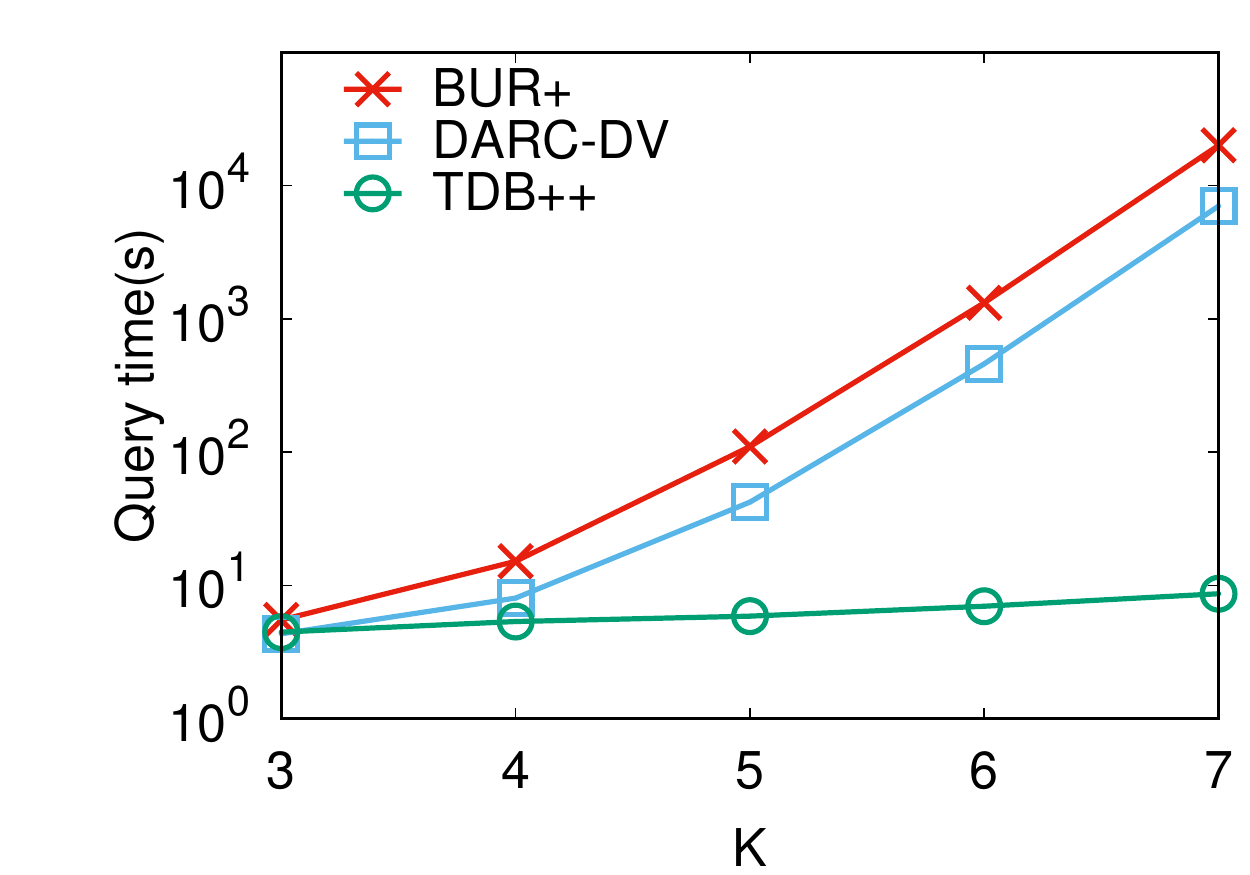}
	 \label{fig:ITL}
     }
     \subfigure[\small{WBS}]{
     \includegraphics[width=0.23\linewidth]{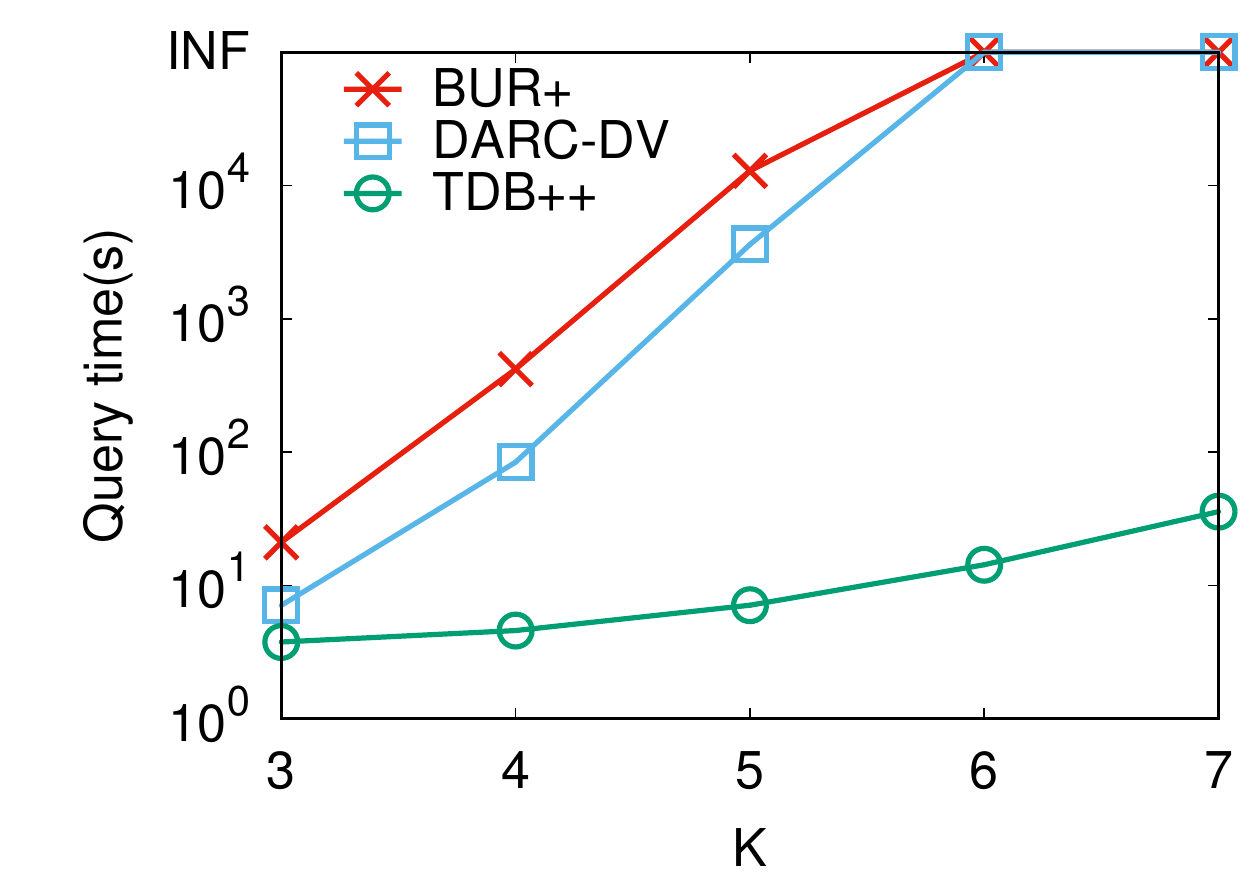}
	 \label{fig:ITW}
     }
	\vspace{-1mm}
\caption{Runtime (s).}
	\vspace{-1mm}
%\vspace{-0.3cm}
\label{fig:Runtime}
\end{figure*}

\begin{figure*}[htbp]
	\vspace{-1mm}
	\newskip\subfigtoppskip \subfigtopskip = -0.1cm
	\newskip\subfigcapskip \subfigcapskip = -0.1cm
%     	\begin{minipage}[b]{\linewidth}
%		\centering
%		\includegraphics[width=0.5\linewidth]{main_graph_keys.eps}%
%	\end{minipage}	
	\centering
     \subfigure[\small{WKV}]{
     \includegraphics[width=0.23\linewidth]{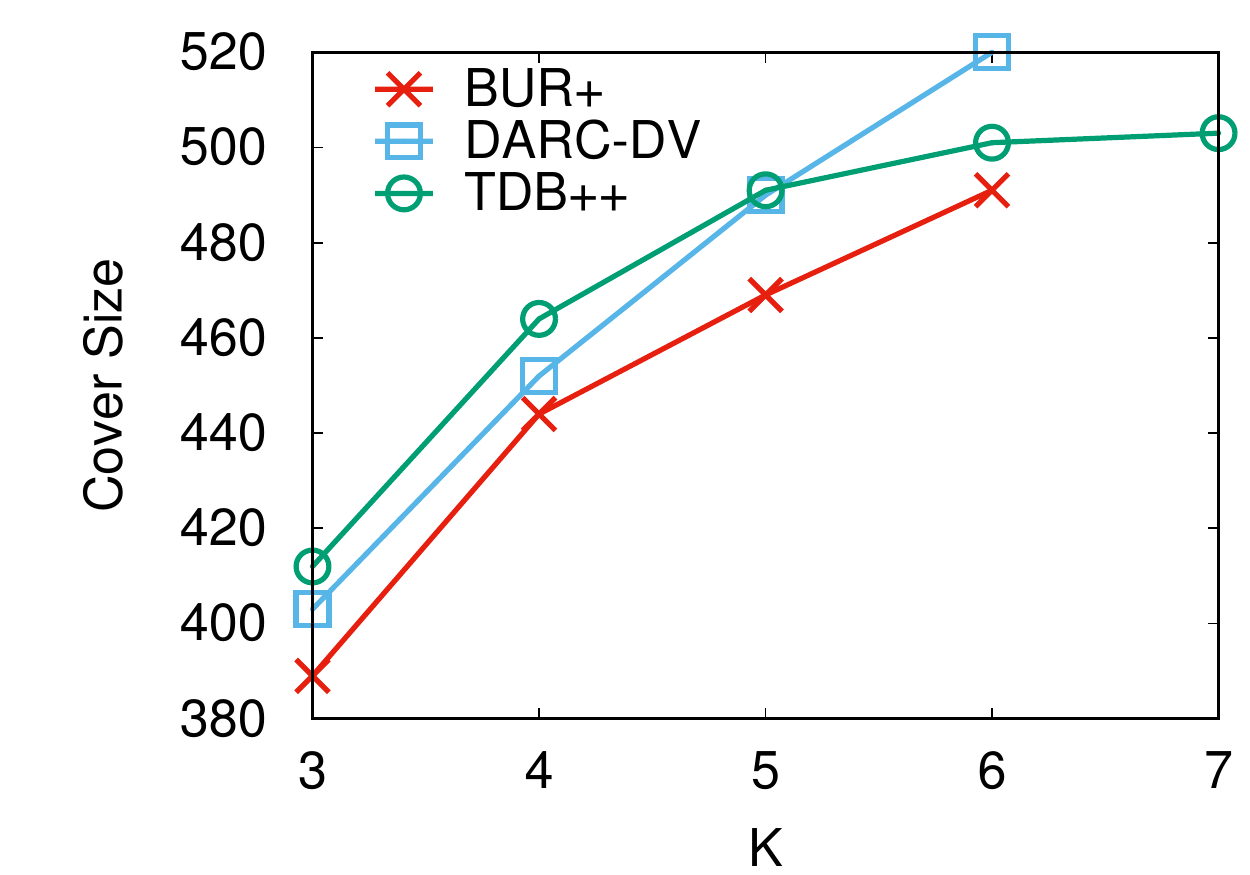}
	 \label{fig:ITL}
     }
     \subfigure[\small{ASC}]{
     \includegraphics[width=0.23\linewidth]{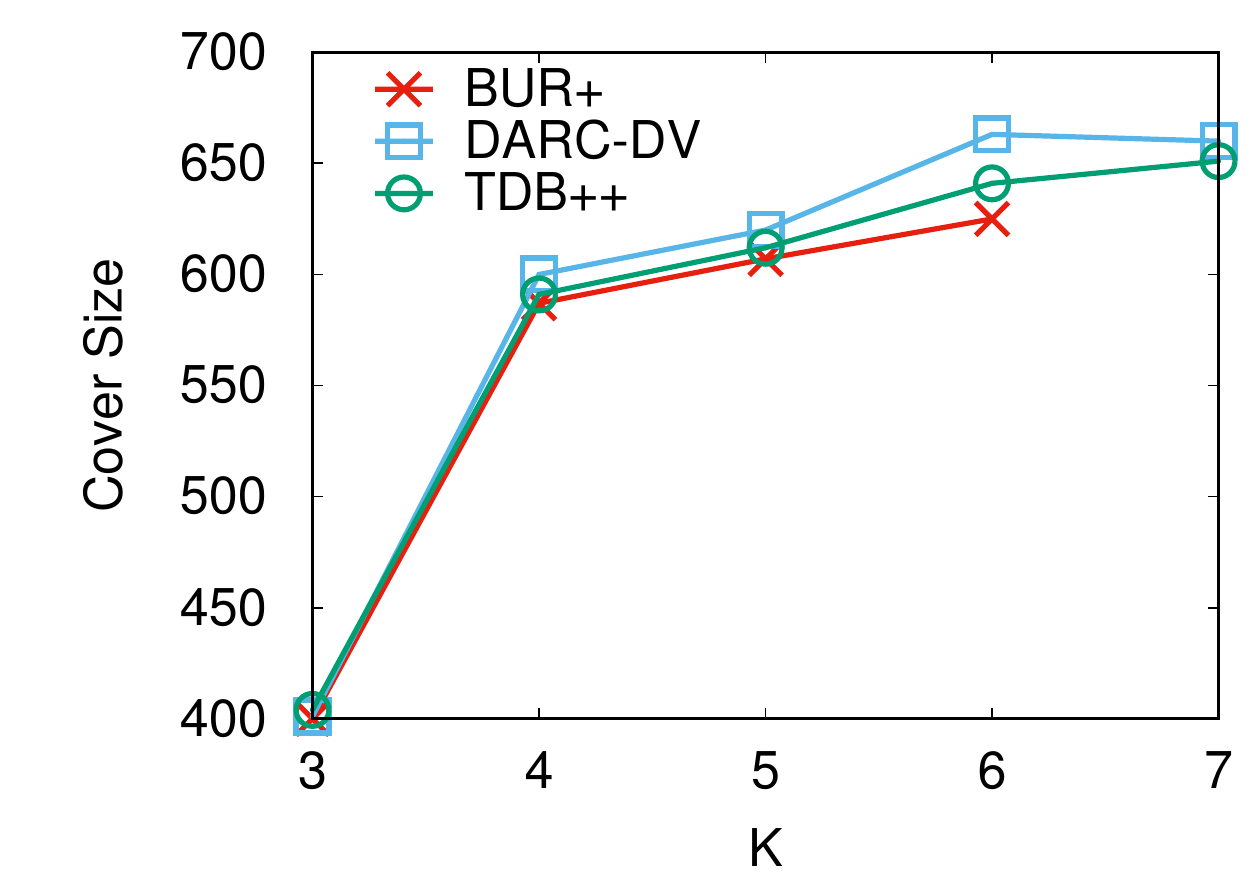}
	 \label{fig:ITW}
     }
          \subfigure[\small{GNU}]{
     \includegraphics[width=0.23\linewidth]{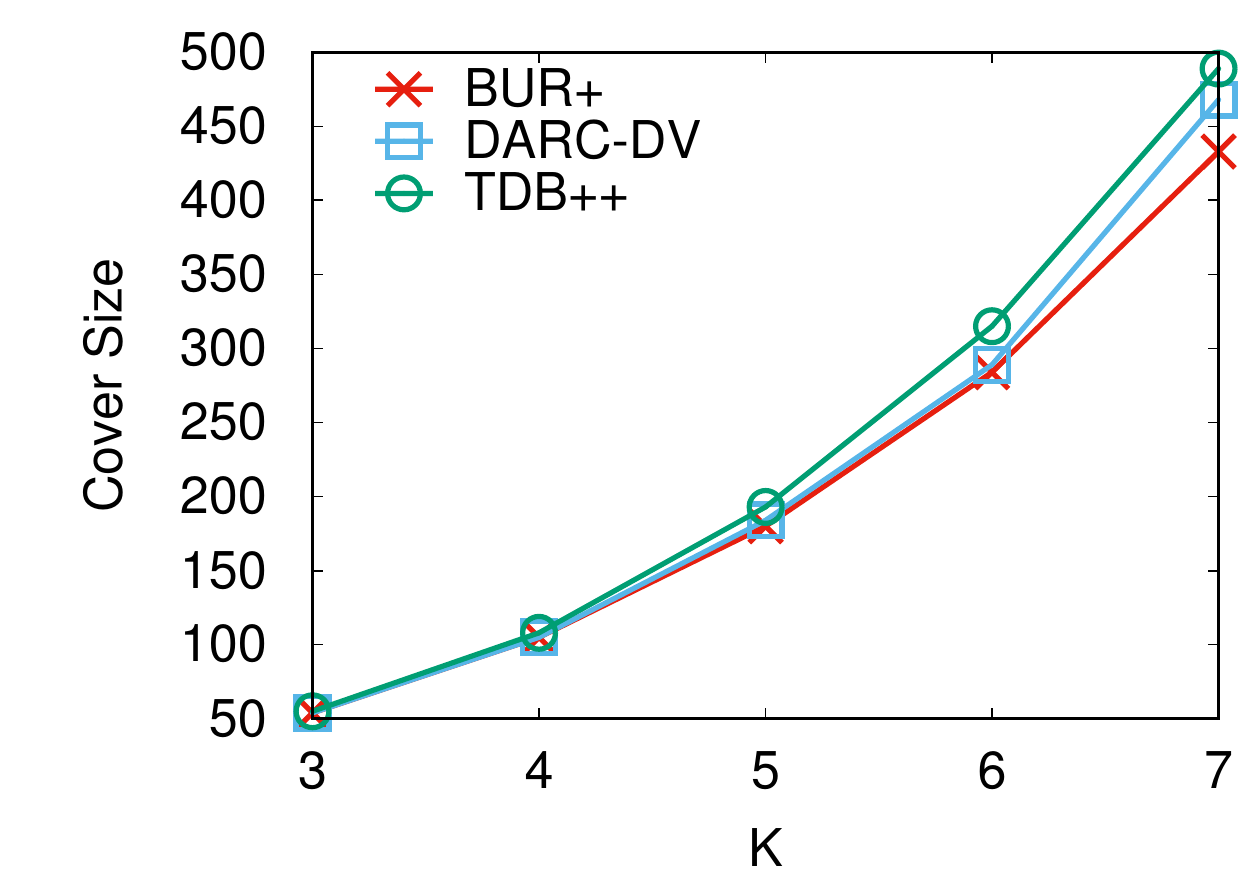}
	 \label{fig:ITL}
     }
     \subfigure[\small{EU}]{
     \includegraphics[width=0.23\linewidth]{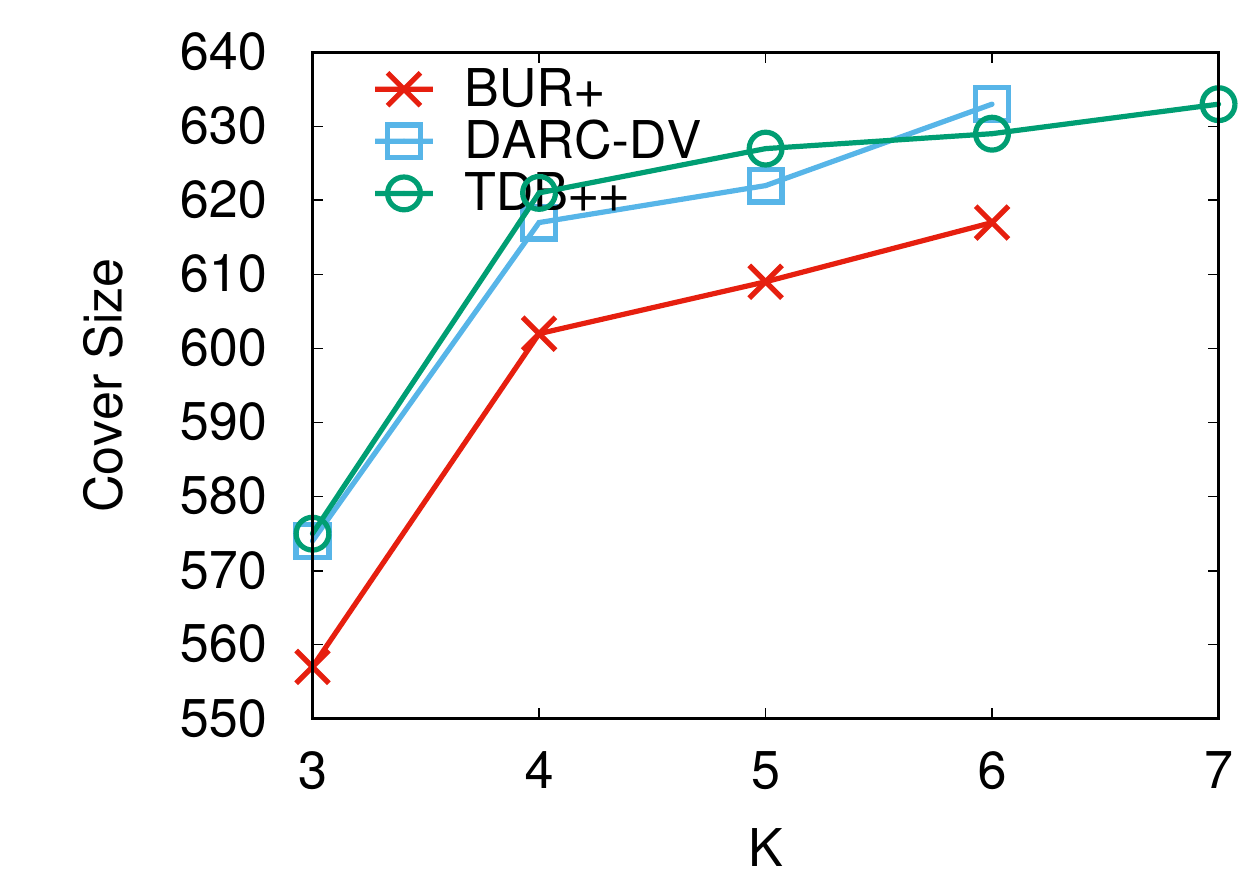}
	 \label{fig:ITW}
     }
          \subfigure[\small{SAD}]{
     \includegraphics[width=0.23\linewidth]{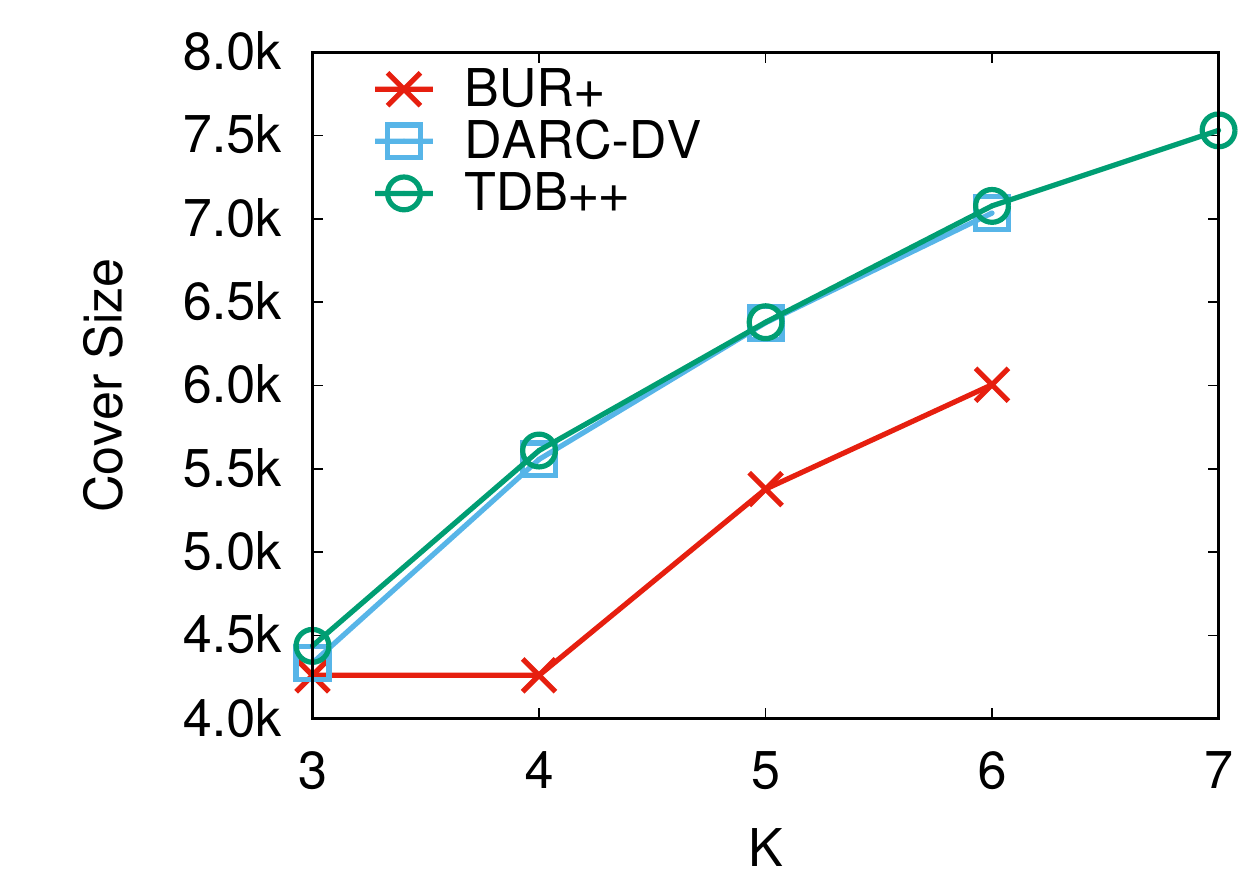}
	 \label{fig:ITL}
     }
     \subfigure[\small{WND}]{
     \includegraphics[width=0.23\linewidth]{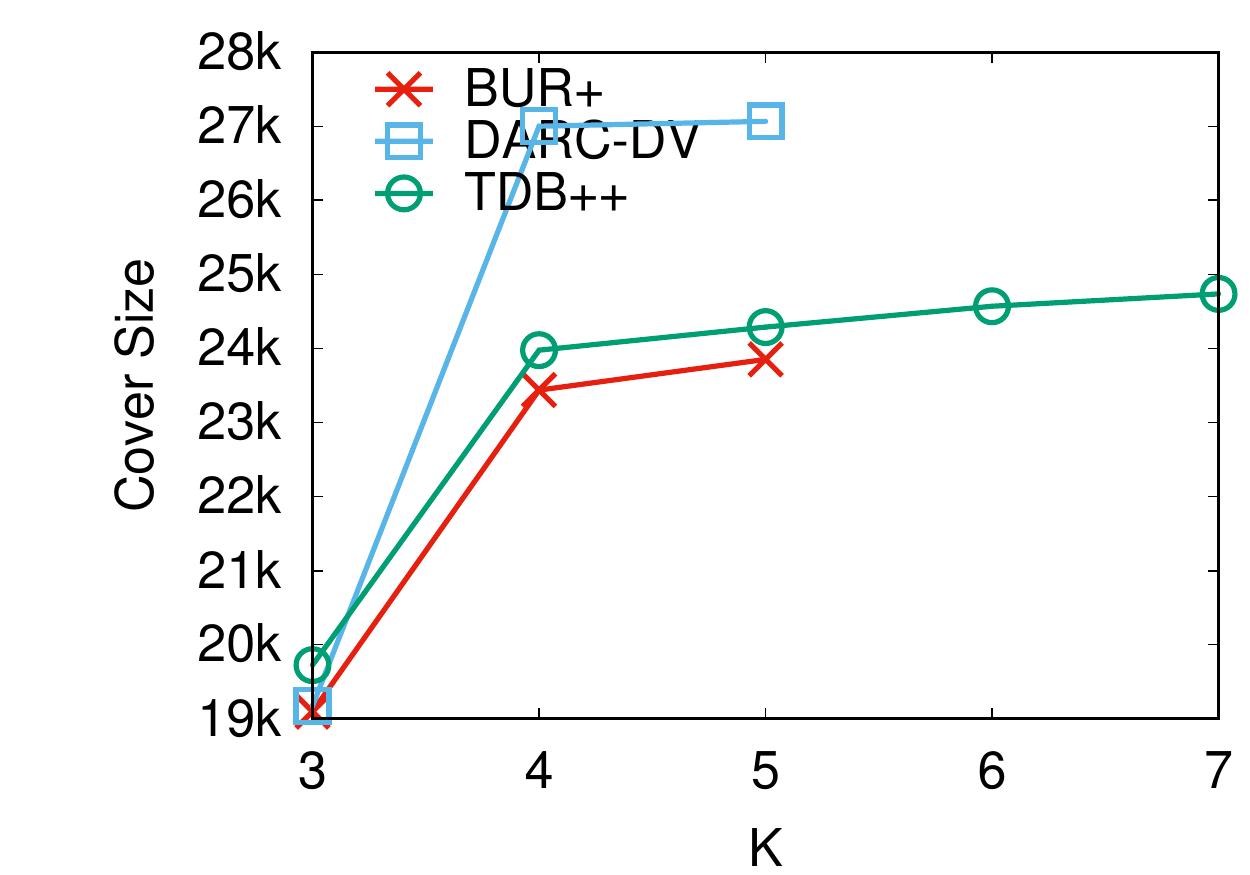}
	 \label{fig:ITW}
     }
          \subfigure[\small{CT}]{
     \includegraphics[width=0.23\linewidth]{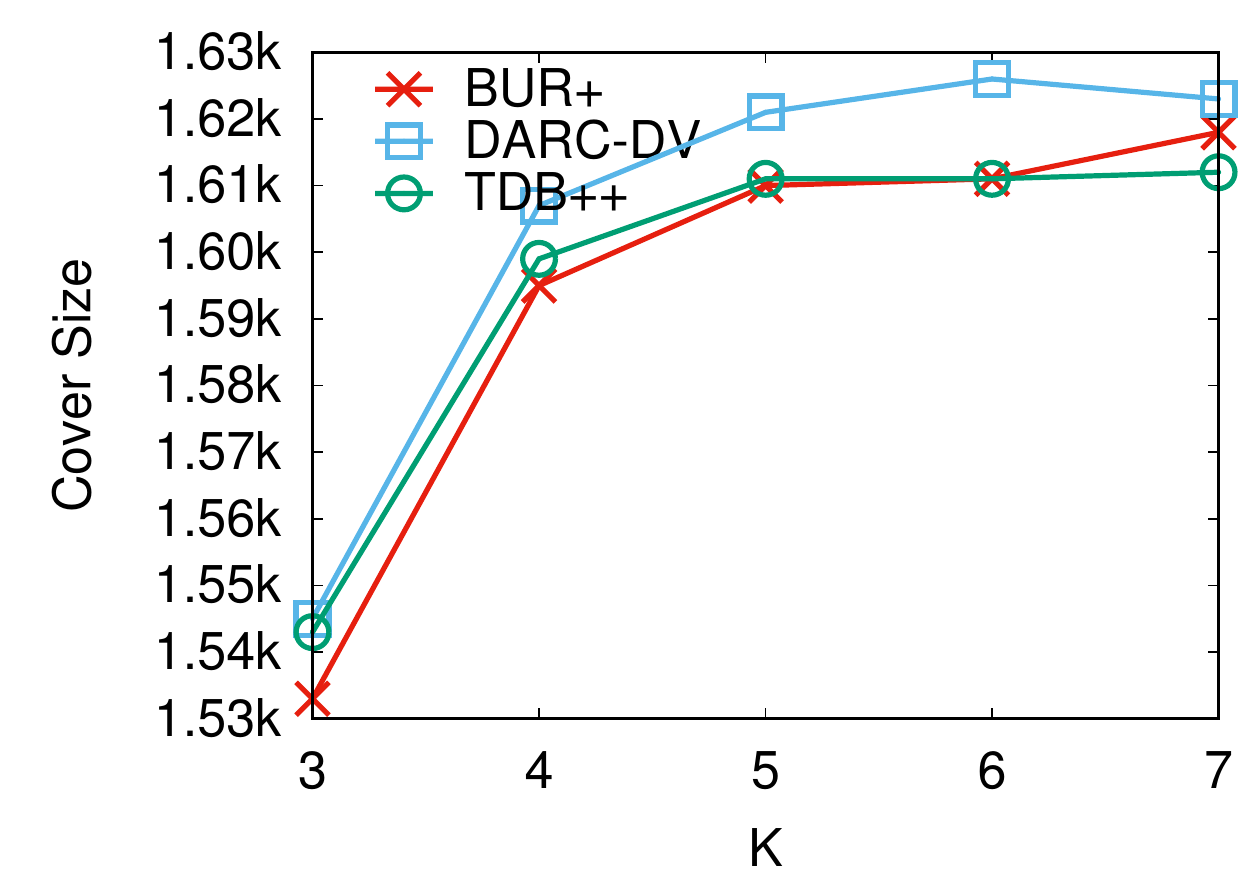}
	 \label{fig:ITL}
     }
     \subfigure[\small{WST}]{
     \includegraphics[width=0.23\linewidth]{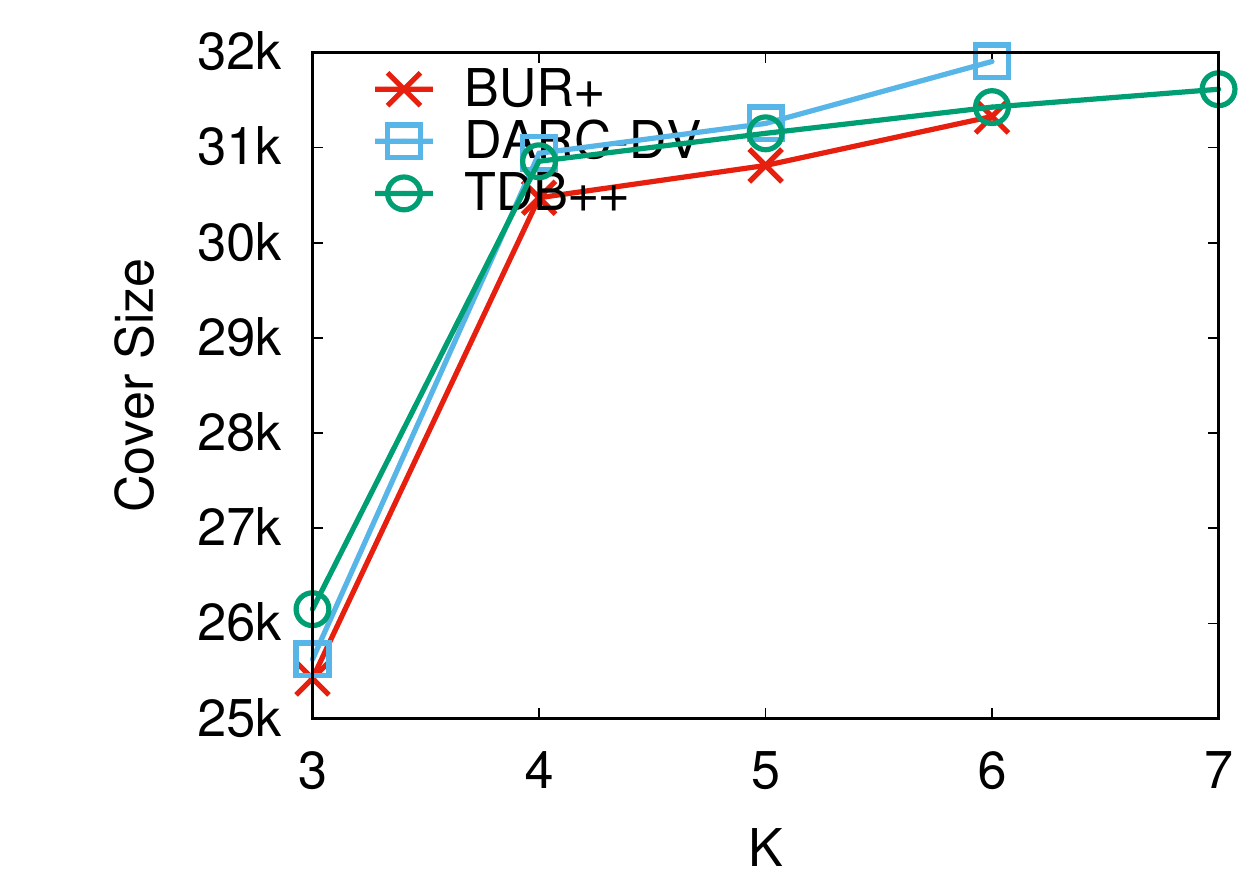}
	 \label{fig:ITW}
     }
               \subfigure[\small{LOAN}]{
     \includegraphics[width=0.23\linewidth]{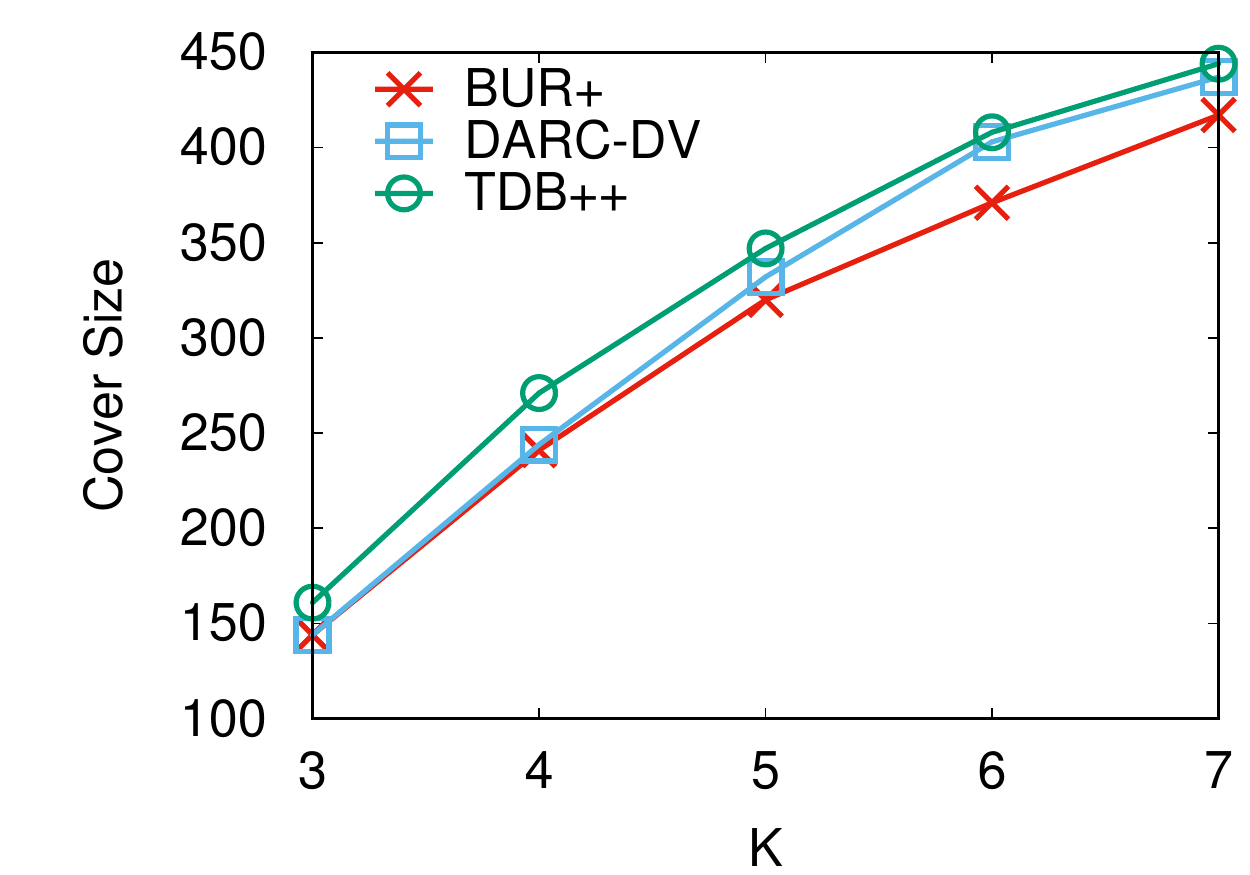}
	 \label{fig:ITL}
     }
     \subfigure[\small{WIT}]{
     \includegraphics[width=0.23\linewidth]{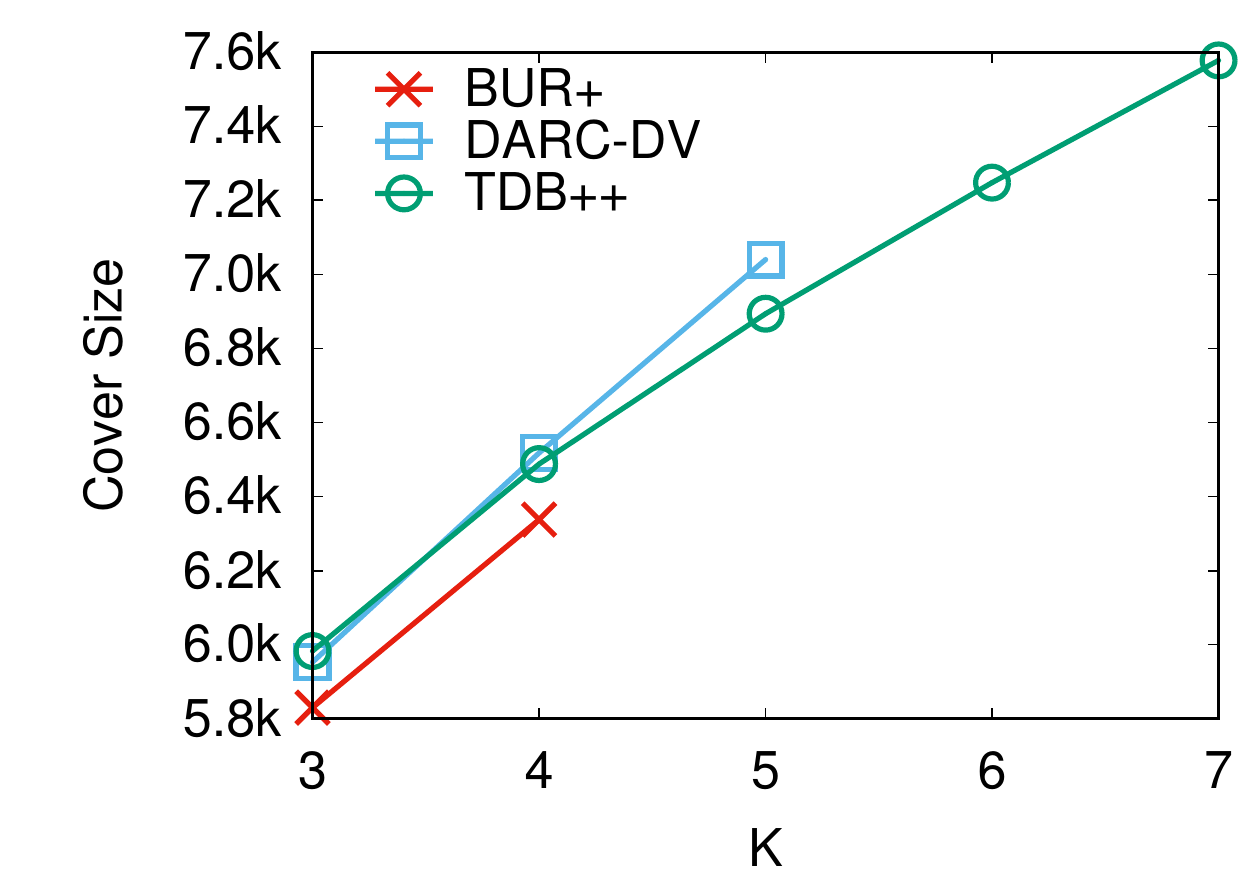}
	 \label{fig:ITW}
     }
               \subfigure[\small{WGO}]{
     \includegraphics[width=0.23\linewidth]{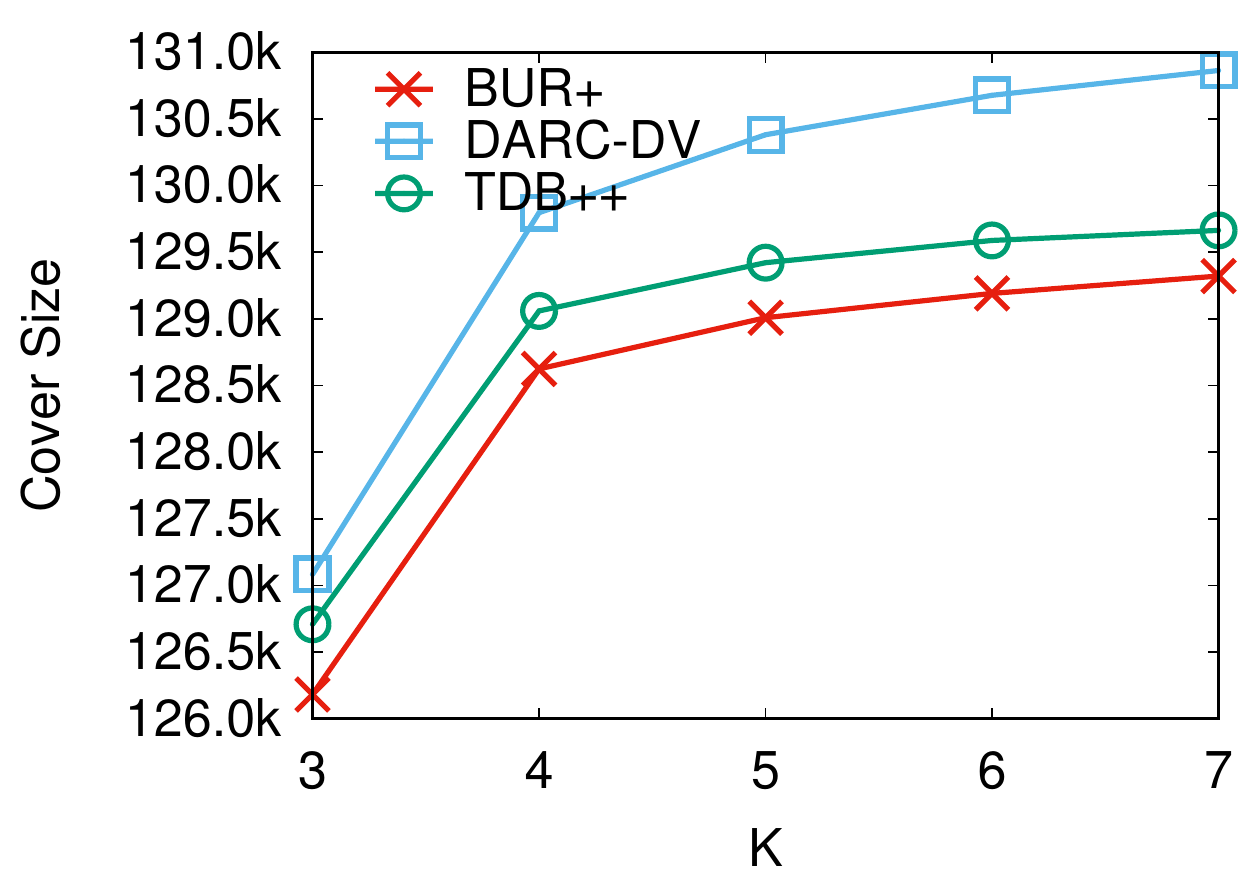}
	 \label{fig:ITL}
     }
     \subfigure[\small{WBS}]{
     \includegraphics[width=0.23\linewidth]{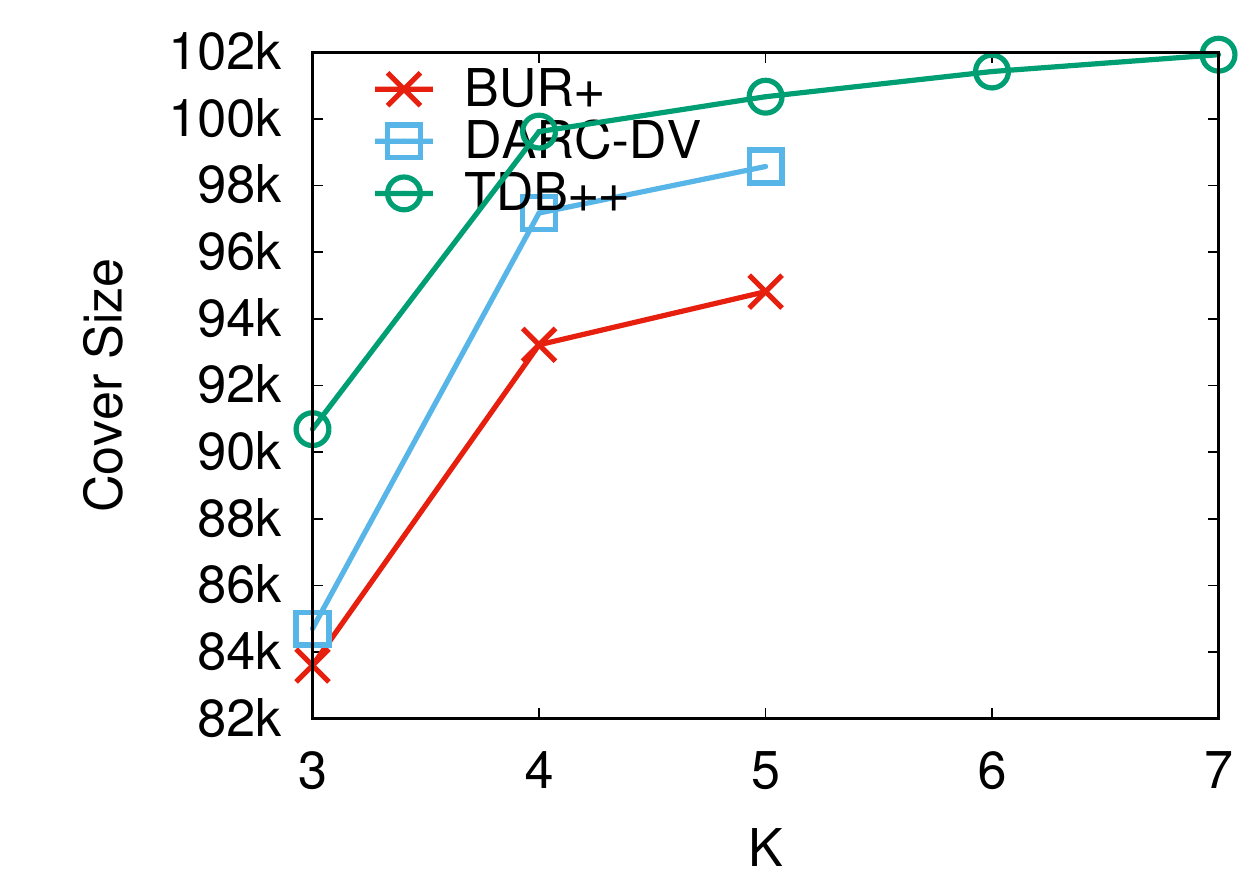}
	 \label{fig:ITW}
     }
	\vspace{-1mm}
\caption{Cover size (\# of vertices).}
	\vspace{-1mm}
%\vspace{-0.3cm}
\label{fig:CoverSize}
\end{figure*}

\begin{figure}[htbp]
	\vspace{-2mm}
	\newskip\subfigtoppskip \subfigtopskip = -0.1cm
	\newskip\subfigcapskip \subfigcapskip = -0.1cm
%     	\begin{minipage}[b]{\linewidth}
%		\centering
%		\includegraphics[width=0.5\linewidth]{main_graph_keys.eps}%
%	\end{minipage}	
	\centering
     \subfigure[\small{WKV}]{
     \includegraphics[width=0.45\linewidth]{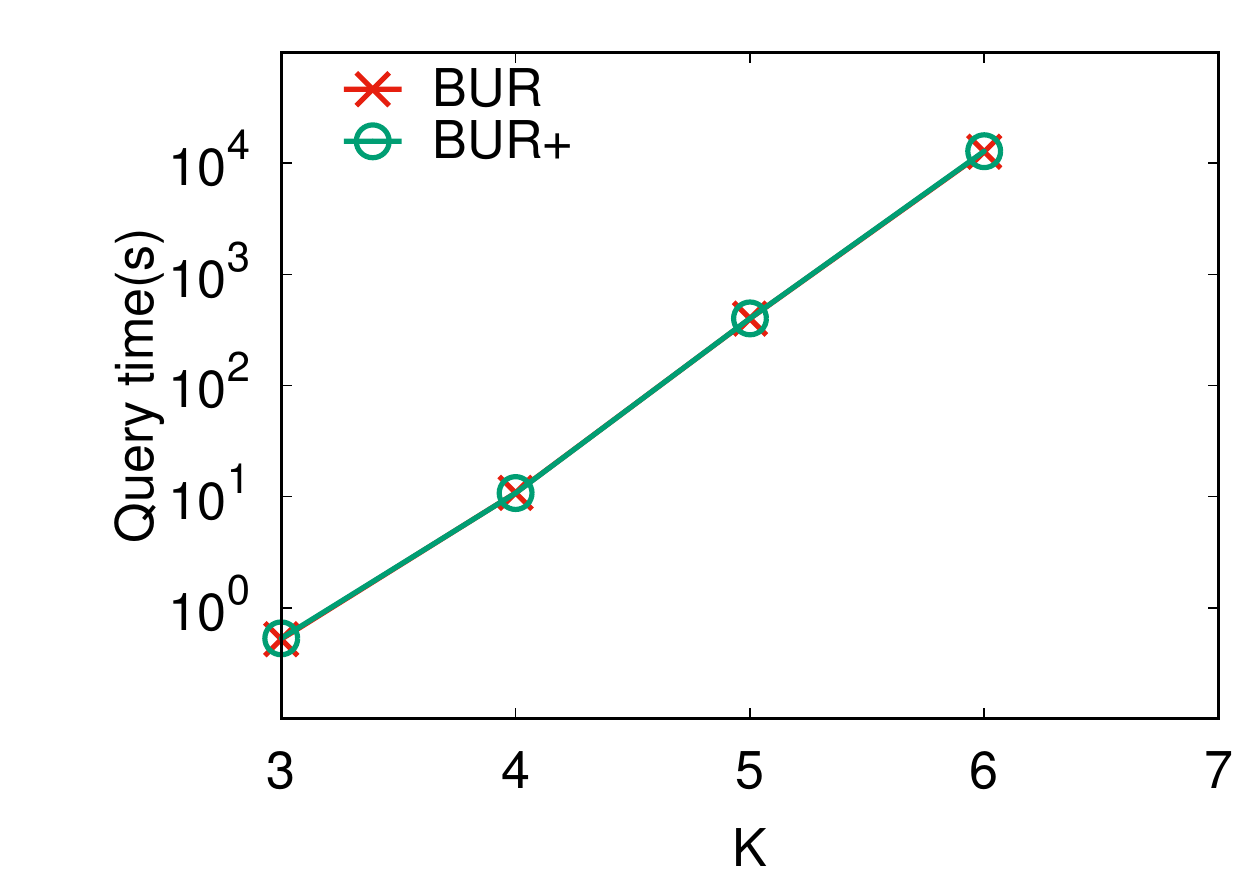}
	 \label{fig:ITLP}
     }
     \subfigure[\small{WGO}]{
     \includegraphics[width=0.45\linewidth]{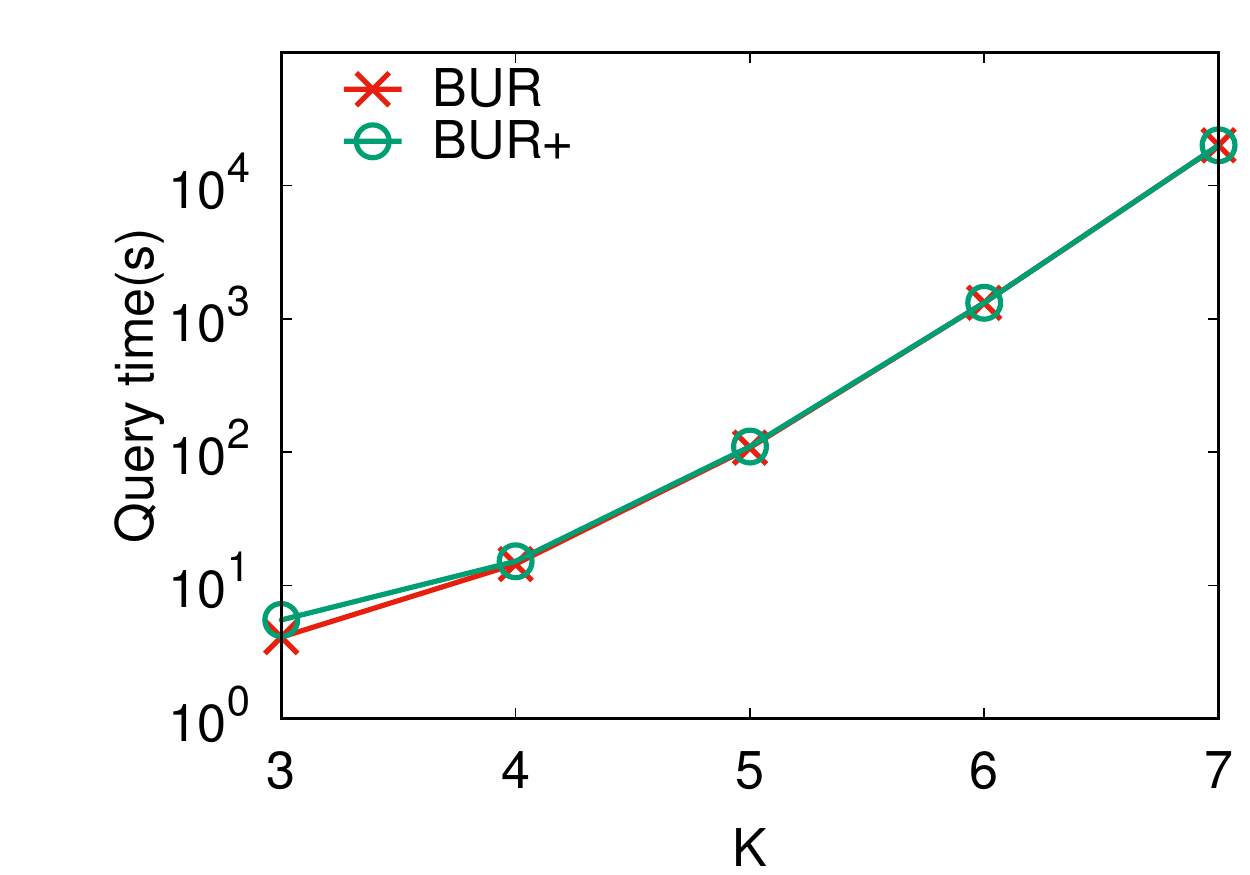}
	 \label{fig:ITWP}
     }
	%\vspace{-3mm}
\caption{Runtime (s).}
	%\vspace{-3mm}
%\vspace{-0.3cm}
\label{fig:RuntimeP}
\end{figure}

\begin{figure}[htbp]
	%\vspace{-2mm}
	\newskip\subfigtoppskip \subfigtopskip = -0.1cm
	\newskip\subfigcapskip \subfigcapskip = -0.1cm
%     	\begin{minipage}[b]{\linewidth}
%		\centering
%		\includegraphics[width=0.5\linewidth]{main_graph_keys.eps}%
%	\end{minipage}	
	\centering
 \subfigure[\small{WKV}]{
     \includegraphics[width=0.45\linewidth]{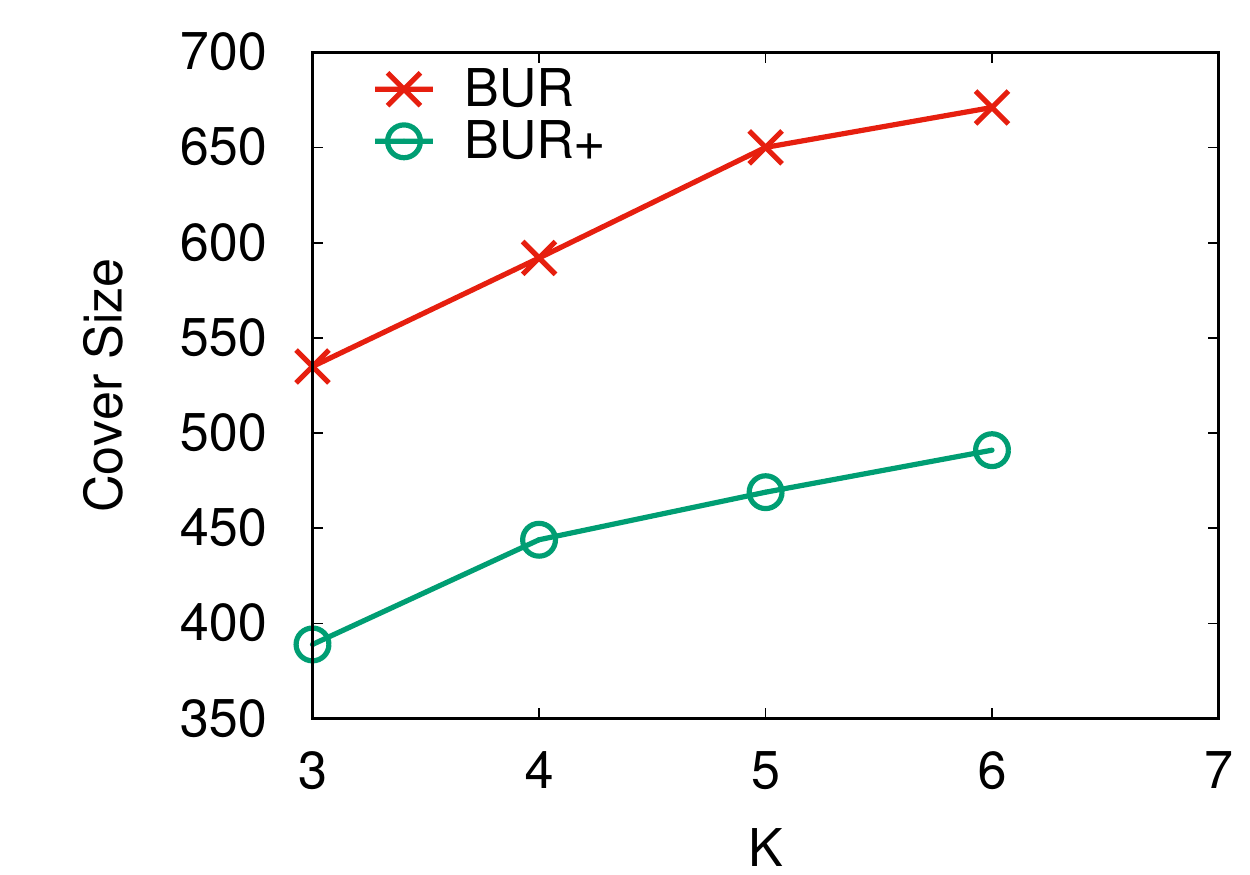}
	 \label{fig:ISLP}
     }
 \subfigure[\small{WGO}]{
     \includegraphics[width=0.45\linewidth]{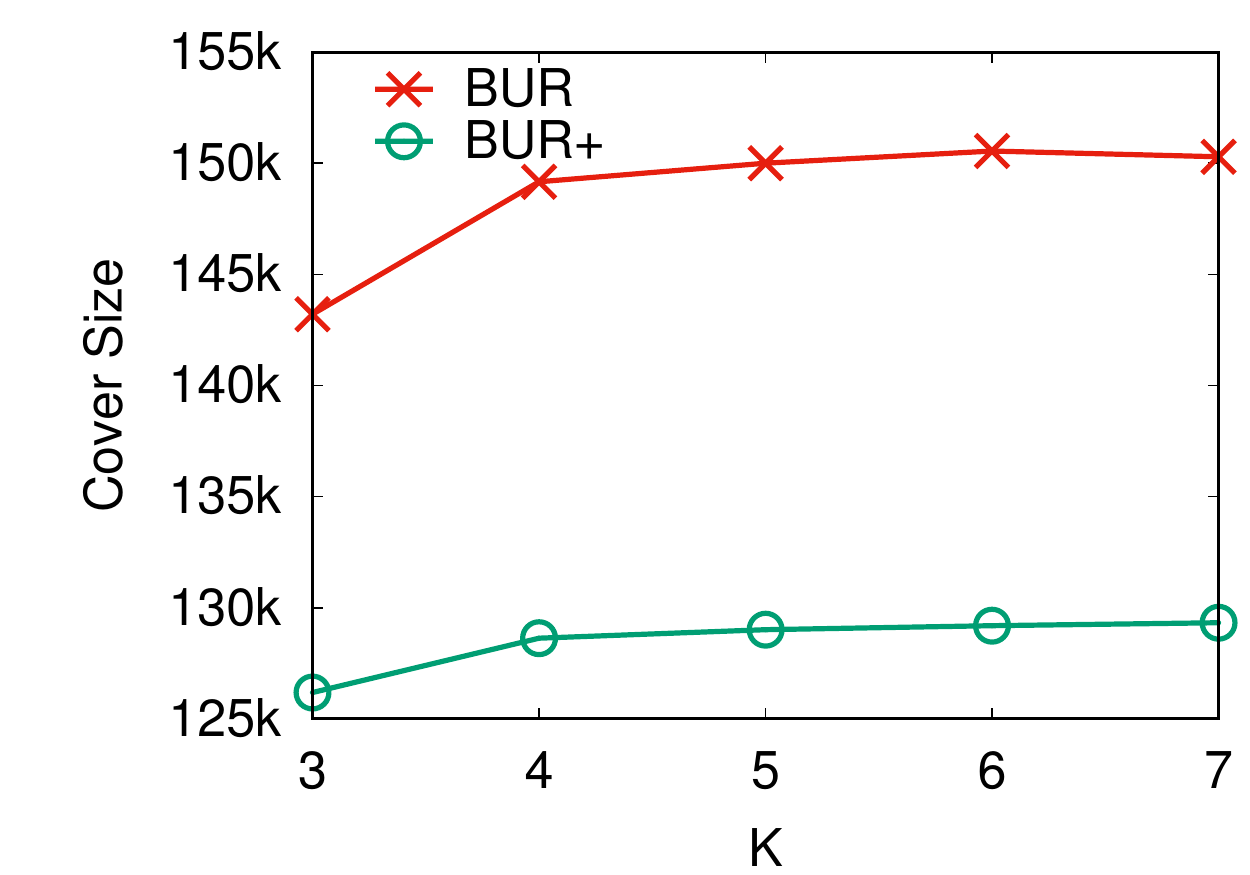}
	 \label{fig:ISTP}
     }
	%\vspace{-3mm}
\caption{Cover size (\# of vertices).}
	%\vspace{-3mm}
%\vspace{-0.3cm}
\label{fig:CoverSizeP}
\end{figure}

\subsection{The Speedup Effects}
\label{sect:speedup}
In this subsection, \rev{all the techniques in \tpd are evaluated.} Figure~\ref{fig:Runtime_sp} illustrates the speed-up benefits of all the techniques in WKV and WGO, varying $k$ from $3$ to $7$. What is remarkable about the figures is that when $k$ is large, the \bft contributes more speedup effect than the \bt. The insight is that the \bft is a linear filter technique and is effective in a wide variety of situations. Nonetheless, both the \bt and \bft contribute comparable speed-up effects in both datasets when $k$ is small. \rev{Since the result sets generated by all three methods are identical, their cover sizes are not reported.}

\begin{figure}[htbp]
	\vspace{-2mm}
	\newskip\subfigtoppskip \subfigtopskip = -0.1cm
	\newskip\subfigcapskip \subfigcapskip = -0.1cm
%     	\begin{minipage}[b]{\linewidth}
%		\centering
%		\includegraphics[width=0.5\linewidth]{main_graph_keys.eps}%
%	\end{minipage}	
	\centering
     \subfigure[\small{WKV}]{
     \includegraphics[width=0.45\linewidth]{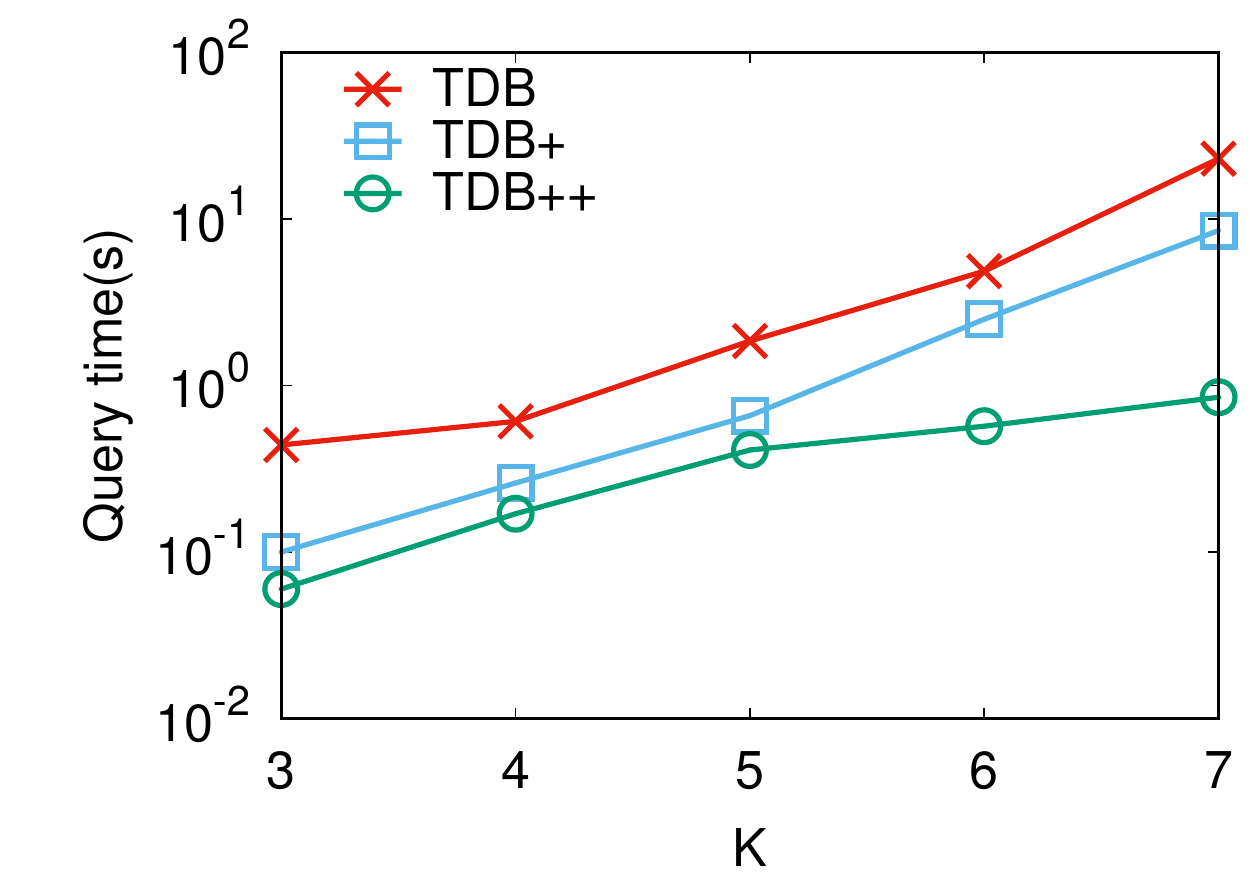}
	 \label{fig:sp_V}
     }
     \subfigure[\small{WGO}]{
     \includegraphics[width=0.45\linewidth]{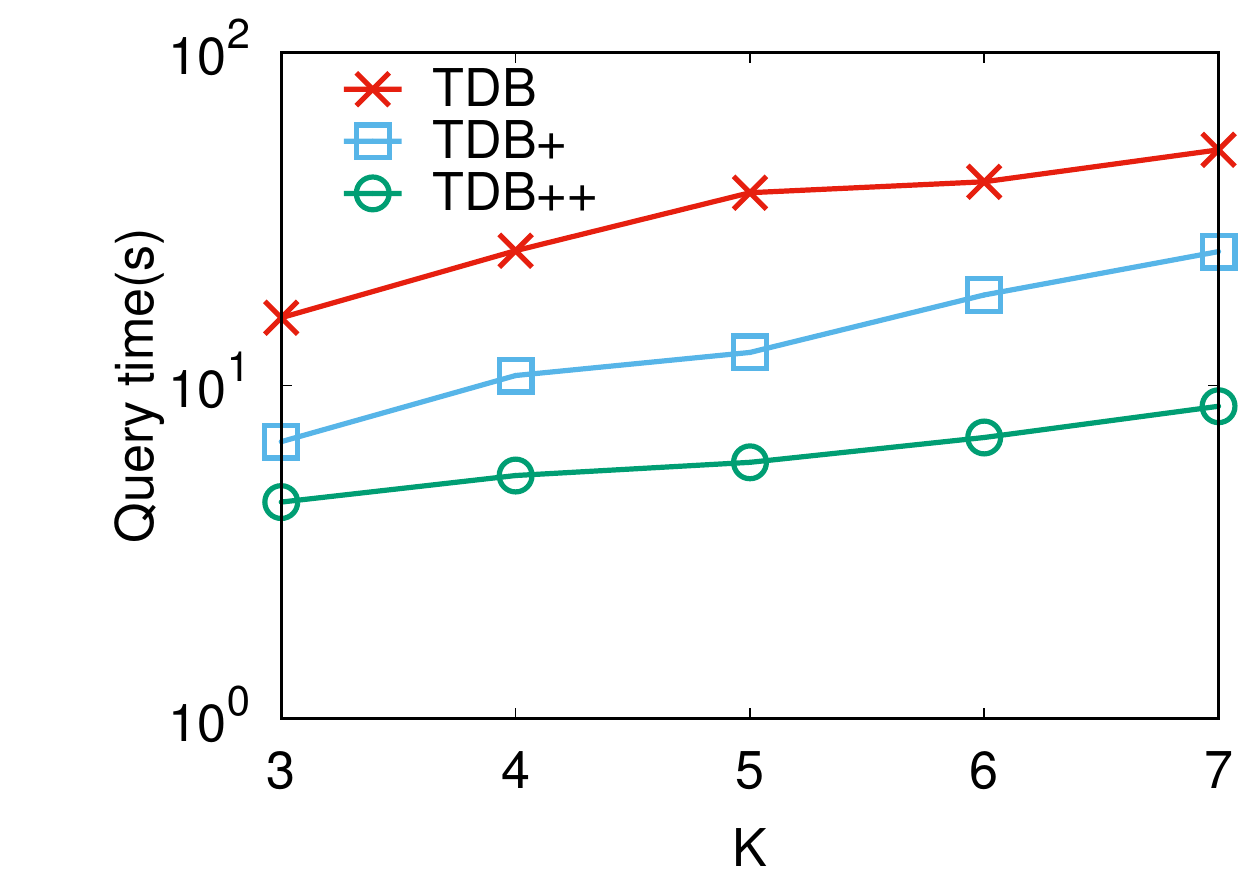}
	 \label{fig:sp_O}
     }
	%\vspace{-3mm}
\caption{Runtime (s) for \tpd techniques.}
	%\vspace{-3mm}
%\vspace{-0.3cm}
\label{fig:Runtime_sp}
\end{figure}

\subsection{Effectiveness and Efficiency on \kpc}
\label{sect:EffectAndEfficency}
In this part, the effectiveness and efficiency of all the experiments are discussed as follows. On real datasets, Table~\ref{tb:CoverTime} presents the cover size and runtime for \adpm, \darc, and \tdb. The $k$ is set to $5$ in this experiment. As demonstrated in Table~\ref{tb:CoverTime}, \adpm consistently returns the smallest cover size in 11 datasets except for WIT, despite using more time among all the algorithms. It is apparent from this table that \tdb method produces the smallest cover size in WIT among 12 datasets. In the remaining $11$ datasets, it provides a cover size that is comparable to \adpm, with an average difference of less than $4\%$. \tdb, on the other hand, runs $3$ orders faster than \adpm and up to $4$ orders in WND. When compared to \darc, what stands out in the table is that \tdb runs about $2$-$3$ orders of magnitude faster while returning a comparable cover size. Notably, only \tdb was capable of producing results on large graphs, i.e., FLK, LJ, WKP, and TW.
%\subsection{Effectiveness and Efficiency on \cvc}
%\label{sect:EffectAndEfficencycr}
%In this part, the effectiveness and efficiency of all the experiments are proposed as follows.
%Table~\ref{tb:CoverTimeInf} presents the cover size and runtime for \adpm, \kapp, \degr, \rand, and \tdb on real datasets. In this experiment, the $k$ is set to $\infty$. As shown in Table~\ref{tb:CoverTimeInf}, \tdb runs the fastest among all the datasets due to the insight of the top-down framework and the novel time complexity of blocks technique. The \tdb runs about 2-3 orders faster than the second-fastest algorithm \kapp.
%As for the cover size, \tdb runs the smallest in all the datasets that \adpm times out. It is apparent from this table that \kapp runs the second fastest but produces the largest cover size. What is striking about this table is \degr dominates \rand in both runtime and cover size in all the datasets they could generate results.
 
\subsection{Tuning the Parameter $k$}
\label{sect:tuningK}
\rev{Additionally, experiments are conducted by varying the parameter $k$.} 
\rev{This experiment is conducted with tuning the value of the parameter $k$ from $3$ to $7$ on $12$ distinct datasets to determine the cover size and runtime.}
As shown in Figure~\ref{fig:Runtime}, \tdb is the fastest algorithm across all the datasets, followed by \darc.
The figure demonstrates that \adpm runs slowest. Nevertheless, what stands out in Figure~\ref{fig:CoverSize} is \adpm generates the smallest cover size. \tdb produces a comparable cover size as \adpm but has the fastest runtime. As for \darc, it returns the worst cover size among these three methods.

\subsection{The Pruning Effects}
\label{sect:pruning}
\rev{This subsection conduct experiments to demonstrate the pruning effects.} As shown in Figure~\ref{fig:RuntimeP}, \adp and \adpm have a similar runtime in both WKV and WGO. Nevertheless, it is shown in Figure~\ref{fig:CoverSizeP} that \adpm has a smaller cover size owing to the minimal pruning approach in both datasets. WKV and WGO vary in that WKV has a higher average degree than WGO. \rev{In WKV, it could prune more percentage results.} \rev{As for WGO, the cover size difference between \adpm and \adp stays steady when $k$ grows.}

\subsection{Cover Size including $2$-cycles}
\label{sect:cover_size_2cycle}
\rev{Table~\ref{tb:Cover2Cycle} illustrates the cover size for our algorithm when including 2-cycles or not. What stands out in this table is that the cover size would be 3 times larger on average when including 2-cycles. For some graphs, e.g., GNU, the cover size does not grow too much. Nevertheless, for graphs ASC, SAD, WND, CT, WST, WIT, WGO and WBS, the cover size significantly grows. Since 2-cycles could be efficiently verified separately, our problem concentrates on constrained cycles without 2-cycles.}
    \vspace{1mm}
\begin{table}[htbp]
  \centering
      \caption{\rev{The cover size (the number of vertices) $k = 5$.}}%``K" indicates $10^3$.}
%\vspace{-1mm}
\label{tb:Cover2Cycle}
 %\resizebox{\textwidth}{7mm}{
%\resizebox{\textwidth}{!}{
%\resizebox{1.0\textwidth}{38mm}{
    \begin{tabular}{l|c|c|c}%{|c|cc|cc|cc|cc|}
      \hline
      	%\textbf{Name} 	& \multicolumn{2}{c|}{\adpm}	& \multicolumn{2}{c|}{\kapp }   & \multicolumn{2}{c|}{\degr}	& \multicolumn{2}{c|}{\rand }\\  % \textbf{$d_{max}$}
     	\cellcolor{gray!25}\textbf{Name} 	&	\cellcolor{gray!25}No $2$-cycle		& \cellcolor{gray!25}With $2$-cycle & \cellcolor{gray!25}Ratio	     		   \\ \hline
	WKV	 	&    491 & 714   & 1.45 		   \\ 
	ASC	 	&    612	& 	 	5,285 & 8.64 		   \\ 
	GNU	 	&    193	 & 222 	& 1.15		   \\ 	
	EU	 	& 627	 & 1,270 	& 2.03		   \\ 	
	
	%EPIN	 	&5,329 	& 28.3		&	\textbf{5,081}	 	& 1,127.6	&	 5,150	 & \textbf{1.04} 			   \\ 	
		
	%SEP	 	&\textbf{11,913}& 578		&	18,398 	& 		 &	12,431 	& 	240	      		   \\ 
	SAD	 		& 6,380	& 	 	27,461 & 4.30    		   \\ 	
	%DAU	 	&\textbf{63}	& 60.5		&	270	 	& 76.4	 &	63		& 	\textbf{60.2}	&	64	 	& 	62.9        		& &   \\

	WND	 	& 24,290 & 51,466	& 2.12 		   \\ 
 
	CT	 		& 1,611 	& 	7,615 & 4.73 	   \\ 	
	WST	 		& 31,148	& 		116,065 & 3.73  		   \\ 
	LOAN	 	& 347 &			568     & 1.64		   \\

	WIT	 		& 6,894	& 	21,781	      	& 3.16	   \\ 	
 	
	WGO	 	& 129,421&  217,799 & 1.68	   \\ 
		WBS	 	& 100,668		&256,281 & 2.55  		   \\ \hline
       \end{tabular}
\end{table}
%!TEX root = DamoGraph.tex
\balance
\vspace{-2mm}
\section{Conclusion}
\label{sect:conclusion}
\rev{This paper introduced the \kpc problem, whose objective is to discover a collection of vertices that covers all hop-constrained cycles in a given directed graph.} 
%It is used in a variety of contexts, including E-commerce networks, database systems, and program analysis.
On the theoretical side, \rev{this work demonstrates that approximating} the \kpc issue with length between $3$ and $k$ is UGC-hard (Unique Games Conjecture) for a given directed, unweighted graph $G$.
%\reat{On the practical level, two distinct algorithms are proposed: \textit{bottom-up} and \textit{top-down}.}
%\reat{Among all algorithms, the \underline{b}ottom-\underline{u}p app\underline{r}oach (\adpm) algorithm has the best cover size among all the algorithms. The \underline{T}op-\underline{D}own \underline{B}locks (\tdb) algorithm has the fastest runtime and comparable cover size by reducing the worst time complexity from $O(n^k)$ to $O(k \cdot n \cdot m)$ for \cvc.} 
Our comprehensive experiments show the effectiveness and efficiency of our proposed methods in terms of cover size and runtime when compared to the state-of-the-art $k$-cycle traversal algorithm \darc.
%The \tpd outperforms it by an average of $2$ to $3$ orders of magnitude while maintaining a comparable cover size. 

\section*{Acknowledgment}
Wenjie Zhang is supported by ARC Future Fellowship FT210100303. Lu Qin is supported by ARC FT200100787 and DP210101347.

\newpage
\bibliographystyle{ieeetr}
{
\small
\bibliography{ref}
}
\end{document}